\documentclass[]{llncs}
\usepackage[utf8]{inputenc}
\usepackage{float}
\floatstyle{ruled}
\newfloat{algorithm}{tbp}{loa}
\floatname{algorithm}{Protocol}
\usepackage{theorem}
\usepackage{latexsym}
\usepackage{amsmath}
\usepackage{graphicx}
\usepackage{amssymb}
\usepackage{array}
\usepackage{setspace}
\usepackage{wrapfig}
\usepackage{algpseudocode}
\usepackage{mathtools}
\usepackage{enumerate}
\usepackage{color}

\newcommand{\ie}{i.e., }

\title{Asynchronous Gathering in a Torus}
\author{Sayaka Kamei\inst{1} \and Anissa Lamani\inst{2} \and Fukuhito Ooshita\inst{3} \and S\'ebastien Tixeuil \inst{4} \and Koichi Wada \inst{5}}

\institute{
              Graduate School of Advanced Science and Engineering, Hiroshima University, Japan
           \and
           University of Strasbourg, ICube, France
              \and
        Graduate School of Science and Technology, Nara Institute of Science and Technology, Japan. \and
        Sorbonne University, CNRS, LIP6, France. \and
        Faculty of Science and Engineering, Hosei University, Japan.
        }

\begin{document}

\maketitle

\begin{abstract}\label{sec:abstract}

We consider the gathering problem for asynchronous and oblivious robots that cannot communicate explicitly with each other, but are endowed with visibility sensors that allow them to see the positions of the other robots.
Most of the investigations on the gathering problem on the discrete universe are done on ring shaped networks due to the number of symmetric configurations. We extend in this paper the study of the gathering problem on torus shaped networks assuming robots endowed with local weak multiplicity detection. That is, robots cannot make the difference between nodes occupied by only one robot from those occupied by more than one robots unless it is their current node. 
As a consequence, solutions based on creating a single multiplicity node as a landmark for the gathering cannot be used. 
We present in this paper a deterministic algorithm that solves the gathering problem starting from any rigid configuration on an asymmetric unoriented torus shaped network.  

\end{abstract}

\section{Introduction}\label{sec:intro}

We consider autonomous robots~\cite{SY99j} that are endowed with visibility 
sensors and motion actuators, yet are unable to communicate explicitly.
They evolve in a discrete environment, \emph{i.e.}, their space is
partitioned into a finite number of locations, conveniently
represented by a graph, where the nodes represent the possible
locations that a robot can be, and the edges denote the possibility for a
robot to move from one location to another.

Those robots must collaborate to solve a collective task despite being
limited with respect to computing capabilities, inputs from the environment, etc. 
In particular, the robots we consider are anonymous, uniform, yet they can 
sense their environment and take decisions according to their own ego-centered 
view. In addition, they are oblivious, \emph{i.e.}, they do not remember their past actions.
Robots operate in \emph{cycles} that include three phases:
\emph{Look}, \emph{Compute}, and \emph{Move} (LCM for short).  
The Look phase consists in taking a snapshot of the other robots positions
using a robot's visibility sensors. During the Compute phase, a robot
computes a target destination based on its previous observation. The
Move phase simply consists in moving toward the computed destination
using motion actuators. Using LCM cycles, three execution models have been
considered in the literature, capturing the various degrees of
synchrony between robots. According to current
taxonomy~\cite{FPS19b}, they are denoted FSYNC, SSYNC, and ASYNC, from
the stronger to the weaker. FSYNC stands for \emph{fully
  synchronous}.  In this model, all robots execute the LCM cycle
synchronously and atomically.  In the SSYNC (\emph{semi-synchronous})
model, robots are asynchronously activated to perform cycles, yet at
each activation, a robot executes one cycle atomically. With the
weaker model, ASYNC (\emph{asynchronous}), robots execute
LCM in a completely independent manner. Of course, the ASYNC model is 
the most realistic.

In the context of robots evolving on graphs, the two benchmarking tasks are \emph{exploration}~\cite{I19bc} and \emph{gathering}~\cite{CSN19bc}.
In this paper, we address the \emph{gathering} problem, which requires 
that robots eventually all meet at a single node, not known beforehand, and terminate upon completion.

We focus on the case where the network is an \emph{anonymous unoriented
  torus} (or simply \emph{torus}, for short).  The terms \emph{anonymous} 
  and \emph{unoriented} mean that no robot has access to any
kind external information (\emph{e.g.}, node identifiers, oracle, local edge
labeling, etc.) allowing to identify nodes or to determine any (global or
local) direction, such as North-South/East-West. Torus networks were previously investigated for the purpose of exploration by Devismes et al.\cite{DLPT2019j}.

\paragraph{Related Works.}

Mobile robot gathering on graphs was first considered for ring-shaped graphs. 
Klasing \emph{et al.}~\cite{KMP08j,KKN10j}, who proposed gathering algorithms for rings with \emph{global-weak} multiplicity detection. Global-weak multiplicity detection enables a robot to detect whether the number of robots on each node is one, or more than one. However, the exact number of robots on a given node remains unknown if there is more than one robot on the node.
Then, Izumi \emph{et al.}~\cite{IIKO13j} provided a gathering algorithm for rings with \emph{local-weak} multiplicity detection under the assumption that the initial configurations are non-symmetric and non-periodic, and that the number of robots is less than half the number of nodes. Local-weak multiplicity detection enables a robot to detect whether the number of robots on its \emph{current} node is one, or more than one. This condition was slightly relaxed by Kamei et al.~\cite{KLOT12c}.
D'Angelo \emph{et al.}~\cite{DNN17j} proposed unified ring gathering algorithms for most of the solvable initial configurations, using local-weak multiplicity detection. 
Overall, for rings, relatively few open cases remain~\cite{BPT16c}, as algorithm synthesis was demonstrated feasible~\cite{MPST14c}. 

The case of gathering in tree-shaped networks was investigated by D'Angelo \emph{et al.}~\cite{DDKN16j} and by Di Stefano et al.\cite{DN17j}. Hypercubes were the focus of Bose at el.~\cite{BKAS18c}. Complete and complete bipartite graphs were outlined by Cicerone et al.~\cite{CDN21j}, and regular bipartite by Guilbault et al.~\cite{GP13j}.
Finite grids were studied by D'Angelo et al.~\cite{DDKN16j}, Das et al.~\cite{DGLM19c}, and Castenow et al.~\cite{CFHJM20j}, while infinite grids were considered by Di Stefano et al.~\cite{DN17j}, and by Durjoy et al.~\cite{DDG17c}. Results on grids and infinite grids do not naturally extend to tori. On the one hand, the proof arguments for impossibility results on the grid can be extended for the torus, since their indistinguishability criterium remains valid. So, if a torus admits an edge symmetry (the robot positions are mirrored over an axial symmetry traversing an edge), or is periodic (a non-trivial translation leaves the robot positions unchanged), gathering is impossible on a torus. On the other hand, both the finite and the infinite grid allow algorithmic tricks to be implemented. For example, the finite grid has three classes of nodes: corners (of degree $2$), borders (of degree $3$), and inner nodes (of degree $4$), and those three classes permit the robots to obtain some sense of direction. By contrast, the infinite grid makes a difference between two locations: the inner space (the set of nodes within the convex hull formed by the robot positions) and the outer space (the rest of the infinite grid), which also give some sense of direction. Now, every node in a torus has degree $4$, and no notion of inner/outer space can be defined.  
To our knowledge, torus-shaped networks were never considered before for the gathering problem. The aforementioned work by Devismes et al~\cite{DLPT2019j} only considers the exploration task.

\paragraph{Our contribution.}

We consider the problem of gathering on torus-shaped networks. In more details, for initial configurations that are \emph{rigid} (i.e. neither symmetric nor periodic), we propose a distributed algorithm that gather all robots to a single node, not known beforehand. We only make use of local-weak multiplicity detection: robots may only know whether at least one other robot is currently hosted at their hosting node, but cannot know the exact number, and are also unable to retried multiplicity information from other nodes. Furthermore, robots have no common notion of North, and no common notion of handedness. Finally, robots operate in the most general and realistic ASYNC execution model. 

\section{Model}\label{sec:model}

In this paper, we consider a distributed system that consists of a collection of $\mathcal{K} \geq 3$ robots evolving on a non-oriented and anonymous $(\ell,L)$-torus (or simply torus for short) of $n$ nodes. Values $\ell$ and $L$ are two integers such that (definition borrowed from Devismes et al.~\cite{DLPT2019j}):
\begin{enumerate}
\item $n= \ell \times L$
\item Let $E$ be a finite set of edges.
There exists an ordering $v_1, \ldots, v_n$ of the nodes of
the torus such that $\forall i \in \{1,\ldots,n\}$:
\begin{itemize}
\item if $i+\ell \leq n$, then $\{i,(i+\ell)\} \in E$, else $\{i,(i+\ell)\bmod n\} \in E$.
\item if $i \bmod \ell \neq 0$, then $\{i,i+1\} \in E$, else
  $\{i,i-\ell+1\} \in E$.
\end{itemize}
\end{enumerate}

Given the previous ordering $v_1, \ldots, v_n$, for every $j \in
\{0,\ldots,L-1\}$, the sequence $v_{1+j\times \ell},$ $v_{2+j\times
  \ell},\ldots,$ $v_{\ell + j\times \ell}$ is called an \emph{$\ell$-ring}.
Similarly, for every $k \in \{1,\ldots,\ell\}$, $v_{k},$ $v_{k + \ell},$ $v_{k +
  2\times\ell},$ $\ldots, v_{k + (L-1) \times \ell}$ is called an
\emph{$L$-ring}. 
In the sequel, we use the term {\em ring}
to designate an $\ell$-ring or an $L$-ring.

On the torus operate $\mathcal{K}\geq 3$ identical robots, \ie they all execute the same algorithm using no local parameters and one cannot distinguish them using their appearance. In addition, they are oblivious, \ie they cannot remember the operations performed before. No direct communication is allowed between robots however, we assume that each robot is endowed with visibility sensors that allow him to see the position of the other robots on the torus. Robots operate in cycles that comprise three phases: \textit{Look}, \textit{Compute} and \textit{Move}. During the first phase (Look), each robot takes a snapshot to see the positions of the other robots on the torus. In the second phase (Compute), they decide to either stay idle or move. In the case they decide to move, a neighboring destination is computed. Finally, in the last phase (Move), they move to the computed destination (if any). 

At each instant $t$, a subset of robots is activated for the execution by an external entity known as the \textit{scheduler}. We assume that the scheduler is fair, \ie all robots must be activated infinitely many times. The model considered in this paper is the asynchronous model (ASYNC) also known as the \textit{CORDA} model. In this model, the time between Look, Compute and Move phases, is finite but not bounded. In our case, we add a constraint that is the move operation is instantaneous, \ie when a robot performs a look operation, it sees all the robots on nodes and never on edges. However, note that even under this constraint, each robot may move according to an outdated view, \ie the robot takes a snapshot to see the positions of the other robots, but when it decides to move, some other robots may have moved already. 

In this paper, we refer by $v_{i,j}$ to the $j^{th}$ node located on $\ell_i$. By $d_{i,j}(t)$ we denote the number of robots on node $v_{i,j}$ at time $t$. We say that $v_{i,j}$ is empty if $d_{i,j}(t)=0$. Otherwise, $v_{i,j}$ is said to be occupied. In the case where $d_{i,j}(t)=1$, we say that there is a single robot on $v_{i,j}$. By contrast, if $d_{i,j}(t)\geq 2$, we say that there is a multiplicity on $v_{i,j}$.  
In this paper, we assume that robots have a local weak multiplicity detection, that is, for any robot $r$, located at node $u$, $r$ can only detect a multiplicity on its current node $u$ (local). Moreover, $r$ cannot be aware of the exact number of robots part of the multiplicity (weak). 

During the process, some robots move and occupy at any time some nodes of the torus, their positions form the configuration of the system at that time. At instant $t=0$, we assume that each node is occupied by at most one robot, \ie the initial configuration contains no multiplicities. 

In the following, we assume that for any occupied node $v_{i,j}$, independently of the number of robots on $v_{i,j}$, $d_{i,j}(t)=1$.  For any $i, j \geq 0$, let $\delta^{+}_{i,j}(t)$ denote the sequence $<d_{i,j}(t), d_{i,j+1}(t), \dots, d_{i,j+\ell-1}(t)>$, and let $\delta^{-}_{i,j}(t)$ denote the  sequence $<d_{i,j}(t), d_{i,j-1}(t), \dots, d_{i,j-(\ell-1)}(t)>$. Similarly, let $\Delta^{+s}_{i,j}(t)$ be the sequence $<\delta^{s}_{i,j}(t), \delta^{s}_{i+1,j}(t), \dots, \delta^{s}_{i+(L-1),j}(t)>$ and $\Delta^{-s}_{i,j}(t)$ to be the sequence $<\delta^{s}_{i,j}(t), \delta^{s}_{i-1,j}(t), \dots, \delta^{s}_{i-(L-1),j}(t)>$ with $s \in \{+,-\}$.

The view of a given robot $r$ located on node $v_{i,j}$ at time $t$ is defined as the pair $view_r(t)=(\mathcal{V}_{i,j}(t), m_j)$ where $\mathcal{V}_{i,j}(t)$ consists of the four sequences $\Delta^{++}_{i,j}, \Delta^{+-}_{i,j}, \Delta^{-+}_{i,j}, \Delta^{--}_{i,j}$ ordered in the lexicographical order and $m_j=1$ if $v_j$ hosts a multiplicity and $m_j=0$ otherwise.

By $view_r(t)(1)$, we refer to $\mathcal{V}_{i,j}(t)$ in $view_r(t)$. Let $r$ and $r'$ be two robots satisfying $view_r(t)(1) > view_{r'}(t)(1)$. Robot $r$ is said in this case to have a larger view than $r'$. Similarly, $r$ is said to have the largest view at time $t$, if for any robots $r' \ne r$, not located on the same node as $r$, $view_r(t)(1) > view_{r'}(t)(1)$ holds.


\begin{figure}[!t]
\begin{center}
\includegraphics[scale=0.46]{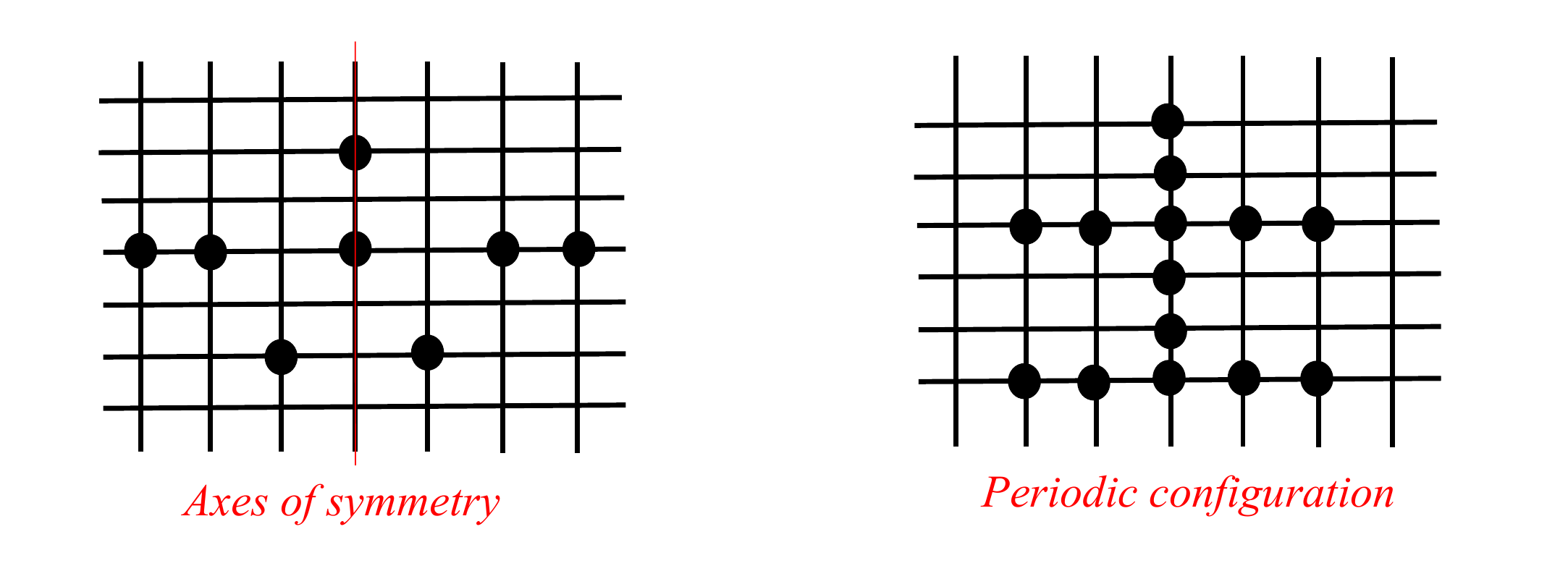}
\caption{Instance of a symmetric configuration and a periodic configuration}\label{fig:symPeriodic}
\end{center}
\end{figure}

A configuration is said to be rigid at time $t$, if for any two robots $r$ and $r'$, located on two different nodes of the torus, $view_{r}(t)(1) \ne view_{r'}(t)(1)$ holds.

 A configuration is said to be periodic at time $t$ if there exist two integers $i$ and $j$ such that $i\ne j$, $i\ne 0 \mod \ell$, $j\ne 0 \mod L$, and for every robot $r_{(x,w)}$ located on $\ell_x$ at node $v_{x,w}$, $view_{r_{(x,w)}}(t)(1) = view_{r_{(x+i,w+j)}}(t)(1)$ (An example is given in Fig.~\ref{fig:symPeriodic}).

As defined by D'Angelo et al.~\cite{DDKN16j}, a configuration is said to be symmetric at time $t$, if the configuration is invariant after a reflexion with respect to either a vertical or a horizontal axis. This axis is called axis of symmetry (An example is given in Fig.~\ref{fig:symPeriodic}).

In this paper, we consider asymmetric $({\ell},L)$-torus, \ie $\ell \ne L$. We assume without loss of generality that $L<\ell$. In this case, we can differentiate two sides of the torus. 
We denote by $nb_{\ell_i}(C)$ the number of occupied nodes on $\ell$-ring $\ell_i$, in configuration $C$. An $\ell$-ring $\ell_i$ is said to be maximal in $C$ if $\forall~ j \in \{1, \dots, \ell\}\setminus\{i\}$, $nb_{\ell_j}(C) \leq nb_{\ell_i}(C)$.\\

Given a configuration $C$ and two $\ell$-rings $\ell_i$ and $\ell_j$. We say that $\ell_j$ is adjacent to $\ell_i$ if $|i-j|=1 \mod L$ holds. Similarly, we say that $\ell_j$ is neighbor of $\ell_i$ in configuration $C$ if $nb_{\ell_j}(C)>0$ and $nb_{\ell_k}(C)=0$ for any
$k \in \{i+1, i+2, \dots, j-1\}$ or $k \in \{i-1, i-2, \dots, j+1\}$. For instance, in Figure~\ref{fig:neighbor}, $\ell_1$ and $\ell_2$ are adjacent while $\ell_2$ and $\ell_3$ are neighbors. 

We also define $dis(x_i,x_j)$ to be a function which returns the shortest distance, in terms of hops, between $x_i$ and $x_j$ where $x_i$ and $x_j$ are two nodes of the torus. We sometimes write $x_i=r_i$ where $r_i$ is a robot. In this case, $x_i$ refers to the node that hosts $r_i$. Finally, we use the notion of $d$.block to refer to a sequence of consecutive nodes in which there are occupied nodes each $d$ hops (distance).

\begin{figure}[!t]
\begin{center}
\includegraphics[scale=0.45]{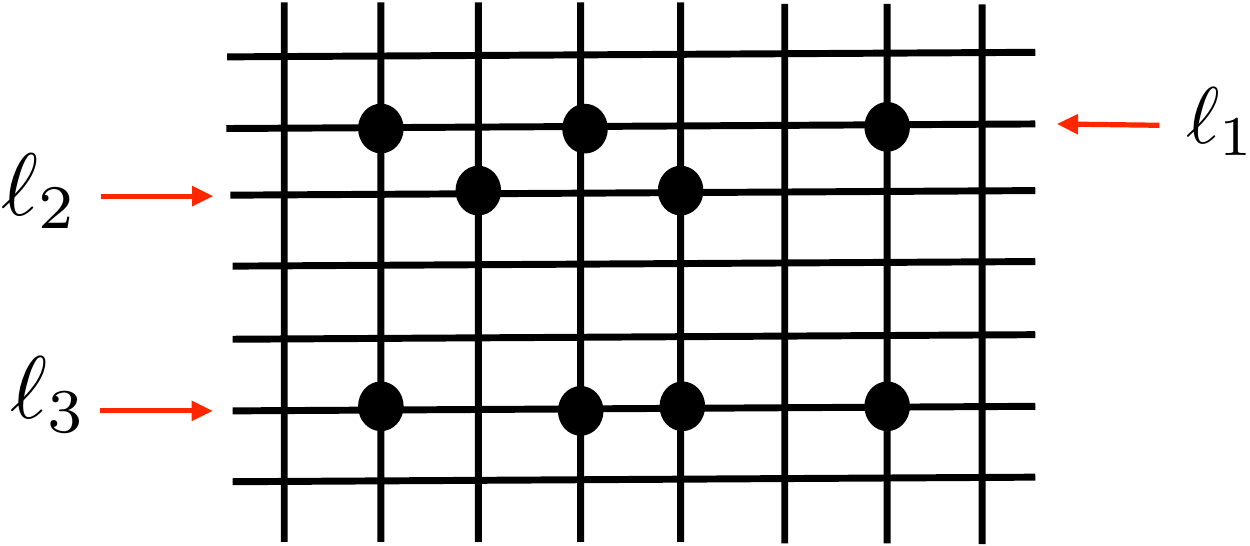}
\caption{Instance of two adjacent/ neighboring $\ell$-rings}\label{fig:neighbor}
\end{center}
\end{figure}

\section{Algorithm}


We describe in the following our strategy to solve the gathering problem is the predefined settings. Before explaining our algorithm in details, let us first define an important set configurations.\\


A configuration $C$ is called $C_{\mathit{target}}$ if there are three $\ell$-rings $\ell_{\mathit{max}}$, $\ell_{\mathit{secondary}}$ and $\ell_{\mathit{target}}$ satisfying following properties:

\begin{enumerate}
\item $\ell_{\mathit{max}}$ is the unique maximal $\ell$-ring in $C$. 

\item $\ell_{\mathit{secondary}}$ and $\ell_{\mathit{target}}$ are adjacent to $\ell_{\mathit{max}}$. 

\item $nb_{\ell_{\mathit{secondary}}}(C)=0$. 

\item $\ell_{\mathit{target}}$ satisfies exactly one of the following conditions:
\begin{enumerate}
\item $nb_{\ell_{\mathit{target}}}(C)=1$. We refer to the occupied node on $\ell_{\mathit{target}}$ by $v_{\mathit{target}}$.
\item $nb_{\ell_{\mathit{target}}}(C)=2$. Let us refer to the robots on $\ell_{\mathit{target}}$ by $r_{1}$ and $r_{2}$ respectively. Then, $r_1$ and $r_2$ are at distance $2$ from each other. We refer to the node that has two adjacent occupied nodes on $\ell_{\mathit{target}}$ by $v_{\mathit{target}}$.
\item $nb_{\ell_{\mathit{target}}}(C)=3$. In this case, there is three consecutive occupied nodes on $\ell_{\mathit{target}}$. By $v_{\mathit{target}}$, we refer to the unique node on $\ell_{\mathit{target}}$ that has two adjacent nodes on $\ell_{\mathit{target}}$. 
\end{enumerate}
\end{enumerate}


Some instances of $C_{\mathit{target}}$ configurations are presented in Figure~\ref{fig:conf-target}. We call $v_{\mathit{target}}$ \emph{target node}.

\begin{figure}[H]
\begin{center}
\includegraphics[scale=0.45]{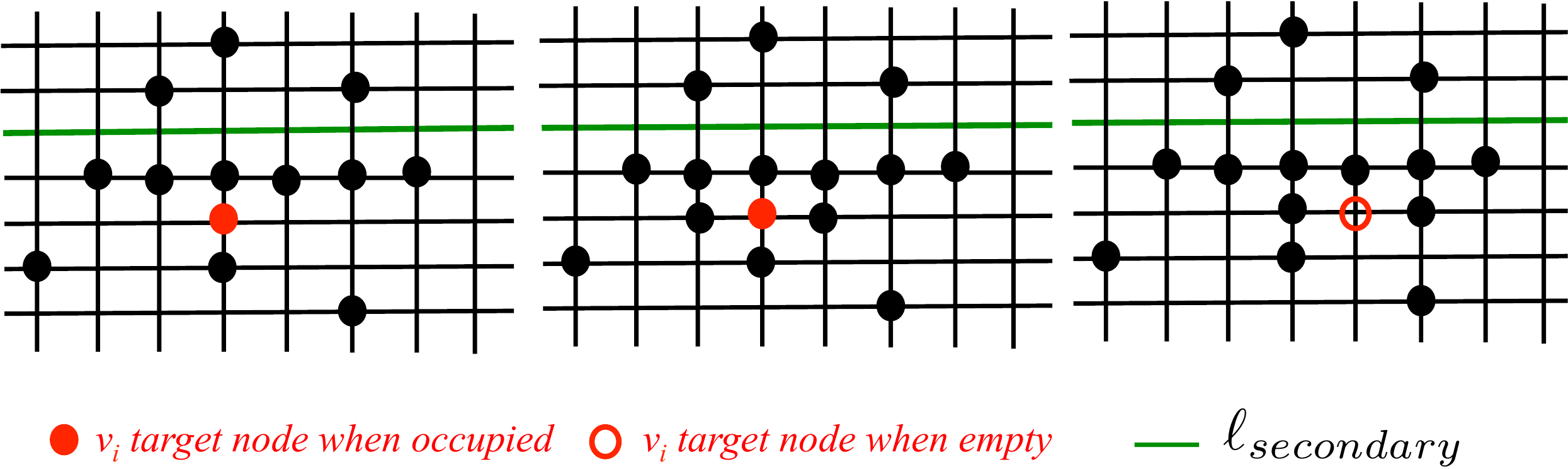}
\caption{Three instances of $C_{\mathit{target}}$ configurations}\label{fig:conf-target}
\end{center}
\end{figure}

Let $\mathcal{C}_{\mathit{target}}$ be the set of all $C_{\mathit{target}}$ configurations. Our algorithm consists of two phases as explained in the following:
  \begin{enumerate}
    \item {\bf Preparation Phase}. This phase starts from an arbitrary rigid configuration $C_0$ in which there is at most one robot on each node, \ie no node contains a multiplicity. The aim of this phase is to reach a configuration $C \in \mathcal{C}_{\mathit{target}}$.
    \item {\bf Gathering Phase}. Starting from a configuration $C \in \mathcal{C}_{\mathit{target}}$ configuration, robots perform the gathering task in such a way that at the end of this phase, all robots are, and remain, on the same node, \ie the gathering is achieved.
  \end{enumerate}

Let us refer by $\mathcal{C}_{p_{1}}$ (respectively $\mathcal{C}_{p_{2}}$) to the set of configurations that appear during the Preparation (respectively the Gathering) phase. Let $C$ be the current configuration, robots execute Protocol~\ref{alg:algo}. 

\begin{algorithm}[H]
\caption{Main protocol}\label{alg:algo}
\begin{algorithmic}
\If {$C \in \mathcal{C}_{p_{2}}$}
\State Execute \textbf{Gathering phase}
\Else
\State Execute \textbf{Preparation phase}
\EndIf
\end{algorithmic}
\end{algorithm}

Observe that $\mathcal{C}_{p_1} \cap \mathcal{C}_{p_2}= \emptyset$ and $\mathcal{C}_{\mathit{target}} \subset \mathcal{C}_{p_2}$. \\

We also define the following predicates on a given configuration $C$:

\begin{itemize}
\item \textbf{Unique($C$)}: There exists a unique $i \in \{1,2, \dots, L\}$ such that $\forall ~j \in \{1, \dots, L\}\setminus\{i\}$, $nb_{\ell_j}(C) < nb_{\ell_i}(C)$. 
\item \textbf{Empty($C$)}: $(C \in \mathcal{C}_{\mathit{target}})~ \wedge~ (\forall~i\in\{1, \dots, L\}$, such that $\ell_i \ne \ell_{\mathit{target}}$ and $\ell_i \ne \ell_{\mathit{max}}$, $nb_{\ell_i}(C)=0$).
\item \textbf{Partial($C$)}: $(C \in \mathcal{C}_{\mathit{target}})~ \wedge~(\exists~i\in\{1, \dots, L\}$, such that $\ell_i \ne \ell_{\mathit{target}}$ and $\ell_i \ne \ell_{\mathit{max}}$, $nb_{\ell_i}(C)\neq 0$).
\end{itemize}

Given a configuration $C$, Unique($C$) indicates that $C$ contains a unique maximal $\ell$-ring $\ell_{\mathit{max}}$.  Empty($C$) indicates that all the $\ell$-rings in $C$, except for $\ell_{\mathit{max}}$ and $\ell_{\mathit{target}}$, are empty. By contrast, Partial($C$) indicates that there is at least one $\ell$-ring besides $\ell_{\mathit{max}}$ and $\ell_{\mathit{target}}$ that is occupied (has at least one occupied node).


\subsection{Procedure Align}

Let us present a procedure referred to by \textbf{Align}($\ell_i$, $\ell_k$) which is called by our algorithm to align robots on $\ell_i$ with respect to robots positions on $\ell_k$. The procedure is only called when the following properties hold on both $\ell_i$ and $\ell_k$:
    \begin{enumerate}
        \item $nb_{\ell_i}(C)=j$ with $j \in \{2,\ldots,5\}$, \ie there are at least two and at most five robots on $\ell_i$. 
        \item $nb_{\ell_i}(C)>nb_{\ell_k}(C)$ holds, and either (1) $nb_{\ell_k}(C)=1$ or (2) $nb_{\ell_k}(C)=2$ and $\ell_k$ contains a $2$.block or (3) $nb_{\ell_k}(C)=3$ and $\ell_k$ contains a $1$.block of size $3$. Let $u_{mark}$ be the node on $\ell_k$ that is 
              \begin{itemize}
                 \item occupied if $nb_{\ell_k}(C)=1$.
                 \item empty in the 2.block if $nb_{\ell_k}(C)=2$.
                 \item occupied in the middle of the $1$.block if $nb_{\ell_k}(C)=3$.
              \end{itemize}
    \end{enumerate}
  
Let $u_1$, $u_2$, $u_3$, $u_4$ and $u_5$ be five consecutive nodes on $\ell_i$ such that $u_3$ is on the same $L$-ring as $u_{mark}$. This notation is used for explanation purposes only, recall that the nodes are anonymous.  The purpose of procedure \textbf{Align}($\ell_i$, $\ell_k$) is to align robots on $\ell_i$ with respect to the robots on $\ell_k$. Depending on the number of robots on $\ell_i$ and $\ell_k$, the following cases are possible: 
  \begin{enumerate}[a)]
      \item $nb_{\ell_i}(C)=2$. In this case, \textbf{Align}($\ell_i$,$\ell_k$) is only called when $nb_{\ell_k}(C)=1$. The aim is to create a $2$.block on $\ell_i$ whose middle node is on the same $L$-ring as $u_{mark}$ (refer to Figure \ref{fig:Align2r}).  Let $r_1$ and $r_2$ be the two robots on $\ell_i$. The aim is to make the robots on $\ell_i$ move to reach a configuration in which both $u_2$ and $u_4$ are occupied. Observe that in the desired configuration, robots $r_1$ and $r_2$ form a $2$.block whose unique middle node is $u_3$. 
      \begin{itemize}
      \item If $u_3$ is occupied and without loss of generality it hosts $r_1$. Robot $r_1$ is the one allowed to move. If both $u_2$ and $u_4$ are empty and without loss of generality $dist(r_2, u_2) < dist(r_2, u_4)$ then $r_1$ moves to $u_4$. If $dist(r_2, u_2) = dist(r_2, u_4)$, then $r_1$ moves to either $u_2$ or $u_4$ (the adversary chooses a node to which $r_1$ moves to). 
      \item If $u_3$ is empty then assume without loss of generality that the path on $\ell_i$ between $r_{1}$ (respectively $r_2$) and $u_2$ (respectively $u_4$) is empty. If $u_2$ (respectively $u_4$) is an empty node then $r_1$ (respectively $r_2$) moves on its adjacent empty node on the empty path toward $u_2$ (respectively $u_4$).
      \end{itemize}
      \begin{figure}[H]
        \begin{center}
            \includegraphics[scale=0.55]{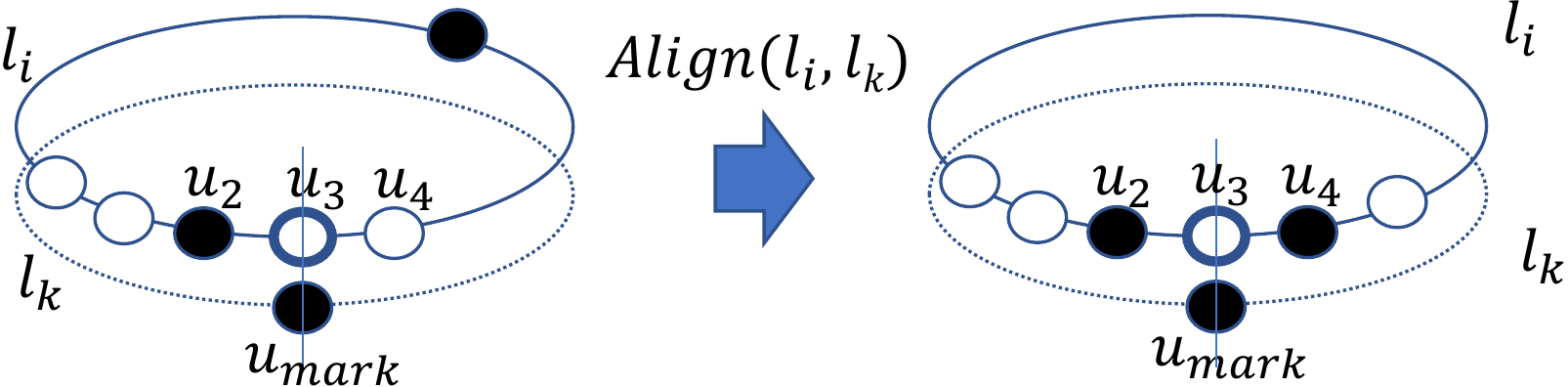}
            \caption{Align($\ell_i, \ell_k$) when $nb_{\ell_i}(C)=2$}\label{fig:Align2r}
        \end{center}
    \end{figure}
       
      \item $nb_{\ell_i}(C)=3$. The aim of \textbf{Align}($\ell_i, \ell_k$) is to create a $1$.block of size $3$ whose middle occupied node is on $u_3$ (refer to Figure \ref{fig:Align3r}). To this end, the robots behave as follows: Let $r_1, r_2$ and $r_3$ be the three robots on $\ell_{i}$. 
      \begin{itemize}
          \item If $u_3$ is occupied (assume without loss of generality that $r_1$ is on $u_3$) then $r_2$ and $r_3$ are both allowed to move. Assume without loss of generality that there is an empty path on $\ell_i$ between $r_2$ and $u_2$ respectively $r_3$ and $u_4$. The destination of $r_2$ (resp. $r_3$) is its adjacent node on $\ell_i$ toward $u_2$ (resp. $u_4$). 
          \item If $u_3$ is empty, then the aim of the robots is to make $u_3$ occupied without creating a tower. To this end, we identify two special cases $c1$ and $c2$. 
          \begin{itemize} \item In Case~$c1$: $r_1$, $r_2$ and $r_3$ form a $1$.block and the two extremities of the $1$.block are equidistant to $u_3$ (Assume without loss of generality that $r_1$ and $r_3$ are at the ones at the borders of the $1$.block, refer to Figure~\ref{fig:Semi-equidistant}). Both $r_1$ and $r_3$ are allowed to move. Their respective destination is their adjacent node on $\ell_i$ outside the $1$.block they belong to. 
          \item In Case~$c2$: without loss of generality, $r_1$ and $r_2$ form a $1$.block while $r_3$ is at distance $2$ from $r_2$. Moreover, $dis(r_1,u_3)=dis(r_3,u_3)+1$  (refer to Figure \ref{fig:Semi-equidistant}). Robot $r_1$ is the only one allowed to move, its destination is its adjacent empty node. Observe that Case~$c2$ can be reached from Case~$c1$ when a unique robot moves outside the block it belongs to. Case $c2$ ensures that the second robot that was supposed to move, also moves.
          \end{itemize}
          Finally, if neither Case~$c1$ nor Case~$c2$ hold then, let $R_m$ be the set of robots that are the closest $u_3$. If $|R_m|=2$ then the third robot (not in $R_m$), say $r_2$, is used to break the symmetry, \ie the robot that is allowed to move is the one that is the closest to $r_2$. Its destination is its adjacent empty node on $\ell_i$ on the empty path toward $u_3$. If $r_2$ is equidistant from both robots in $R_m$ then $r_2$ first moves to one of its adjacent empty nodes on $\ell_i$ (the choice is made by the adversary), the symmetry is then broken. If $|R_m|=1$, then the unique robot on $R_m$ moves to its adjacent empty node on its current $\ell$-ring taking the shortest path to $u_3$. 
          \end{itemize}

         \begin{figure}[H]
        \begin{center}
            \includegraphics[scale=0.55]{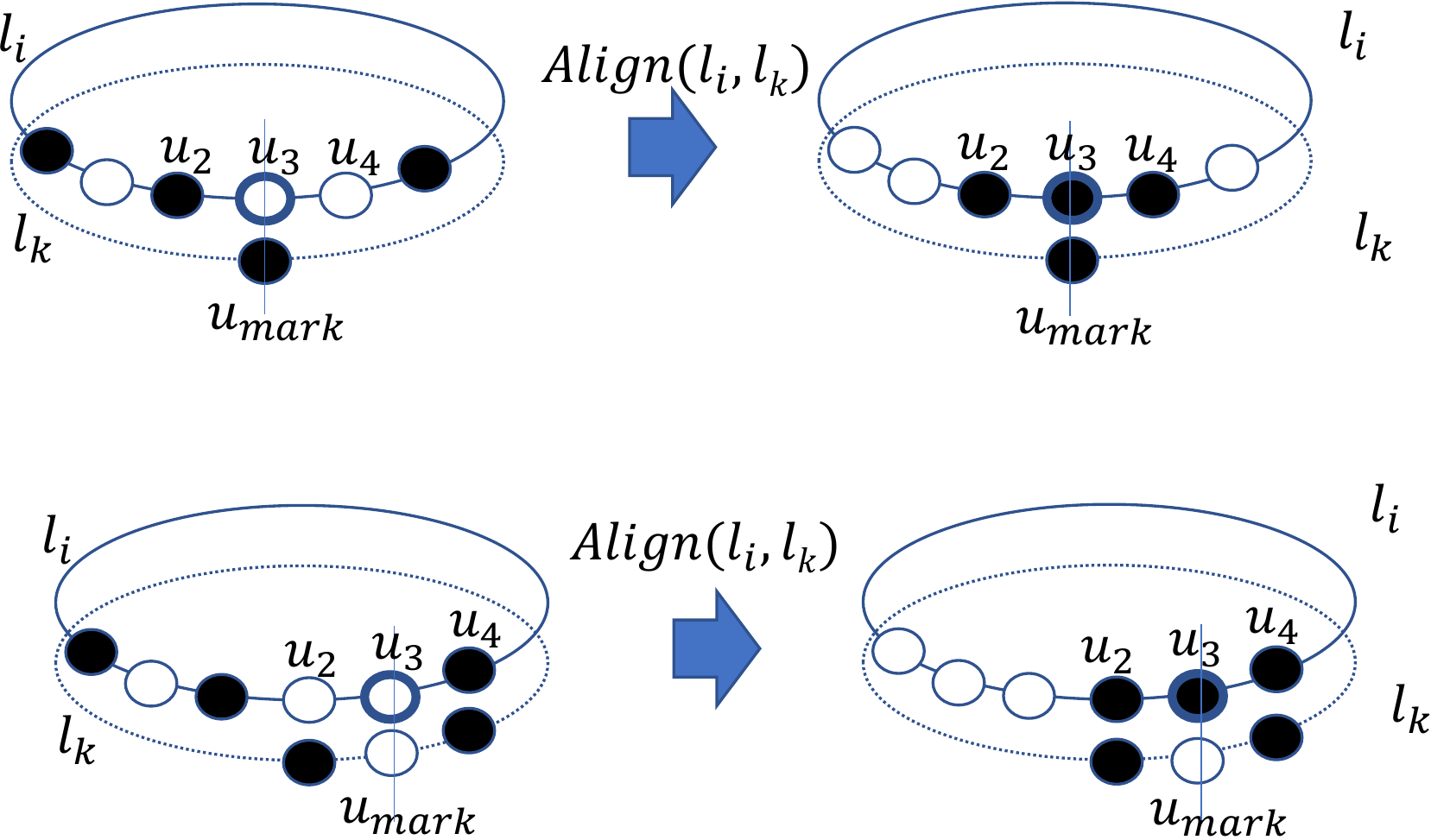}
            \caption{Align($\ell_i, \ell_k$) when $nb_{\ell_i}(C)=3$}\label{fig:Align3r}
        \end{center}
    \end{figure}
    
      \begin{figure}[htb]
        \begin{center}
            \includegraphics[scale=0.46]{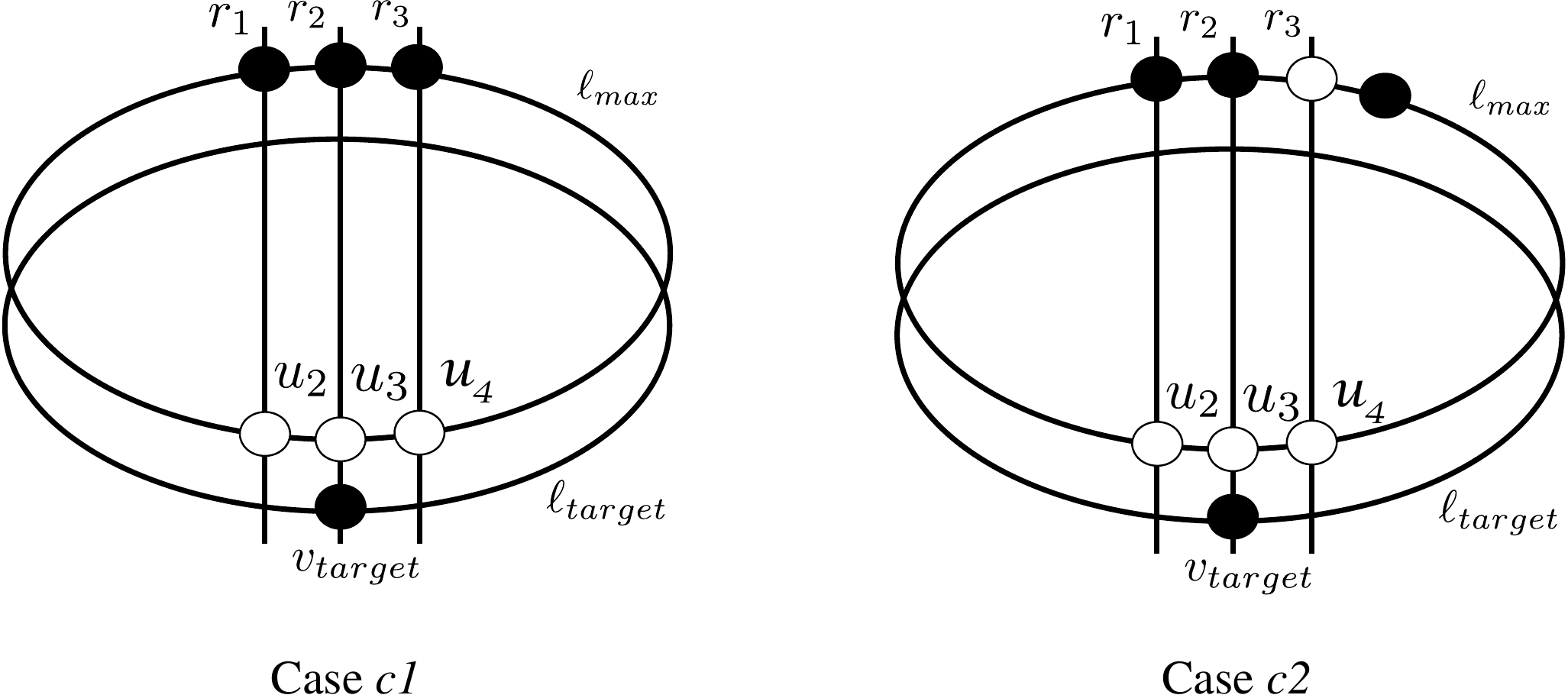}
            \caption{Special cases $c1$ and $c2$}\label{fig:Semi-equidistant}
        \end{center}
    \end{figure}

      \item $nb_{\ell_i}(C)=4$. The aim of \textbf{Align}($\ell_i, \ell_k$) is to create two $1$.blocks of size $2$ being at distance $2$ from each other in such a way that the unique empty node between the two 1.blocks is on the same $L$-ring as $u_{mark}$ (refer to Figure \ref{fig:Align4r}). Let $r_1$, $r_2$, $r_3$ and $r_4$ be the four robots on $\ell_i$, we distinguish the following two cases: 
             \begin{enumerate}
				\item Node $u_3$ is empty. Let $\rightarrow$ and $\leftarrow$ be two directions of the $\ell$-ring starting from~$u_3$. Let $r_1 \leq r_2 \leq r_3 \leq r_4$ be the ordering of robots according to their distance to $u_3$ with respect to a given direction $\rightarrow$, \ie $r_1$ is the closest to $u_3$ with respect to $\rightarrow$ while $r_4$ is the farthest one. Observe that the order of the robots according to their distance to $u_3$ with respect $\leftarrow$ is $r_4 \leq r_3 \leq r_2 \leq r_1$.  If $u_2$ (respectively $u_4$) is empty then $r_1$ (respectively $r_4$) moves toward $u_3$ with respect to the direction $\rightarrow$ (respectively $\leftarrow$). If $u_2$ is occupied while $u_1$ is empty (respectively $u_4$ is occupied while $u_5$ is empty) then $r_2$ (respectively $r_3$) moves toward $u_1$ (respectively $u_5$) with respect to the direction $\rightarrow$ (respectively $\leftarrow$). 
    			 \item Note $u_3$ is occupied. (i) If both $u_2$ and $u_4$ are empty then, the robots on $u_3$ moves to one of its adjacent nodes on $\ell_i$ (either $u_2$ or $u_4$, the choice is made by the scheduler). (ii) if without loss of generality, $u_2$ is empty while $u_4$ is occupied then the robot on $u_3$ moves to $u_2$. Finally, (iii) if both $u_2$ and $u_4$ are occupied then since $nb_{\ell_i}(C)=4$, we are sure that either $u_1$ or $u_5$ is empty. If without loss of generality, if only $u_1$ is empty then the robot on $u_2$ moves to $u_1$. If both $u_1$ and $u_5$ are empty  then assume that $r_1$ is the robot not part of the $1$.block of size $3$ on $\ell_i$. If $dis(u_1,r_1)=dis(u_5,r_1)$ then $r_1$ moves to one of its adjacent empty node on $\ell_i$ (the choice is made by the adversary). By contrast, if without loss of generality, $dis(u_1,r_1)< dis(u_5,r_1)$, then $r_1$ moves to its adjacent empty node toward $u_1$.
 
              \end{enumerate}
              
              \begin{figure}[H]
                \begin{center}
                        \includegraphics[scale=0.55]{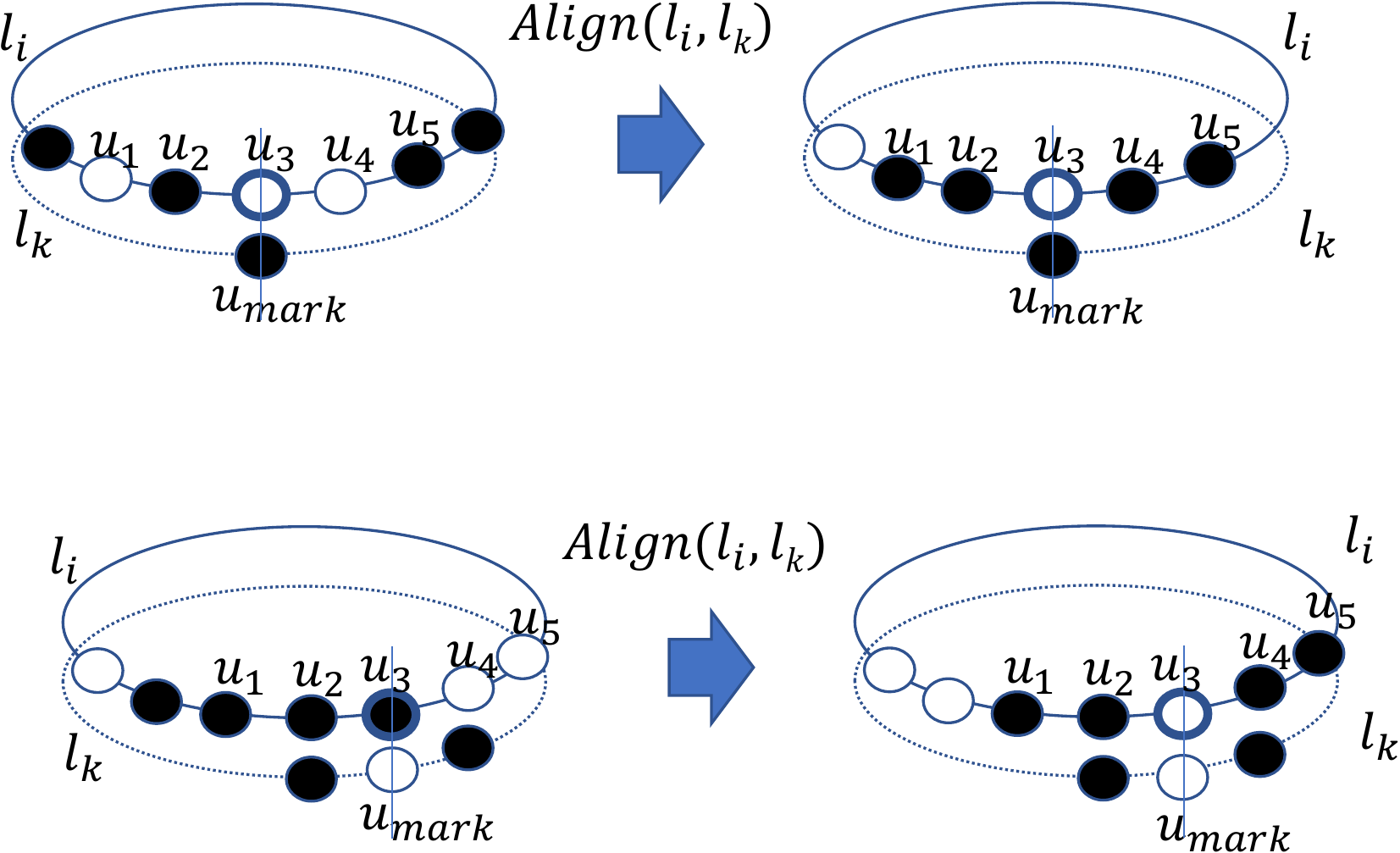}
                        \caption{Align($\ell_i, \ell_k$) when $nb_{\ell_i}(C)=4$}\label{fig:Align4r}
                 \end{center}
            \end{figure}
    
    \item $nb_{\ell_i}(C)=5$. The aim of \textbf{Align}($\ell_i, \ell_k$) is to create a $1$.block of size $5$ whose middle robot is on the same $L$-ring as $u_{mark}$.  Let $\rightarrow$ and $\leftarrow$ be two directions of the $\ell$-ring starting from~$u_3$. Let $r_1 \leq r_2 \leq r_3 \leq r_4 \leq r_5$ be the ordering of robots according to their distance to $u_3$ with respect to a given direction $\rightarrow$, \ie $r_1$ is the closest to $u_3$ with respect to $\rightarrow$ while $r_5$ is the farthest one. 
    \begin{itemize}
        \item If $u_3$ is occupied, then $u_3$ hosts $r_1$ by assumption. If $u_2$ (respectively $u_5$) is empty then $r_2$ (respectively $r_5$) is allowed to move. Its destination is its adjacent node on $\ell_i$ towards $u_2$ (respectively $u_4$) taking the empty path. If $u_2$ (respectively $u_4$) is occupied then robot $r_3$ (respectively $r_4$) is allowed to move. Its destination is its adjacent node on $\ell_i$ toward $u_2$ (respectively $u_4$) taking the empty path. 
        \item If $u_3$ is empty, then as for the case in which $nb_{\ell_i}(C)=3$, robots need to be careful not to create a tower. We first distinguish some special configurations that help us deal with robot with potentially outdated view. These cases are presented in Figures \ref{fig:Align5r1} and \ref{fig:Align5r2} along with the robots to move. If robots are not in any of these special cases then they proceed as follows: Let $R$ be the set of the robots on $\ell_i$ which are the closest to $u_3$. If $|R|=1$ then the unique robot in $R$ moves to its adjacent empty node toward $u_3$ taking the shortest path. By contrast, if $|R|=2$ then by assumption $r_1$ and $r_5$ are equidistant from $u_3$. To choose the one to move, the two robots compare dist($r_1$,$r_3$) and dist($r_5$,$r_3$). If dist($r_1$,$r_2$) $=$ dist($r_5$,$r_4$) then $r_3$ moves to one of its adjacent empty node on $\ell_i$ (Observe that the case in which there is no such empty node is handled by the special configurations identified in Figures \ref{fig:Align5r1} and \ref{fig:Align5r2}). Finally, if without loss of generality, $r_3$ is closer to $r_1$ than $r_5$, then $r_1$ is the one to move. 
    \end{itemize}
    
        \begin{figure}[H]
                \begin{center}
                        \includegraphics[scale=0.46]{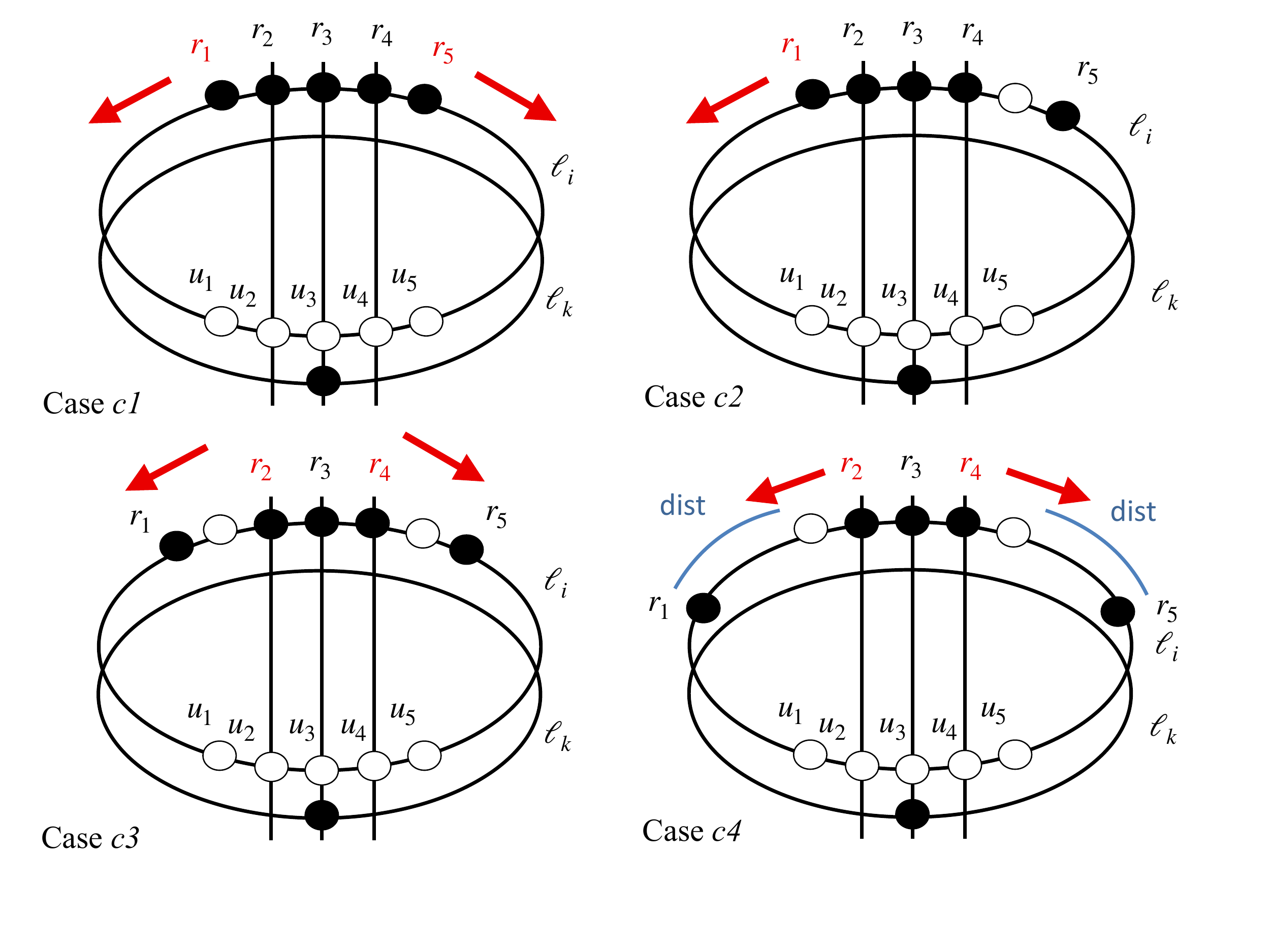}
                        \caption{Align($\ell_i, \ell_k$) when $nb_{\ell_i}(C)=5$}\label{fig:Align5r1}
                 \end{center}
            \end{figure}
            
              \begin{figure}[H]
                \begin{center}
                        \includegraphics[scale=0.46]{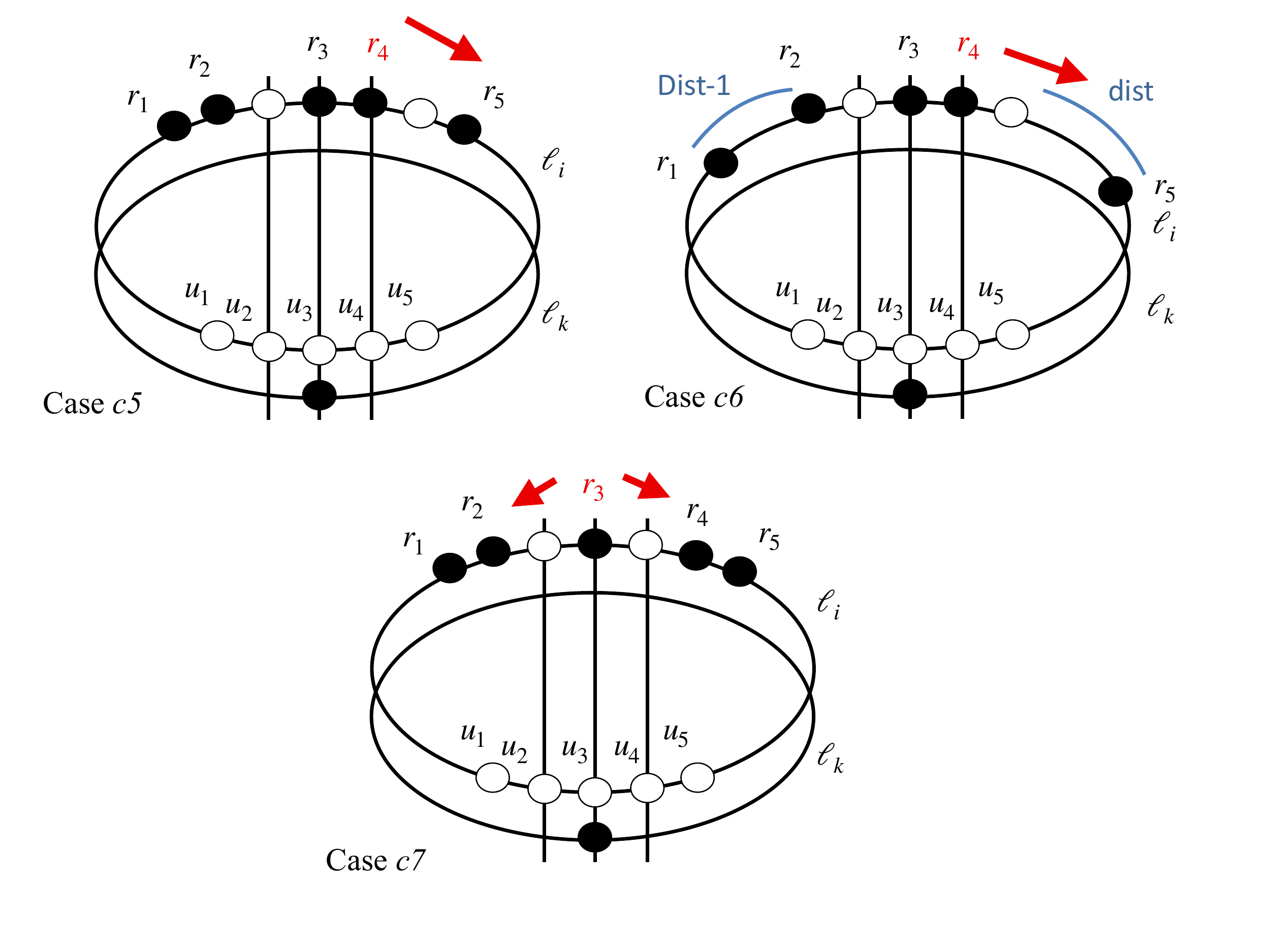}
                        \caption{Align($\ell_i, \ell_k$) when $nb_{\ell_i}(C)=5$}\label{fig:Align5r2}
                 \end{center}
            \end{figure}

  \end{enumerate}




\subsection{Preparation phase}\label{subsec:preparation}
Let $C \in \mathcal{C}_{p1}$. The main purpose of this phase is to reach a configuration $C' \in C_{\mathit{target}}$ from $C$. For this aim, robots first decrease the number of maximal $\ell$-rings to reach a configuration $C''$ in which Unique($C''$) is true. From $C''$, a configuration $C' \in C_{\mathit{target}}$ is then created. In the case in which Unique($C$) is true, we refer to the unique maximal $\ell$-ring in $C$ by $\ell_{\mathit{max}}$ and to the two adjacent $\ell$-rings of $\ell_{\mathit{max}}$ by respectively $\ell_k$ and $\ell_i$. To ease the description of our this phase, we distinguish five main sets of configurations when \textbf{Unique($C$) is true}: 

\begin{enumerate}
\item  Set $\mathcal{C}_{Empty}$: $C \in \mathcal{C}_{Empty}$ if $nb_{\ell_i}(C)= nb_{\ell_k}(C)=0$.
 
\item Set $\mathcal{C}_{Semi-Empty}$: $C \in \mathcal{C}_{Semi-Empty}$ if without loss of generality $nb_{\ell_i}(C)= 0$ and $nb_{\ell_k}(C)>1$.

\item Set $\mathcal{C}_{Oriented}$: $C \in \mathcal{C}_{Oriented}$ if without loss of generality $nb_{\ell_i}(C)= 1$ and $nb_{\ell_k}(C)>1$. Set $\mathcal{C}_{Oriented}$ includes: 
          \begin{enumerate}
             \item $\mathcal{C}_{Oriented-1}$. In this case either (i) $nb_{\ell_k}(C)=3$ and $\ell_k$ contains a $1$.block of size $3$ whose middle robot is on the same $L$-ring as the unique occupied node on $\ell_i$. (ii) $nb_{\ell_k}(C)=2$ and $\ell_k$ contains a $2$.block. Moreover, the unique empty node in the $2$.block is on the same $L$-ring as the unique robot on $\ell_i$. 
             
             \item $\mathcal{C}_{Oriented-2}$. Contains all the configuration in $\mathcal{C}_{oriented}$ that are not in $\mathcal{C}_{Oriented-1}$. That is, $\mathcal{C}_{Oriented-2}= \mathcal{C}_{Oriented} - \mathcal{C}_{Oriented-1}$.
          \end{enumerate}

\item Set $\mathcal{C}_{Semi-Oriented}$: $C \in \mathcal{C}_{Semi-Oriented}$ if without loss of generality $nb_{\ell_i}(C)= 1$ and $nb_{\ell_k}(C)=1$. 

\item Set $\mathcal{C}_{Undefined}$: $C \in \mathcal{C}_{Undefined}$ if $nb_{\ell_i}(C)>1$ and $nb_{\ell_k}(C)>1$.
\end{enumerate}

\subsubsection*{Robots behavior} 

We describe in the following robots behavior during the preparation phase. Let $C$ be the current configuration. Recall that if Unique($C$) is false then robots first aim at decreasing the number of maximal $\ell$-rings to reach a configuration $C'$ in which $Unique(C')$ is true. From there, robots create a configuration $C'' \in \mathcal{C}_{\mathit{target}}$. That is, we distinguish the following two cases: 

\begin{enumerate}

    \item \textbf{Unique($C$) is false}. The idea of the algorithm is to reduce the number of maximal $\ell$-rings while keeping the configuration rigid. We distinguish two cases: 
        \begin{enumerate}
            \item $nb_{\ell_{\mathit{max}}}(C)=\ell$. Using the rigidity of $C$, a single maximal $\ell$-ring is selected and a single robot on this $\ell$-ring is selected to move. Its destination is its adjacent occupied node on its current $\ell$-ring. The robot to move is the one that keeps the configuration rigid (the existence of such a robot is proven in Lemma \ref{lem:L--}).  
        
            \item $nb_{\ell_{\mathit{max}}}(C)<\ell$. As there is at least one empty node on each maximal $\ell$-ring, the idea is to fill exactly one of these nodes. Let $R$ be the set of robots that are the closest to an empty node on a maximal $\ell$-ring. Under some conditions, using the rigidity of $C$, one robot of $R$, say $r$, is elected to move. Its destination is its adjacent empty node toward the closest empty node on a maximal $\ell$-ring, say $u$, taking the shortest path. Among robots in the set $R$, the robot to move is the one that does not create a symmetric configuration. If no such robot exists in $R$ then, some extra steps are taken to make sure that the configuration remains rigid. We discuss the various cases in what follows : 
        \begin{itemize}
        \item Assume that $C$ contains exactly two occupied $\ell$-rings. This means that $C$ contains two maximal $\ell$-rings and $r$ belongs to a maximal $\ell$-ring. Using the rigidity of $C$, one robot is elected to move, its destination is its adjacent empty node on an empty $\ell$-ring. 
        \item If $C$ contains more than two occupied $\ell$-rings then the robots proceed as follows: let $r$ be the robot in $R$ with the largest view. By $u$ and target-$\ell$ we refer to respectively the closest empty node on a maximal $\ell$-ring to $r$ and the $\ell$-ring including $u$. If by moving, $r$ does not create a symmetric configuration then $r$ simply moves to its adjacent node toward $u$ taking the shortest path. By contrast, if by moving $r$ creates a symmetric configuration then let $C'$ be the configuration reached once $r$ moves. Using configuration $C'$ that each robot can compute without $r$ moving, a robot in $C$ is selected to move. 
            We show later on that a symmetric configuration can only be reached when $r$ joins an empty node on the same $L$-ring as $u$ for the first time or when it joins $u$. For the other cases, the configuration remains rigid (refer to Lemma \ref{lem:uniquelessl}, Claims 1 and 2). Hence, we only consider the following two cases: 
            \begin{enumerate}
            
                \item  Robot $r$ joins an empty node on the same $L$-ring as $u$ for the first time in $C'$. In this case, in $C$, the robot that is on target-$\ell$ being on the same $L$-ring as $r$ moves to $u$ (refer to Figure \ref{fig:joinL}).
            
                 \begin{figure}[htb]
                    \begin{center}
                        \includegraphics[scale=0.3]{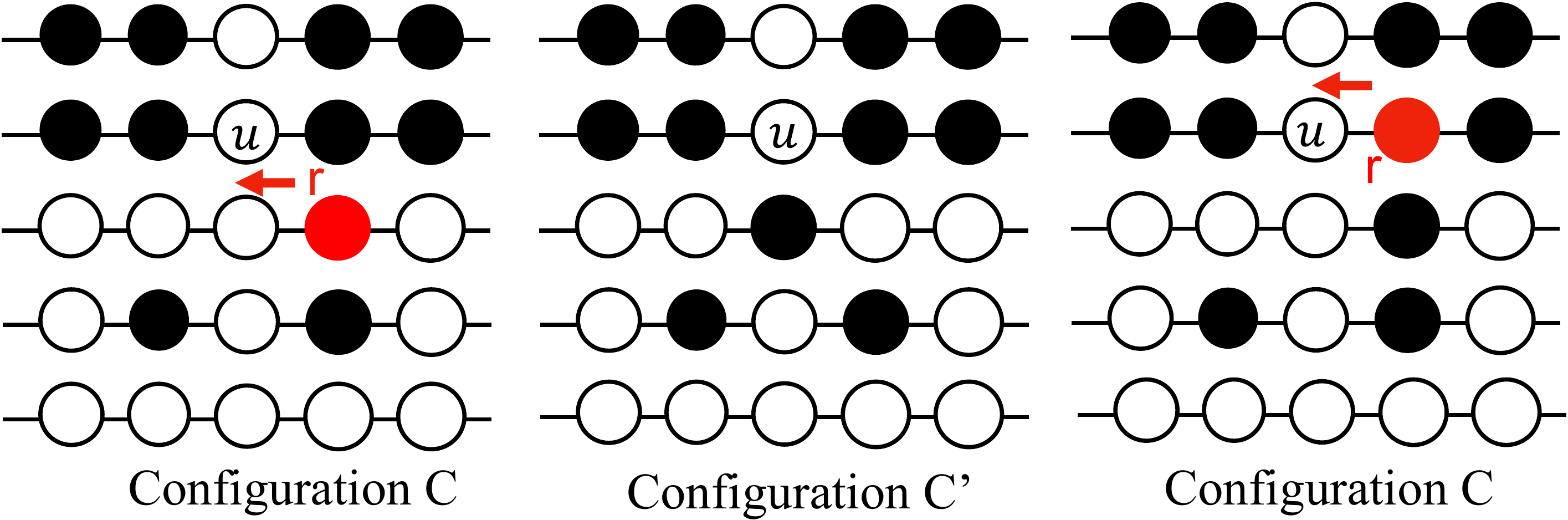}
                        \caption{On the left, $r$ is suppose to move but by moving, it creates a symmetric configuration $C'$ shown in the middle. The robot on target-$\ell$ on the same $L$-ring as $r$ moves to $u$.}\label{fig:joinL}
                    \end{center}
                \end{figure}

                \item Robot $r$ joins $u$ in $C'$. If in $C'$ there are only two occupied $\ell$-rings then using the rigidity of $C$, one robot from a maximal $\ell$-ring is elected to move. Its destination is its adjacent empty node on an empty $\ell$-ring. By contrast, if there are more than two occupied $\ell$-rings in $C'$ then robots proceed as follows: 
                \begin{itemize}
            \item If the axes of symmetry lies on the unique $\ell_{max}$ in $C'$ then we are sure that there are two $\ell$-rings which are maximal in $C$ and that are symmetric with respect to the unique maximal $\ell$-ring in $C'$. Using the rigidity of $C$, one robot from such an $\ell$-ring is allowed to move. Its destination is its adjacent empty node on its current $\ell$-ring. 
            \item If the axes of symmetry is perpendicular to the unique maximal $\ell$-ring in $C'$ then let $T$ be the set of occupied $\ell$-rings in $C$ without target-$\ell$. If there is an $\ell$-ring in $T$ which does not contain two $1$.blocks separated by a single empty node on each side then using the rigidity of $C$, a single robot on such an $\ell$-ring which is the closest to the biggest $1$.block is elected to move. Its destination is the closest $1$.block. If there no such $\ell$-ring in $T$ (all $\ell$-rings contains two $1$.blocks separated by a unique empty node, then using the rigidity of $C$, one robot being on an $\ell$-ring of $T$ who has an empty node as a neighbor on its $\ell$-ring is elected to move. Its destinations is its adjacent empty node on its current $\ell$-ring.
            \end{itemize}
            \end{enumerate}   
            Note that we only discussed the cases in which the reached configuration is either rigid and symmetric. We will show in the correctness proof that when $r$ moves, it cannot create a periodic configuration. This is mainly due to the fact that in $C'$ there is a unique maximal $\ell$-ring and $C$ is assumed to be rigid.
        \end{itemize}
        \end{enumerate}

    \item \textbf{Unique($C$) is true}. From $C$, robots aim to create a configuration $C' \in \mathcal{C}_{\mathit{target}}$. Let $\ell_{\mathit{max}}$ be the unique maximal $\ell$-ring and let $\ell_i$ and $\ell_k$ be the two adjacent $\ell$-rings of $\ell_{\mathit{max}}$. We use the set of configurations defined previously to describe robots behavior:
    
    \begin{enumerate}
         \item $C \in \mathcal{C}_{Empty}$. Let $\ell_{n_i}$ and $\ell_{n_k}$ be the two neighboring $\ell$-rings of $\ell_{\mathit{max}}$ (one neighboring $\ell$-ring from each direction of the torus). Observe that in the case where $\ell_{n_i}=\ell_{n_k}=\ell_{\mathit{max}}$, then $C$ contains a single $\ell$-ring that is occupied. Using the rigidity of~$C$, one robot from $C$ is selected to move to its adjacent empty node outside its $\ell$-{ring} (the scheduler chooses the direction to take: move toward $\ell_i$ or $\ell_k$). Otherwise, let $R_m$ be the set of robots which are the closest to either $\ell_{i}$ or $\ell_{k}$. If $|R_m|=1$ then, the unique robot in $R_m$, referred to by $r$, is the one allowed to move. Assume without loss of generality that $r$ is the closest to $\ell_{i}$. The destination of $r$ is its adjacent empty node outside its current $\ell$-ring on the shortest empty path toward $\ell_i$. If $r$ is the closest to both $\ell_{i}$ and $\ell_k$ then the scheduler chooses the direction to take (it moves either toward $\ell_i$ or $\ell_{k}$). In the case where $|R_m|>1$ ($R_m$ contains more than one robot) then, by using the rigidity of $C$, one robot $r$ is selected to move. Its behavior is the same as $r$ in the case where $|R_m|=1$.
     
        \item $C \in \mathcal{C}_{Semi-Empty}$. Without loss of generality $nb_{\ell_k}(C)>1$ and $nb_{\ell_i}(C)=0$. We distinguish two cases as follows:
        \begin{enumerate} 
              \item $nb_{\ell_k}(C)> 3$ or $nb_{\ell_k}(C)=2$. Recall that $C \not \in \mathcal{C}_{\mathit{target}}$. Let "$\uparrow$" be the direction defined from $\ell_{\mathit{max}}$ to $\ell_{k}$ taking the shortest path and let $\ell_n$ be the $\ell$-ring that is neighbor of $\ell_{i}$. Observe that $\ell_{n}=\ell_k$ is possible (in the case where only two $\ell$-rings are occupied in $C$). Using the rigidity of configuration $C$, one robot from $\ell_n$ is elected. This robot is the one allowed to move, its destination is its adjacent node outside $\ell_n$ and towards $\ell_i$ with respect to the direction $\uparrow$. 

            \item  $nb_{\ell_k}(C)=3$. Again, recall that $C \not \in \mathcal{C}_{\mathit{target}}$. The aim is to make the three robots form a single $1$.block. To this end, if the configuration contains a single $d$.block of size~$3$ with $d > 1$ then the robot in the middle of the $d$.block moves to its adjacent node on $\ell_k$ (the scheduler chooses the direction to take). By contrast, if the configuration contains a single $d$.block of size $2$ ($d \geq 1$) then the robot not part of the $d$.block moves towards its adjacent empty node towards the $d$.block taking the shortest empty path.  
        \end{enumerate} 

    \item $C \in \mathcal{C}_{Oriented}$. Let $r_i$ be the single robot on $\ell_i$. Recall that two sub-sets of configurations are defined in this case:
    \begin{enumerate} 

        \item $C \in \mathcal{C}_{Oriented-1}$. 
           If $nb_{\mathit{max}}(C)>4$ then the unique robot on $\ell_i$  moves to its adjacent node on $\ell_{\mathit{max}}$. Otherwise, let $u$ be the node on $\ell_{\mathit{max}}$ which is adjacent to the unique robot on $\ell_i$.
           \begin{itemize}
           \item If $nb_{\mathit{max}}(C)=3$ and the robots form a $1$.block of size $3$ whose middle robot is adjacent to $u$ then the unique robot on $\ell_i$ moves to its adjacent node on $\ell_{\mathit{max}}$. Otherwise, robots on $\ell_{\mathit{max}}$ execute \textbf{Align}$(\ell_{\mathit{max}}, \ell_i)$. 
           \item If $nb_{\mathit{max}}(C)=4$ and $u$ is empty, then the unique robot on $\ell_i$  moves to $u$. Otherwise ($u$ is occupied), then let $r$ be the robot on $u$. 
           \begin{itemize}
           \item If $r$ has an adjacent empty node on $\ell_{\mathit{max}}$ then $r$ moves to one of its adjacent nodes (the scheduler chooses the node to move to in case of symmetry). 
           \item If $r$ does not have an adjacent empty node on $\ell_{\mathit{max}}$, then let $r'$ be the robot on $\ell_{\mathit{max}}$ which is adjacent to $r$ and which does not have a neighboring robot on $\ell_{\mathit{max}}$ at distance $\lfloor \ell / 2\rfloor $. Robot $r'$ is the one allowed to move. Its destination is its adjacent empty node on $\ell_{\mathit{max}}$. 
           \end{itemize}
               \end{itemize}
            
            \item $C \in \mathcal{C}_{Oriented-2}$. If $nb_{\ell_k}(C)=2$ or $nb_{\ell_k}(C)=3$ then \textbf{Align}($\ell_k$, $\ell_i$) is executed. Otherwise, if $nb_{\ell_k}(C)>3$ then, $nb_{\ell_{k}}(C)-2$ robots gather on the node $u_k$ located on $\ell_k$ and which is on the same $L$-ring as the unique occupied node on $\ell_i$. For this purpose, the robot on $\ell_k$ which is the closest to $u_k$ with the largest view is the one allowed to move. Its destination is its adjacent node on $\ell_k$ toward $u_k$.  

       \end{enumerate} 

\item $C \in \mathcal{C}_{Semi-Oriented}$. Let $\ell_{n_i}$ and $\ell_{n_k}$ be the  two neighboring $\ell$-rings of $\ell_{i}$ and $\ell_{k}$ respectively. 
\begin{itemize}
\item If without loss of generality, $\ell_{i}=\ell_{n_k}$ (and hence $\ell_{k}=\ell_{n_i}$). Then configuration $C$ contains only 3 occupied $\ell$-rings $\ell_i$, $\ell_{\mathit{max}}$ and  $\ell_{j}$. Using the rigidity of $C$, one robot from either $\ell_{n_i}$ or $\ell_{n_k}$ (not both) is selected to move. Its destination is its adjacent empty node outside its current $\ell$-ring in the opposite direction of $\ell_{\mathit{max}}$. 
\item In the case where  $\ell_{i}\ne\ell_{n_k}$ (and hence $\ell_{k}\ne\ell_{n_i}$) then, again, by using the rigidity of $C$, a unique robot is selected from either $\ell_{n_i}$ or either $\ell_{n_k}$ (not both) to move. Its destination is its adjacent empty node outside its current $\ell$-ring toward $\ell_i$ (respectively $\ell_k$) if the robot was elected from $\ell_{n_i}$ (respectively $\ell_{n_k}$). 
If $\ell_{n_i}=\ell_{n_k}$ then the direction of the selected robot is chosen by the adversary.
\end{itemize}
\item $C \in \mathcal{C}_{Undefined}$. Depending on the number of robots on $\ell_i$ and $\ell_k$, we distinguish the following two cases: 
    \begin{enumerate}
        \item $nb_{\ell_i}(C) < nb_{\ell_k}(C)$. The idea is to make robots on $\ell_i$ gather on a single node on $\ell_i$. We define in the following a configuration, denoted $\Gamma(C)$, built from $C$ ignoring some $\ell$-rings and that will be used in order to identify a single node on $\ell_i$ on which all robots on $\ell_i$ will gather. In the case in which there are at least four occupied $\ell$-rings in $C$ then $\Gamma(C)$ is the configuration built from $C$ ignoring both $\ell_i$ and $\ell_k$. By contrast, if there are only three occupied $\ell$-rings then $\Gamma(C)$ is the configuration built from $C$ ignoring only $\ell_i$. The following cases are possible:
            \begin{enumerate}
                \item Configuration $\Gamma(C)$ is rigid. Using the rigidity of $\Gamma(C)$, one node on $\ell_i$ is elected as the gathering node. Let us refer to such a node by $u$. Robots on $\ell_i$ moves to join $u$ starting from the closest ones and taking the shortest path. 
                \item Configuration $\Gamma(C)$ has exactly one axes of symmetry. In this case the axes of symmetry of $\Gamma(C)$ either intersects with $\ell_i$ on a single node (edge-node symmetric), or on two nodes (node-node symmetric) or only on edges (edge-edge symmetric). We discuss in the following each of the mentioned cases:
                \begin{itemize}
                    \item $\Gamma(C)$ is node-edge symmetric:  The single node on $\ell_i$ that is on the axes of symmetry of $\Gamma$ is identified as the gathering node. Robots on $\ell_i$ move to join this node starting from the closest robots and taking the shortest path. 
                    
                    \item $\Gamma(C)$ is node-node symmetric:  Let $u_1$ and $u_2$ be the two nodes on $\ell_i$ on which the axes of symmetry passes through. If both nodes are occupied, then using the rigidity of $C$, exactly one of the two nodes is elected. Assume without loss of generality that $u_1$ is elected. Robots on $u_1$ move to their adjacent node. If both $u_1$ and $u_2$ are empty then let $R$ be the set of robots on $\ell_i$ that are at the smallest distance from either $u_1$ or $u_2$.  If $|R|=1$ (Let $r$ be the robot in $R$ and assume without loss generality that $r$ is the closest to $u_1$) then, $r$ moves to its adjacent node on its current $\ell$-ring toward $u_1$ taking the shortest path. By contrast, if $|R|>1$ then using the rigidity of $C$, exactly one robot of $R$ is elected to move. Its destination is its adjacent node on its current $\ell$-ring toward the closest node being on the axes of symmetry of $\Gamma(C)$, taking the shortest path.    
                    
                    \item $\Gamma(C)$ is edge-edge symmetric. Without loss of generality, assume that the axes of symmetry of $\Gamma(C)$ passes through $\ell_i$ on the following two edges $e_1=(u_1,u_2)$ and $e_1=(u_3,u_4)$ with $u_1$ and $u_3$ being on the same side. Let $U=\{u_j,~ j \in [1-4] \}$. We distinguish the following cases:
                    \begin{itemize}
                        \item For all $u \in U$, $u$ is occupied. Using the rigidity of $C$, a single node $u \in U$ is elected. Robots on $u$ move to their adjacent node $u'\in U$ (refer to Figure \ref{fig:e-e1}, (A)).
                           
                        \item Three nodes of $U$ are occupied. Assume without loss of generality that $u_1$ is the empty node of $U$. If there are robots on $\ell_i$ which are located on the same side as $u_1$ and $u_3$ with respect to the axes of symmetry of $\Gamma(C)$ then the robots among these which are the closest to $u_3$ are the ones to move. Their destination is their adjacent node on their current $\ell$-ring toward $u_3$ (refer to Figure (refer to Figure \ref{fig:e-e1}, (B)). By contrast, if there are no robots on $\ell_i$ which are on the same side of $u_1$ and $u_3$ then robots on $u_2$ are the ones allowed to move. Their destination is their adjacent node in the opposite direction of $u_1$ (refer to Figure \ref{fig:e-e1}, (C)).

                    \begin{figure}[htb]
                    \begin{center}
                    \includegraphics[scale=0.35]{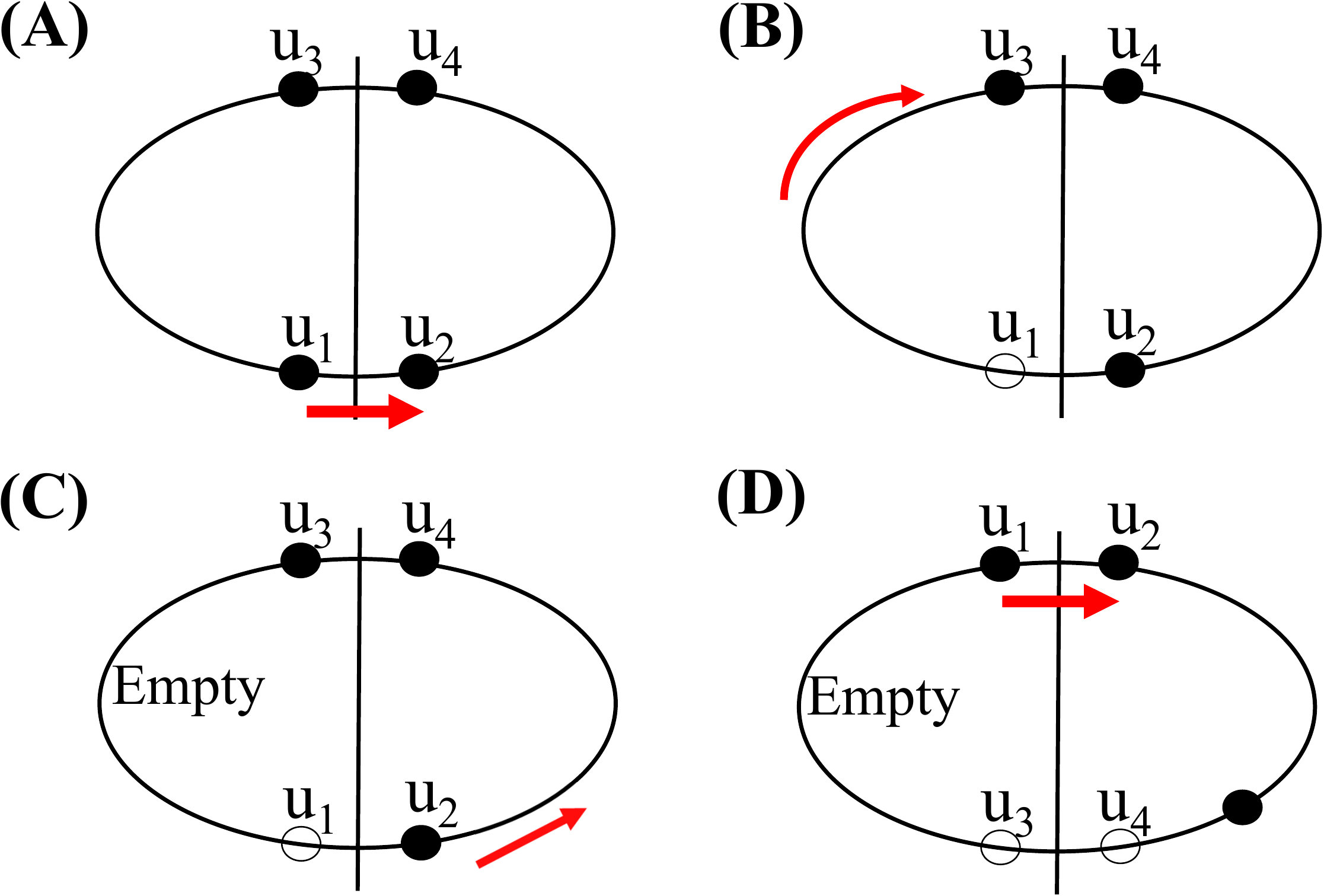}
                    \caption{$\Gamma(C)$ is edge-edge symmetric - Part 1}\label{fig:e-e1}
                    \end{center}
                    \end{figure}
                        
                        \item Two nodes of $U$ are occupied. First, assume without loss of generality that $u_1$ and $u_2$ are occupied (the case where the two nodes are neighbors). If all robots on $\ell_i$ are on the same side of the axes of symmetry (assume without loss of generality that they are at the same side as $u_1$). Robots on $u_2$ are the ones allowed to move. Their destination is their $u_1$ (refer to Figure \ref{fig:e-e2}, (A)). By contrast, if there are robots on both sides of the axes of symmetry of $\Gamma(C)$ then let $U'$ be the set of occupied nodes on $\ell_i$ which are the farthest from the occupied node of $U$ which is on the side (of the axes of symmetry). If there are two of such nodes (one at each side), as $C$ is rigid, the scheduler elects exactly one of these two nodes. Let us refer to the elected node by $u$. Robots on $u$ are the ones allowed to move. Their destination is their adjacent node on their current $\ell$-ring towards the occupied node of $U$ which is on their side (refer to Figure \ref{fig:e-e2}, (B)). By contrast if there is only one node in $U'$ then, robots on the other side of the axes of symmetry are the ones allowed to move starting starting from the robots that are the closest to the occupied node of $U$ which is on their side. Their destination is their adjacent node on their current $\ell$-ring toward the occupied node of $U$ on their side (refer to Figure \ref{fig:e-e2}, (C)). Finally, if there are no robots on both side of the axes of symmetry, then using the rigidity of $C$, one occupied node of $U$ is elected. Robots on the elected node are the ones allowed to move. Their destination is their adjacent occupied node in $U$.
                        
                        Next, assume without loss of generality that $u_1$ and $u_3$ are occupied (the two node of $U$ are not neighbors but are at the same side of $\Gamma(C)$'s axes of symmetry). Robots on a node of $U$ with the largest view are the ones allowed to move. Their destination is their adjacent node in the opposite direction of a node of $U$ (refer to Figure \ref{fig:e-e2}, (D)). Finally, assume without loss of generality that $u_1$ and $u_4$ are occupied (the two nodes of $U$ are not neighbors and are in opposite sides of the axes of symmetry). Robots on a node of $U$ with the largest view are the ones allowed to move. Their destination is their adjacent node in the opposite direction of their adjacent node in $U$ (refer to Figure \ref{fig:e-e2}, (A)). 
                        
                    \begin{figure}[htb]
                    \begin{center}
                    \includegraphics[scale=0.35]{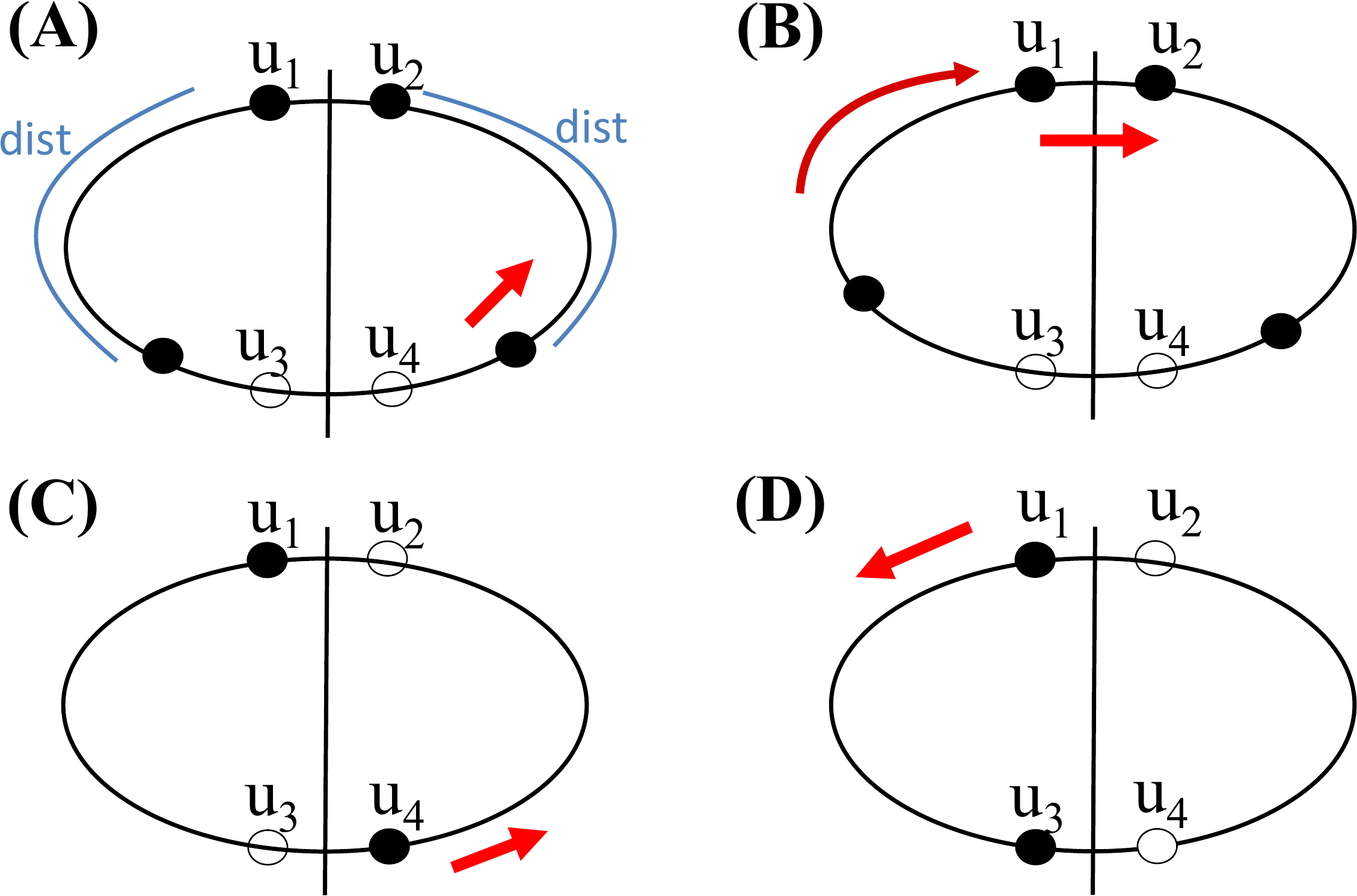}
                    \caption{$\Gamma(C)$ is edge-edge symmetric - Part 2}\label{fig:e-e2}
                    \end{center}
                    \end{figure}
                        
                        \item There is only one node of $U$ that is occupied. Assume without loss of generality that $u_1$ is occupied. If all robots on $\ell_i$ are on the same side of the axes of symmetry as $u_1$ then the closest robots to $u_1$ on $\ell_i$ is allowed to move. Its destination is its adjacent node towards $u_1$ taking the shortest path. By contrast, if all robots on $\ell_i$ are in the opposite side of the axes of symmetry of $u_1$ then robots on $u_1$ are the ones to move. Their destination is $u_2$. Finally, if robots on $\ell_i$ are on both sides of the axes of symmetry then the closest robots to $u_1$ which are on the same side of the axes of symmetry as $u_1$ are the ones allowed to move. Their destination is their adjacent node on $\ell_i$ towards $u_1$ taking the shortest path. 
                        
                        \item All nodes of $U$ are empty. Let $d$ be the smallest distance between a node of $u \in U$ and a robot being on the same side of the axes of symmetry as $u$. Let $R$ be the set of robots that are at distance $d$ from a node $u \in U$. If $|R|=1$ then the unique robot in $R$ moves towards the closest node $u \in U$. By contrast, if $|R|>1$ then, using the rigidity of $C$, a unique robot in $R$ is selected to move. Its destination is its adjacent node toward the closest node $u\in U$.   
                         
                    \end{itemize}
                \end{itemize}
                
                \item Configuration $\Gamma(C)$ has more than one axes of symmetry. Using the rigidity of $C$, a single robot from $\Gamma(C)$ is elected to move. Its destination is its adjacent empty node on its current $\ell$-ring. This reduces the number of axes of symmetries to either $1$ or $0$ (Please refer to Lemma\ref{lem:periodicGamma}). 
            \end{enumerate}
            \item $nb_{\ell_i}(C) = nb_{\ell_k}(C)$. Let $\ell'_i$ and $\ell'_k$ be the two $\ell$-rings that are adjacent to respectively $\ell_i$ and $\ell_k$. Assume without loss of generality that $\ell'_i$ is empty while $\ell'_k$ hosts at least one robot. Using the rigidity of $C$, one robot from $\ell_i$ is elected to move. Its destination is its adjacent empty node on $\ell'_i$. In the case in which both $\ell'_i$ and $\ell'_k$ are empty, using the rigidity of $C$, one robot is elected to move. Assume without loss of generality that the elected robot is on $\ell_i$. The destination of the elected robot is its adjacent empty node on $\ell'_i$. By contrast, if neither $\ell'_i$ nor $\ell'_k$ is empty then, as for the case in which $nb_{\ell_i}(C) < nb_{\ell_k}(C)$, we use configuration $\Gamma(C)$ to elect the robot to move. Recall that $\Gamma(C)$ is defined as the configuration in which both $\ell_i$ and $\ell_k$ are ignored (the set of $\ell$-ring without $\ell_i$ and $\ell_k$). Robots proceed as follows: 
            \begin{enumerate}
                \item Configuration $\Gamma(C)$ is rigid. Using the rigidity of $\Gamma(C)$, a unique node $u$ is elected on either $\ell_i$ or $\ell_k$. Assume without loss of generality that $u$ is elected on $\ell_i$. If $u$ is empty and $nb_{\ell_i}(C)=\ell-1$ then the robot on $\ell_i$ that is adjacent to $u$ and is not part of a multiplicity moves to $u$. Otherwise, robots on $\ell_i$ which are the closest to $u$ move to their adjacent node toward $u$ taking the shortest path.
                \item Configuration $\Gamma(C)$ contains one axes of symmetry. We distinguish the following cases: 
                    \begin{itemize}
                        \item  $\Gamma(C)$ is node-edge symmetric. Let $u$ and $u'$ be the nodes on respectively $\ell_i$ and $\ell_k$ on which the axes of symmetry of $\Gamma(C)$ passes through. If, without loss of generality, only $u$ is occupied then the closest robot to $u$ on $\ell_i$ is the one allowed to move. Its destination is its adjacent node towards $u$ taking the shortest path. If there are two such robots, one is elected using the rigidity of $C$. On another hand, if both $u$ and $u'$ are occupied then using the rigidity of $C$ one closest robot to either $u$ or $u'$ is elected. Assume that the elected robot $r$ is on $\ell_i$. The destination of $r$ is its adjacent node on its current $\ell$-ring towards $u$ taking the shortest path. Finally, if both $u$ and $u'$ are empty then the idea is to make either $u$ or $u'$ occupied. Let $R_1(C)$ (respectively $R_2(C)$) be the set of robots on $\ell_i$ (respectively $\ell_k$) which are not on $u$ (respectively $u'$) but are the closest $u$ (respectively $u'$). Let $R(C) = R_1(C) \cup R_2(C)$. If $|R(C)|>1$ then using the rigidity of $C$ a unique robot of $R(C)$ is elected to move otherwise, the unique robot in $R(C)$ is the one allowed to move. Let us refer to such a robot by $r$ and assume without loss of generality that $r$ is on $\ell_i$. Robot $r$ is the only one allowed to move. Its destination is its adjacent empty node on its current $\ell$-ring toward $u$ taking the shortest path. 
                        
                        
                        \item $\Gamma(C)$ is node-node symmetric.  Let $u_i$ and $u'_i$ (respectively $u_k$ and $u'_k$) be the two nodes on $\ell_i$ (respectively $\ell_k$) on which the axes of symmetry of $\Gamma(C)$ passes through. First assume that $\forall u \in \{u_i, u'_i, u_k, u'_k\}$, $u$ is occupied. Let $U \subseteq \{u_i, u'_i, u_k, u'_k\}$ be the set of nodes that have an occupied adjacent node on their $\ell$-ring. If $|U|\geq 1$ then, using the rigidity of $C$ a single node in $U$ is elected.  Assume without loss of generality that this node is $u_i$. The robot on $u_i$ is the one allowed to move. Its destination is its adjacent occupied node on its $\ell$-ring. By contrast, if  $|U| = 0$. Again, using the rigidity of $C$, one node of $ u \in  \{u_i, u'_i, u_k, u'_k\}$ is elected. The robot on node $u$ is the one allowed to move. Its destination is its adjacent empty node on its current $\ell$-ring (the scheduler chooses the direction to take). Next, let us consider the case in which $\exists u \in  \{u_i, u'_i, u_k, u'_k\}$ such that $u$ is empty. Several cases are possible depending on the nodes that are occupied. If there is a unique node $u \in  \{u_i, u'_i, u_k, u'_k\}$ which is occupied then the closest robot to $u$ on the same $\ell$-ring as $u$ moves to its adjacent node toward $u$. By contrast, if there are three nodes in $\{u_i, u'_i, u_k, u'_k\}$ which are occupied then assume without loss of generality that $u'_i$ is the empty node. Robots on $\ell_i$ which are the closest to $u_i$ are the only ones allowed to move. Their destination is their adjacent node towards $u_i$. Finally, if there are two nodes of $ \{u_i, u'_i, u_k, u'_k\}$ which are occupied then, if the two nodes are part of the same $\ell$-ring (assume without loss of generality that these two nodes are $u_i$ and $u'_i$) then using the rigidity of $C$ one node among $u_i$ and $u'_i$ is elected. The robot on the elected node is the one allowed to move. Its destination is its adjacent node on its current $\ell$-ring (the direction is chosen by the scheduler). If the two occupied nodes are on two different $\ell$-rings then let $R(C)$ be the set of robots on $\ell_i$ and $\ell_k$ which are the closest to the occupied node of $\{u_i, u'_i, u_k, u'_k\}$ on their $\ell$-ring. Using the rigidity of $C$, one robot from $R(C)$ is elected to move. Its destination is its adjacent node on its current $\ell$-ring toward the occupied node of $\{u_i, u'_i, u_k, u'_k\}$ which is on its $\ell$-ring. 
                        
                        \item  $\Gamma(C)$ is edge-edge symmetric. Assume without loss of generality that the axes of symmetry of $\Gamma(C)$ passes through respectively $e_1=(u_1,u_2)$ and $e_2=(u_3,u_4)$ on $\ell_i$ and $e'_1=(u'_1,u'_2)$ and $e'_2=(u'_3,u'_4)$ on $\ell_k$ and that $u_1$ and $u_3$ (respectively $u'_1$ and $u'_3$) are on the same side of $\Gamma(C)$'s axes of symmetry on $\ell_i$ (respectively $\ell_k$). Let $\mathcal{L}=\{\ell_i, \ell_k\}$, $U_i=\{u_1, u_2, u_3, u_4\}$, $U_k=\{u'_1, u'_2, u'_3, u'_4\}$ and  $U=U_i \cup U_k$. The main idea is to gather all robots on either $\ell_i$ or $\ell_k$. For this purpose, the number of occupied nodes is decreased in exactly one of the two $\ell$-rings of $\mathcal{L}$ so that we can reach a configuration $C'$ in which $nb_{\ell_i}(C') \ne nb_{\ell_k}(C')$ and hence use the strategy explained previously. That is, a single $\ell$-ring $\ell_r \in \mathcal{L}$ needs to be selected. Robots on the selected $\ell$-ring, $\ell_r$ behaves in the same manner as in the case in which $nb_{\ell_i}(C) \ne nb_{\ell_k}(C)$. That is, we describe in the following how a unique $\ell$-ring in $\mathcal{L}$ is selected. 
                        
                        Given two nodes $u$ and $u'$ of $U$ located on the same side of $\Gamma(C)$'s axes of symmetry and being on the same $\ell$-ring $\ell_j$, we refer by $Free(u,u')$ to the predicate that is equal to true is there is no robot on the shortest sequence of nodes between $u$ and $u'$ (excluding $u$ and $u'$) on $\ell_j$. In the case where $Free(u,u')$ is true for a given $\ell$-ring $\ell_j$, we say that $\ell_j$ has an empty side. Otherwise, we say that $\ell_j$ has an occupied side. Observe that each $\ell$-ring of $\mathcal{L}$ have at most two empty sides. In order to select a unique $\ell$-ring in $\mathcal{L}$, robots checks if the following properties hold in this order and decide accordingly: \\

                        - First, there exists an $\ell$-ring in $\mathcal{L}$, let this $\ell$-ring be without loss of generality $\ell_i$, in which:  $|U_i|=2$ with the two occupied nodes of $U_i$ being adjacent to each other (Assume without loss of generality that $u_1$ and $u_2$ are the occupied node of $U_i$), $Free(u_2,u_4) \wedge Free(u_1,u_3)$ holds. Observe that since $C$ is rigid and $\Gamma(C)$ is symmetric, on $\ell_k$, if $|U_k|=2$ then the two occupied nodes of $U_k$ cannot be neighbors. That is, $\ell_k$ cannot satisfy the same properties at the same time in $C$. In this case $\ell_i$ is the one that is selected. As this case is not considered in a configuration $C'$ in which $nb_{\ell_i}(C') \ne nb_{\ell_k}(C')$, in addition to the selection process, we describe robots behavior: the robot that is on a node of $U_i$ with the largest view is the one allowed to move. Its destination is $u_1$.  \\

                        - Next, there exists an $\ell$-ring in $\mathcal{L}$, let this $\ell$-ring be without loss of generality $\ell_i$, in which:  $|U_i|=1$ (Assume without loss of generality that $u_1$ is the occupied node of $U_i$), $Free(u_2,u_4) \wedge \neg Free(u_1,u_3)$ holds. In this case, if $\ell_k$ does not satisfy the same properties then $\ell_i$ is elected. Otherwise, assume without loss of generality that $u'_1$ is the occupied node in $U_k$. Let $d_i$ (respectively $d_k$) be the smallest distance between a robot on $\ell_i$ (respectively $\ell_k$) and $u_1$ (respectively $u'_1$). If $d_i \ne d_k$ (assume without loss of generality that $d_i < d_k$) then $\ell_i$ is elected. Otherwise (\ie $d_i = d_k$), let $r$ (respectively $r'$) be the robot on $\ell_i$ (respectively $\ell_k$) which is at distance $d_i$ from $u_1$ (respectively $u'_1$). If $view_{r}(t)(1) > view_{r'}(t)(1)$ then $\ell_i$ is elected. Otherwise, $\ell_k$ is elected. \\
                        
                        - Next, there exists an $\ell$-ring in $\mathcal{L}$, let this $\ell$-ring be without loss of generality $\ell_i$, in which: $|U_i|=2$ with the two occupied nodes of $U_i$ being neighbors (let these two nodes be respectively $u_1$ and $u_2$), $Free(u_1,u_3) \vee Free(u_2,u_4)$ holds. If $\ell_k$ does not satisfy the same properties as $\ell_i$ then, $\ell_i$ is the one that is selected. Otherwise, the $\ell$-ring that hosts a robot located on a node of $U$ with the largest view is the one that is selected. \\
                        
                        - Next, there exists an $\ell$-ring in $\mathcal{L}$, let this $\ell$-ring be without loss of generality $\ell_i$, in which:  $|U_i|=2$ with the two occupied nodes of $U_i$ being on the same side of $\Gamma(C)$'s axes of symmetry (let these robots be respectively $u_1$ and $u_3$), $Free(u_2, u_4)$ and $\neg Free(u_1,u_3)$ holds. If $\ell_k$ does not satisfy the same properties as $\ell_i$ then, $\ell_i$ is the one that is selected. Otherwise, the $\ell$-ring of $\mathcal{L}$ that hosts a robot on a node of $U$ with the largest view is the one to be selected.
                        
                        - Next, there exists an $\ell$-ring in $\mathcal{L}$, let this $\ell$-ring be without loss of generality $\ell_i$, in which: $|U_i|=3$ (let $u_1$ be the unique empty node in $U_i$) and $Free(u_1,u_3)$ holds. If $\ell_k$ does not satisfy the same properties then, $\ell_i$ is selected. Otherwise, the $\ell$-ring that hosts a robot on a node of $U$ with the largest view is the one to be elected. 
                        
                        For all the remaining cases, if, without loss of generality, $|U_i|<|U_k|$ then $\ell_i$ is the selected. Otherwise, in the case where $|U_i|=|U_k|$ then, the selection is achieved as follows: 
                        
                         \begin{itemize}
                            \item $|U_i|=4$. Let $R$ be the set of robots on a node of $U$. The $\ell$-ring that hosts the robots of $R$ with the largest view is the one to be selected (Recall that $C$ is rigid). 
                            
                            \item $|U_i|=3$. Assume without loss of generality that $u_1$ and $u'_1$ are the empty nodes on $\ell_i$ and $\ell_k$ respectively. Observe that $\neg Free(u_1,u_3)$ an $\neg Free(u'_1,u'_3)$ holds.  Let $R_i$ (respectively $R_k$) be the set of robots on $\ell_i$ (respectively $\ell_k$) being on the same side of the axes of symmetry as $u_1$ (respectively $u'_1$). Let $R= R_i \cup R_k$. If without loss of generality $|R_i| < |R_k|$ then, $\ell_i$ is elected. By contrast, if $|R_i| = |R_k|$ then, let $r_1$ (respectively $r'_1$) be the robot on $\ell_i$ (respectively $\ell_k$) which is the closest to $u_3$ (respectively $u'_3$). If without loss of generality  $dist(r_1, u_3) < dist(r'_1, u'_3)$ then $\ell_i$ is selected. Otherwise, if $dist(r_1, u_3) = dist(r'_1, u'_3)$ then the views of $r_1$ and $r'_1$ are used for the selection. Assume without loss of generality that $view_{r_1}(t)(1)> view_{r'_1}(t)(1)$. Then, $\ell_i$ is selected. 
                            
                            \item $|U_i|=2$. Three cases are possible on each $\ell$-ring of $\mathcal{L}$ depending on the nodes of $U$ that are occupied: $Case(I)$ the two nodes are neighbors, $Case(II)$ the two nodes are on the same side of $\Gamma(C)$'s axes of symmetry. $Case(III)$ the two nodes are not neighbors and are on different sides of $\Gamma(C)$'s axes of symmetry. We set: $Case(I) > Case(II) > Case(III)$ where $Case(a) > Case(b)$ means that $Case(a)$ has a higher priority than $Case(b)$. That is, if $\ell_i$ and $\ell_k$ have two different priorities then, the $\ell$-ring with the largest priority is elected. By contrast, if the two $\ell$-rings have the same priority (they belong to the same case), the selection is done in the following manner: if both $\ell_i$ and $\ell_k$ belongs to $Case(I)$ (assume without loss of generality that $u_1$, $u_2$ are the two nodes of $U_i$ which are occupied and $u'_1$ and $u'_2$ are the two occupied nodes of $U_k$). Let $F_1$, $F_2$ (respectively $F'_1$, $F'_2$) be the number of robots on $\ell_i$ (respectively $\ell_k$) being on each side of $\Gamma(C)$'s axes of symmetry. Let $F = \mathit{min}(F_1,F_2,F'_1,F'_2)$. The $\ell$-ring that has a side with $F$ robots is selected. If both $\ell$-rings has a side with $F$ robots, we use the distance to break the symmetry. That is, let $d$ be the smallest distance between a robot on an $\ell$-ring of $\mathcal{L}$ and the occupied node of $U$ on the same $\ell$-ring. The $\ell$-ring that hosts a robot at distance $d$ from a node of $U$ is elected. Again, if both $\ell$-rings satisfy the property, the $\ell$-ring hosting the robot at distance $d$ from a node of $U$ with the largest view is elected. Lastly, if both $\ell$-ring of $\mathcal{L}$ belong to $Case(II)$ or if they both belong to $Case(III)$ then, let $r$ (respectively $r'$) be the robot allowed to move on $\ell_i$ (respectively $\ell_k$) with respect to our algorithm. If $view_r(t)(1) > view_{r'}(t)(1)$ then $\ell_i$ is elected. Otherwise, $\ell_k$ is selected. 
                            
                            \item $|U|_i=1$. let $d$ be the smallest distance between a robot on an $\ell$-ring of $\mathcal{L}$ and a node of $U$ located on the same $\ell$-ring.  The $\ell$-ring that hosts a robot at distance $d$ from a node of $U$ with the largest view is the one that is elected (this is possible as $C$ is rigid).  
                            
                            \item $|U_i|=0$. Let $d$ be the smallest distance between a robot on an $\ell$-ring of $\mathcal{L}$ and a node of $U$ on its $\ell$-ring. The $\ell$-ring that hosts a robot at distance $d$ from a node of $U$ with the largest view is the one that is elected.

                         \end{itemize}

                        \item $\Gamma(C)$ has more than one axes of symmetry. Using the rigidity of $C$, a single robot of $\Gamma(C)$ is elected to move to its adjacent node on its current $\ell$-ring. By doing so, the number of axes of symmetry is reduced to either $1$ or $0$ (Please refer to Lemma\ref{lem:periodicGamma}).
                    \end{itemize}
            \end{enumerate}

    \end{enumerate}

\end{enumerate} 

\end{enumerate}

\subsection{Gathering Phase}\label{subsec:gathering} 

Recall that this phase starts from a configuration $C \in \mathcal{C}_{\mathit{target}}$. From $C$, an orientation of the torus can be defined (from $\ell_{\mathit{target}}$ to $\ell_{\mathit{max}}$). The idea is to make all robots that are neither on $\ell_{\mathit{target}}$ nor $\ell_{\mathit{max}}$ move to join $v_{target}$ following the predefined orientation. Then, some robots on $\ell_{\mathit{max}}$ move to join $v_{target}$ while the other are aligned with respect to $v_{target}$ to finally gather all on one node. 

To ease the description of our algorithm, we first identify a set of configurations that can appear in this phase and then present for each of them the behavior of the robots. 

\paragraph{Set of configurations.} Three sets are identified as follows:

 \begin{enumerate}
      \item Set $C_{sp}$ which includes the following four sub-sets: 
             \begin{enumerate}
             
                \item SubSet $\mathcal{C}_{sp-1}$: $C \in C_{sp-1}$ if there are exactly two occupied $\ell$-rings in~$C$.  Let these two $\ell$-rings be $\ell_i$ and $\ell_j$ respectively. The following conditions are verified: (1) $\ell_i$ and $\ell_j$ are adjacent. (2) $nb_{\ell_{j}}(C) < nb_{\ell_i}(C)$ (3) either : 
                 \begin{itemize}
                     \item $nb_{\ell_i}(C)=4$ and $\ell_i$ contains two $1$.blocks of size $2$ being at distance $2$ from each other. Let $u$ be the unique node between the two $1$.blocks on $\ell_i$.
                     \item  $nb_{\ell_i}(C)=3$ and $\ell_i$ contains a $1$.block of size $3$. Let $u$ be the middle node of the $1$.block of size $3$. 
                     \item $nb_{\ell_i}(C)=5$ and $\ell_i$ contains a $1$.block of size $5$. Let $u$ be the middle node of the $1$.block of size $5$. 
                 \end{itemize}
                 
                 (4) Either $nb_{\ell_j}(C)=3$ and $\ell_j$ contains a $1$.block of size $3$ whose middle node is adjacent to $u$ or $nb_{\ell_j}(C)=2$ and $\ell_j$ contains either a $2$.block of size $2$ whose middle node is adjacent to $u$ or a $1$.block of size $2$ having one extremity adjacent to $u$ (refer to Figure~\ref{fig:Sp} for some examples).
             
                 \item SubSet $\mathcal{C}_{sp-2}$: $C \in \mathcal{C}_{sp-2}$ if $C \in \mathcal{C}_{\mathit{target}}$ and $nb_{\ell_{\mathit{target}}}(C)=1$. In addition either one of the following conditions are verified: (1) $nb_{\ell_{\mathit{max}}}(C)=4$ and on $\ell_{\mathit{max}}$ there are two $1$.blocks of size $2$ being at distance $2$ from each other. Let $u$ be the unique node between the two $1$.blocks then $u$ is adjacent to $v_{\mathit{target}}$. (2) $nb_{\ell_{\mathit{max}}}=5$ and on $\ell_{\mathit{max}}$ there is a $1$.block of size $5$ whose middle robot is adjacent to $v_{\mathit{target}}$. (3) $nb_{\ell_{\mathit{max}}}(C)=4$ and on $\ell_{\mathit{max}}$ there is a $1$.block of size $3$ having a unique occupied node at distance $2$. Let $u$ be the unique empty node between the 1.block of size $3$ and the $1$.block of size $1$. Then $u$ is adjacent to $v_{\mathit{target}}$ (refer to Figure~\ref{fig:Sp}).

                    \begin{figure}[h]
                        \begin{center}
                            \includegraphics[scale=0.37]{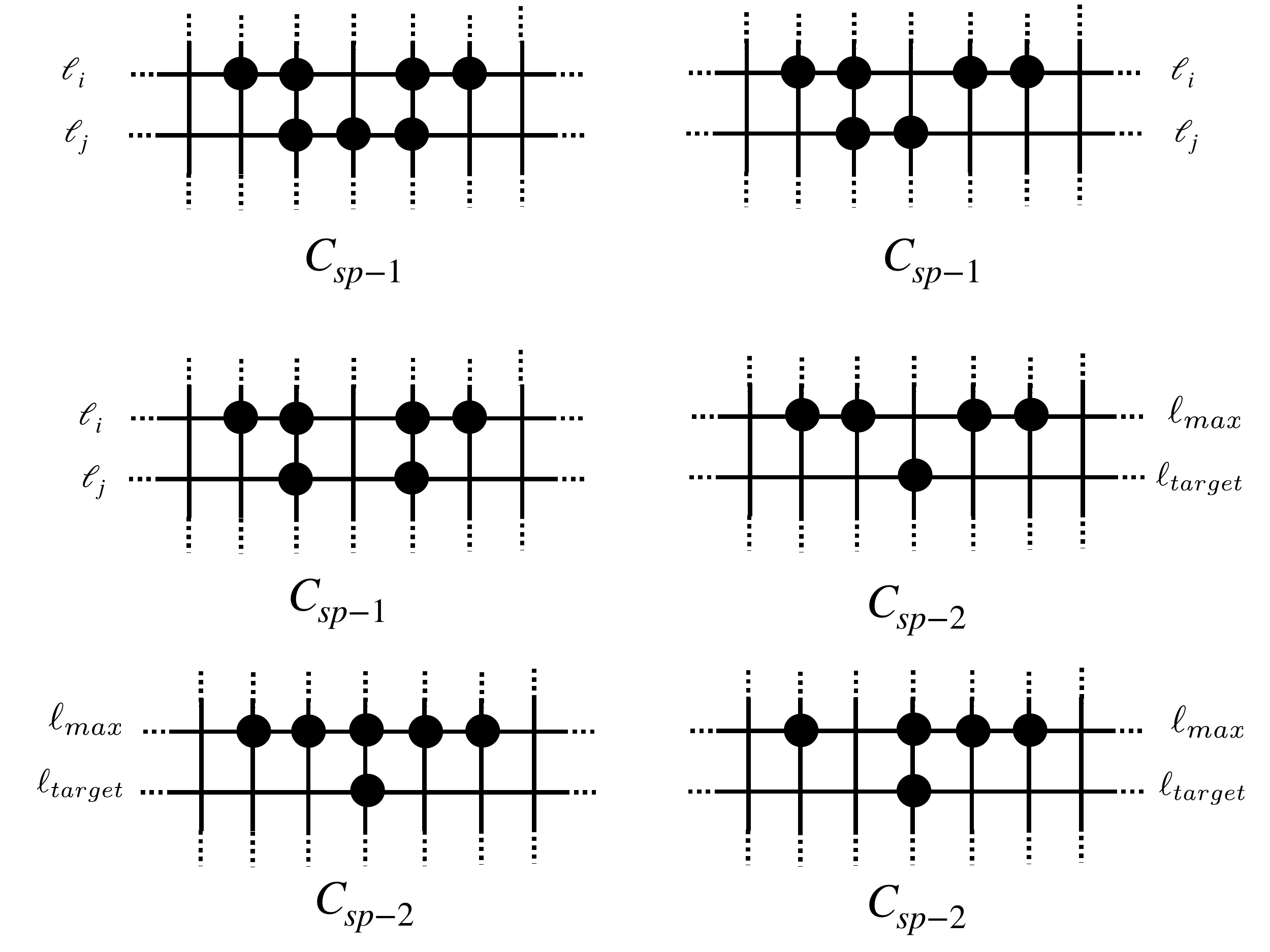}
                            \caption{Subsets $\mathcal{C}_{sp-1}$ and $\mathcal{C}_{sp-2}$}\label{fig:Sp}
                        \end{center}
                    \end{figure}

                 \item SubSet $\mathcal{C}_{sp-3}$: $C \in \mathcal{C}_{sp-3}$ if $C \in \mathcal{C}_{\mathit{target}}$, $Empty(C)$ is true,  $nb_{\ell_{\mathit{target}}}(C)= 1$ and one of the two following conditions holds: (1) $nb_{\ell_{\mathit{max}}}(C)=3$ and $\ell_{\mathit{max}}$ contains an $1$.block of size $3$ whose middle robot is adjacent to $v_{\mathit{target}}$. (2)  $nb_{\ell_{\mathit{max}}}(C)=2$ and the two robots form a 2.block on $\ell_{\mathit{max}}$. Let $u$ be the unique empty node between the two robots on $\ell_{\mathit{max}}$, then $u$ is adjacent to $v_{\mathit{target}}$ (refer to Figure~\ref{fig:FinalSet}).
                 
                 \item SubSet $\mathcal{C}_{sp-4}$: $C \in C_{sp-4}$ if there is a unique $\ell$-ring that is occupied and on this $\ell$-ring there are either two or three occupied nodes that form a $1$.block (refer to Figure~\ref{fig:FinalSet}).

                \begin{figure}[!h]
                    \begin{center}
                        \includegraphics[scale=0.37]{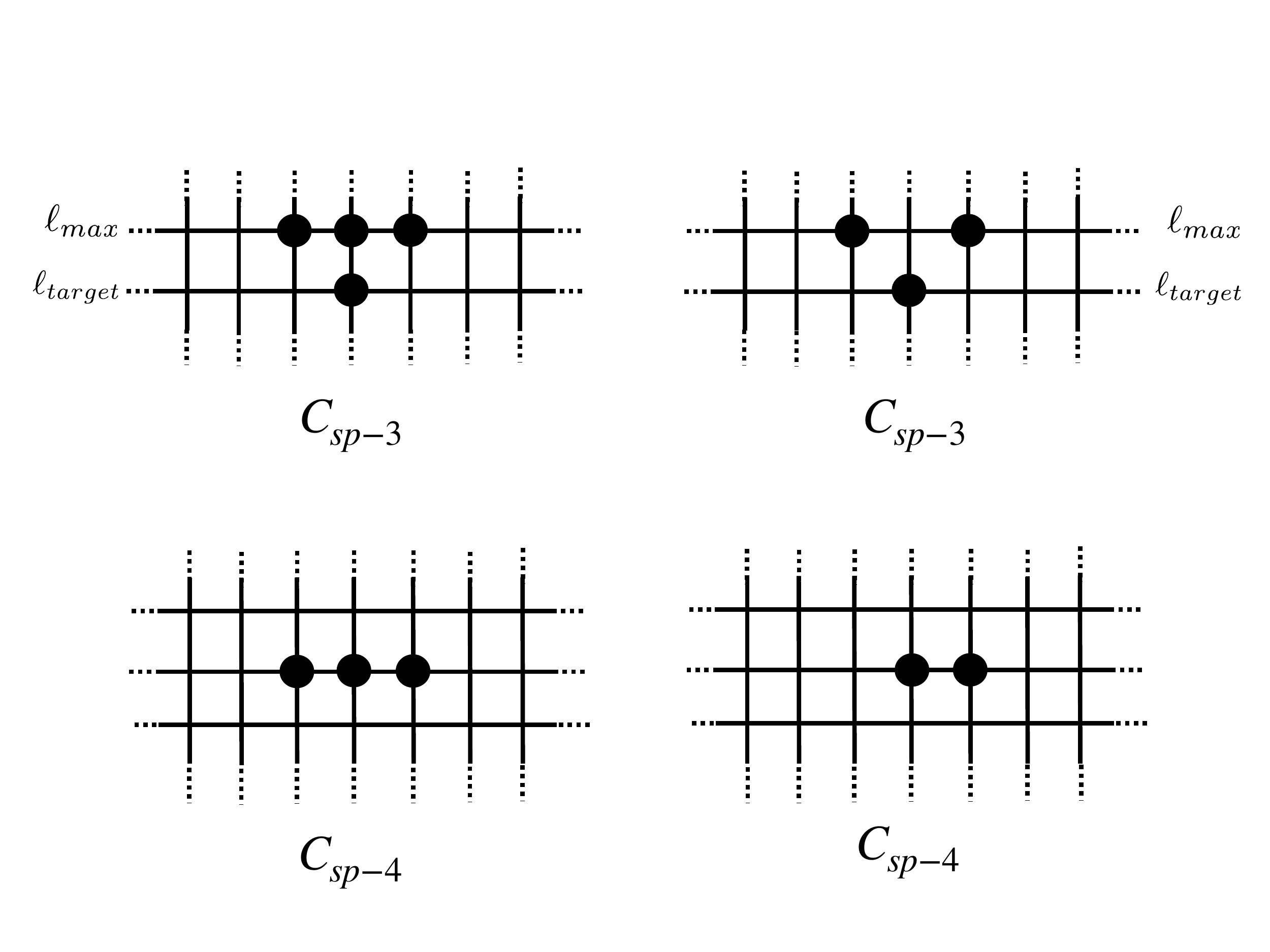}
                        \caption{Subsets $\mathcal{C}_{sp-3}$ and $\mathcal{C}_{sp-4}$}\label{fig:FinalSet}
                    \end{center}
                \end{figure}

             \end{enumerate}
   
       \item  Set $\mathcal{C}_{pr}$: $C \in \mathcal{C}_{pr}$ if $C\in \mathcal{C}_{\mathit{target}}$ and $Partial(C)$ is true. That is, $\exists~i \in \{1, \dots, L\}$ such that $\ell_i \ne \ell_{\mathit{max}}$ and $\ell_i \ne \ell_{\mathit{target}}$ and $nb_{\ell_i}(C)>0$. Note that we are sure that $C \not \in \mathcal{C}_{sp}$.
        
        \item Set $\mathcal{C}_{ls}$: $C \in \mathcal{C}_{ls}$ if $C \in \mathcal{C}_{\mathit{target}}$ and $C \not \in \mathcal{C}_{sp}$ and $Empty(C)$. In other words, there are only two $\ell$-rings that are occupied: $\ell_{\mathit{max}}$ and $\ell_{\mathit{target}}$.\\

   \end{enumerate}

\paragraph{Robots behavior.} We present now the behavior of robots during the gathering phase. If the current configuration $C \in \mathcal{C}_{\mathit{target}}$ then we define $\uparrow$ as the direction from $\ell_{\mathit{target}}$ to $\ell_{\mathit{max}}$ taking the shortest path. Observe that $\uparrow$ can be computed by all robots and in addition, $\uparrow$ is unique (recall that $\ell_{\mathit{max}}$ is unique and $\forall~C \in \mathcal{C}_{\mathit{target}}$, $nb_{\ell_{\mathit{target}}}(C) \ne nb_{\ell_{\mathit{secondary}}}(C)$). 
Using Direction $\uparrow$, we define a total order on the $\ell$-rings of the torus such that $\ell_i \leq \ell_j$ if $\ell_i$ is not further to  $\ell_{\mathit{target}}$ than $\ell_j$  with respect to Direction $\uparrow$. 

Note that $\mathcal{C}_{p_2}= \mathcal{C}_{pr} \cup \mathcal{C}_{ls} \cup \mathcal{C}_{sp}$. Let $C$ be the current configuration, we present robots behavior for each defined set:  

\begin{enumerate} 
\item $C \in \mathcal{C}_{pr}$. Let us refer by $\ell_i$ to the $\ell$-ring that is adjacent to $\ell_{\mathit{target}}$ such that $\ell_{i}\ne \ell_{\mathit{max}}$. Depending on the number of robots on $\ell_i$, two cases are possible as follows:

			\begin{enumerate} 
				\item $nb_{\ell_i}(C)>0$. Let $R_m$ be the set of robots on $\ell_i$ that are the closest to $v_{\mathit{target}}$. We distinguish the following cases:
					\begin{enumerate} 
       						\item There is an occupied node on $\ell_i$ that is adjacent to $v_{\mathit{target}}$. Let us refer to such a node by $u_i$. Robots on $u_i$ are the ones allowed to move. Their destination is their adjacent node on $\ell_{\mathit{target}}$ \ie they move to $v_{\mathit{target}}$.
       						\item There is no robots on $\ell_i$ that is adjacent to $v_{\mathit{target}}$ and $nb_{\ell_{i}}(C)< \ell-1$. In this case, robots in $R_m$ are the ones allowed to move, their destination is their adjacent empty node on $\ell_{i}$ on the empty path toward $v_{\mathit{target}}$. 
       						\item There is no robots on $\ell_i$ that is adjacent to $v_{\mathit{target}}$ and $nb_{\ell_{i}}(C)= \ell-1$. In this case, let $R_{m'}$ be the set of robots that share a hole with $u_i$ where $u_i$ is the node on $\ell_i$ that is adjacent to $v_{\mathit{target}}$. Robots in $R_{m'}$ are allowed to move only if they are not part of a multiplicity. Their destination is their adjacent empty node towards $u_i$ taking the empty path. 
					\end{enumerate} 

				\item $nb_{\ell_i}(C)=0$. Let $\ell_k$ be the closest neighboring $\ell$-ring of $\ell_{\mathit{target}}$ with respect to $\uparrow$. Let $R_{m}$ be the set of robots on $\ell_{k}$ that are the closest to $v_{\mathit{target}}$. Robots on $R_{m}$ are the ones allowed to move, their destination is their adjacent node outside $\ell_{k}$ and toward $\ell_{\mathit{target}}$ with respect to $\uparrow$. 
				\end{enumerate} 

\item $C \in \mathcal{C}_{ls}$. The aim for the robots is to reach a configuration $C' \in \mathcal{C}_{sp}$. 

If $nb_{\ell_{\mathit{max}}}(C)<=5$, robots on $\ell_{\mathit{max}}$ execute \textbf{Align}$(\ell_{\mathit{max}}, \ell_{\mathit{target}})$. Otherwise, robots behave as follows: Let $u_1, u_2, u_3$, $u_4$ and $u_5$ be a sequence of five consecutive nodes on $\ell_{\mathit{max}}$ such that $u_3$ is adjacent to $v_{\mathit{target}}$. If $u_3$ is occupied and has exactly one adjacent occupied node on $\ell_{\mathit{max}}$ (assume without loss of generality that this node is $u_2$) then the robot on $u_2$ is the one allowed to move. Its destination is $u_3$. By contrast, if $u_3$ has either no adjacent occupied nodes on $\ell_{\mathit{max}}$ or two adjacent occupied nodes on $\ell_{\mathit{max}}$ then robots on $u_3$ are the ones allowed to move. Their destination is $v_{\mathit{target}}$. Finally, if $u_3$ is empty then let $R$ be the set of robots that are on $\ell_{\mathit{max}}$ which are the closest to $u_3$. If $|R|=2$ then both robots in $R$ are allowed to move. Their destination is their adjacent node on $\ell_{\mathit{max}}$ toward $u_3$.  By contrast, if $|R|=1$ then first assume that the distance between the robot in the set $R$ and $u_3$ is $d$. If there is a robot on $\ell_{\mathit{max}}$ which shares a hole with $u_3$ and which is at distance $d+1$ from $u_3$ then this robot is the one allowed to move. Its destination is its adjacent empty node towards $u_3$ taking the shortest path. If no such robot exists then the robot in the set $R$ is the one allowed to move. Its destination is its adjacent node toward $u_3$ taking the shortest path.  

\item $C \in \mathcal{C}_{sp}$. We distinguish:
        \begin{enumerate} 
            \item $C \in \mathcal{C}_{sp-1}$. If $C \in \mathcal{C}_{\mathit{target}}$ then the robots on $\ell_{\mathit{target}}$ that are at the extremities of the $1$.block or the $2$.block are the ones allowed to move. Their destination is their adjacent occupied node on $\ell_{\mathit{max}}$. By contrast, if $C \not \in \mathcal{C}_{\mathit{target}}$ then the robot not on $\ell_{\mathit{max}}$ which has two adjacent occupied nodes is the one allowed to move. Its destination is its adjacent node on $\ell_{\mathit{max}}$. 
           \item $C \in \mathcal{C}_{sp-2}$. Recall that three cases are possible as follows: If there is a $1$.block of size $3$ on $\ell_{\mathit{max}}$ then the robots that are in the middle of the $1$.block of size $3$ moves to their adjacent occupied node that has one robot at distance $2$. If $\ell_{\mathit{max}}$ contains a $1$.block of size $5$ then the robots on $\ell_{\mathit{max}}$ that are adjacent of the extremities of the $1$.block move on $\ell_{\mathit{max}}$ in the opposite direction of the extremities of the $1$.block. Finally, if $\ell_{\mathit{max}}$ contains two $1$.blocks of size $2$ then the robots that share a hole of size $1$ move toward each other.
           \item $C \in \mathcal{C}_{sp-3}$. In this case, robots on $v_{\mathit{target}}$ are the ones that are allowed to move (note that $v_{\mathit{target}}$ can be occupied by either a single robot or a tower). Their destination is their adjacent node on $\ell_{\mathit{max}}$.
           \item $C \in \mathcal{C}_{sp-4}$. If $C$ contains a $1$.block of size $3$ then the robots that are at the extremities of the $1$.block are the ones allowed to move. Their destination is their adjacent occupied node. By contrast, if $C$ contains a $1$.block of size $2$ then the robot that is not part of a tower moves to its adjacent occupied node (As it will be shown later, the proposed solution ensures that in this case, one node hosts a tower while the other one hosts only one robot). 
                   
        \end{enumerate} 

\end{enumerate} 

\section{Proof of correctness}

Let us first show the correctness of Procedure $Align(\ell_i, \ell_k)$. For this purpose, we define the set of configurations to which we refer to by \textit{Aligned}$(\ell_i, \ell_k)$. A configuration $C$ is said to be \textit{Aligned}$(\ell_i, \ell_k)$ if the following conditions hold on $\ell_i$ and $\ell_k$:

     \begin{enumerate}
        \item $nb_{\ell_i}(C)=2$ and robots on $\ell_i$ form a $2$.block of size $2$. By $u$, we refer to the empty node which is the middle of the $2$.block on $\ell_i$. Or $nb_{\ell_i}(C)=3$ (respectively $nb_{\ell_i}(C)=5$) and robots on $\ell_i$ form a $1$.block of size $3$ (respectively of size $5$). By $u$, we refer to the occupied node which is in the middle of the $1$.block on $\ell_i$. Or, $nb_{\ell_i}(C)=2$ and robots on $\ell_i$ form two $1$.blocks of size $2$ being at distance $2$ from each other. Let $u$ be the unique empty node between the two $1$.blocks. 
        \item $nb_{\ell_i}(C)>nb_{\ell_k}(C)$ holds, and either (1) $nb_{\ell_k}(C)=1$ or (2) $nb_{\ell_k}(C)=2$ and $\ell_k$ contains a $2$.block or (3) $nb_{\ell_k}(C)=3$ and $\ell_k$ contains a $1$.block of size $3$. Let $u_{mark}$ be the node on $\ell_k$ that is 
              \begin{itemize}
                 \item occupied if $nb_{\ell_k}(C)=1$.
                 \item empty in the 2.block if $nb_{\ell_k}(C)=2$.
                 \item occupied in the middle of the $1$.block if $nb_{\ell_k}(C)=3$.
              \end{itemize}
        \item Nodes $u$ and $u_{mark}$ are on the same $L$-ring. 
    \end{enumerate}

\begin{lemma}\label{lem:align}
Starting from a configuration in which Align($\ell_i$, $\ell_k$) is called with no robots with outdated views, an Aligned($\ell_i$, $\ell_k$) configuration is eventually reached. 
\end{lemma}

\begin{proof}
As Align($\ell_i$, $\ell_k$) is called only when the second property of \textit{Aligned}$(\ell_i, \ell_k)$ is verified and as no robot is allowed to move on $\ell_k$ when executing Align($\ell_i$, $\ell_k$), we only focus in this proof on robots which are on $\ell_i$. Let $u_1$, $u_2$, $u_3$, $u_4$ and $u_5$ be five consecutive nodes on $\ell_i$ such that $u_3$ is on the same $\ell$-ring as $\ell_k$.  Let $r_1$, $r_2$, ..., $r_j$ be the robots that are located on $\ell_i$ in the case where $nb_{\ell_i}(C)=j$. Depending on the number of robots on $\ell_i$, procedure Align($\ell_i$, $\ell_k$) considers the following cases: 
    \begin{enumerate}
    \item \label{case:2rp} $nb_{\ell_i}(C)=2$. Assume without loss of generality that there is no robot between $r_1$ and $u_2$ on $\ell_i$. By procedure Align($\ell_i$, $\ell_k$), if $u_2$ (respectively $u_4$) is empty then $r_1$ (respectively $r_2$) is allowed to move. Its destination is its adjacent empty node on $\ell_i$ toward $u_2$ (respectively $u_4$). By doing so, the robot gets even more closer to $u_2$ (respectively $u_4$) and hence it is still allowed to move, eventually $u_2$ (respectively $u_4$) becomes occupied and the Lemma holds. 
    
    \item $nb_{\ell_i}(C)=3$. If $u_3$ is occupied then as for case \ref{case:2rp}, we can deduce that eventually $u_2$ and $u_4$ become occupied. That is, we only focus in the following in the case in which $u_3$ is empty. As the scheduler is asynchronous, we need to be careful not to create a tower. If there is a unique robot, say $r_1$ which is the closest to $u_3$ then by procedure Align, $r_1$ is the only one allowed to move. Its destination is its adjacent empty node on its current $\ell$-ring towards $u_3$. That is, eventually $u_3$ becomes occupied. By contrast if there are two robots which are the closest to $u_3$ (assume that these two robots are $r_1$ and $r_2$) then procedure Align uses $r_3$ to break the symmetry. Two special cases are identified to deal with robots with  potential outdated views: the case in which dist($r_1$, $r_3$) $=$ dist($r_2$, $r_3$) and $r_1$ and $r_2$ are part of the same $1$.block (refer to Figure \ref{fig:Align3r}). By procedure Align, both $r_1$ and $r_2$ are allowed to move. Their destination is their adjacent empty node in the opposite direction of $r_3$. If only one of the two robots move (assume without loss of generality that this robot is $r_1$) then a configuration in which there is a $1$.block of size $2$ on $\ell_i$ with dist($r_1$, $r_3$) $=$ dist($r_2$, $r_3$) $+1$ is reached. This situation corresponds the second special case $(c2)$ (refer to Figure \ref{fig:Align3r}). In this case, by procedure Align, robot $r_2$ is the only one allowed. That is, eventually $r_2$ moves to its adjacent empty node in the opposite direction of $r_3$. That is, we are sure that in the configuration reached  none of the robots on $\ell_i$ has an outdated view. Now note that, in the case where dist($r_1$, $r_3$) $=$ dist($r_2$, $r_3$) with $r_1$ and $r_2$ not part of the same $1$.block, then $r_3$ is the one allowed to move, its destination is its adjacent empty node on its current $\ell$-ring (the scheduler chooses the direction to take), that is eventually dist($r_1$, $r_3$) $<$ dist($r_2$, $r_3$) or dist($r_1$, $r_3$) $>$ dist($r_2$, $r_3$). Finally, if without loss of generality dist($r_1$, $r_3$) $<$ dist($r_2$, $r_3$), then $r_1$ is the only one allowed to move. Its destination is its adjacent node on its current $\ell$-ring towards $u_3$. By moving, $r_1$ becomes the only one which is the closest $u_3$ and as discussed previously $u_3$ becomes occupied eventually. We can thus deduce that the lemma holds.
    
    \item $nb_{\ell_i}(C)=4$. If node $u_3$ is occupied and either $u_2$ or $u_4$ is empty. Then by procedure Align, robot on $u_3$ is the one allowed to move. Its destination is its adjacent empty node (in case of symmetry, the scheduler chooses the direction to take). By contrast, if both $u2$ and $u_4$ are occupied then since $nb_{\ell_i}(C)=4$ then either $u_1$ or $u_5$ is empty. Assume without loss of generality that $u_1$ is the node that is empty. In this case, by procedure Align, the robot on $u_2$ is the one allowed to move, its destination is its adjacent empty node in the opposite direction of $u_3$. Hence, we can deduce that $u_3$ becomes eventually empty. Let us now consider the case in which $u_3$ is empty. By procedure align, an order on robots is defined according to the robots distance from $u_3$ and a given orientation of the ring $\ell_i$. Assume that $r_1 \leq r_2 \leq r_3 \leq r_4$. Robot $r_1$'s (respectively $r_4$'s) destination is $u_2$ (respectively $u_4$). Similarly $r_2$'s (respectively $r_3$'s) destination is $u_1$ (respectively $u_5$). As they move to their respective target node starting from the ones which are the closest, we can deduce that in this case too the Lemma hods. 
    
    \item $nb_{\ell_i}(C)=5$. If $u_3$ is occupied, then as for the case in which $nb_{\ell_i}(C)=4$ with $u_3$ empty (robots not on $u_3$ have the same behavior), we can easily show that eventually $u_1$, $u_2$, $u_4$ and $u_5$ become occupied which makes the lemma hold. By contrast, if $u_3$ is empty then given an orientation of the ring, assume that $r_1 \le r_2 \le r_3 \le r_4 \le r_5$ holds where $r \le r'$ means that $r$ is closer to $u_3$ than $r'$ with respect to the chosen orientation. If there is a unique robot which is the closest to $u_3$, by procedure Align, this robot is the only one allowed to move. Its destination is its adjacent empty node towards $u_3$ taking the shortest path. By moving, the robot either joins $u_3$ or becomes even closer. In the later case, the same robot remains the only one allowed to move. That is eventually, $u_3$ becomes occupied. If there are more than one robot that is the closest to $u_3$ (by hypothesis, these two robots are $r_1$ and $r_5$), then procedure Align identifies some special configurations to deal with robots with possible outdated views (refer to Figures \ref{fig:Align5r1} and \ref{fig:Align5r2}): if all robots are part of a single $1$.block with distance($r_1$, $u_3$) $=$ distance($r_5$, $u_3$) then Align makes $r_1$ and $r_5$ moves on their $\ell$-ring outside the $1$.block they belong to. If the scheduler activates only one robot (say $r_1$) then in the reached configuration, by procedure Align, the robot that was supposed to move previously is the only one allowed to move. That is,a configuration in which both $r_1$ and $r_5$ have moved is eventually reached. Now, if dist($r_1$, $r_2$) $=$ dist($r_5$, $r_4$) with $r_2$ and $r_5$ being on the same $1$.block then by procedure Align, $r_2$ and $r_4$ are the only one allowed to move, their destination is their adjacent empty node toward respectively $r_1$ and $r_5$. If only one of the two robots moves, in the reached configuration the robot that was supposed to move is the only one allowed to move. That is a configuration in which both $r_2$ and $r_4$ have moved is reached. For all the other cases in which distance($r_1$, $u_3$) $=$ distance($r_5$, $u_3$), robot $r_3$ is used to break the symmetry. That is, if without loss of generality dist($r_1$, $r_3$) $=$ dist($r_5$, $r_3$) then $r_3$ moves to one of its adjacent empty node on $\ell_i$. By contrast, if without loss of generality, dist($r_1$, $r_3$) $<$ dist($r_5$, $r_3$), then $r_1$ is the one allowed to move. Its destination is its adjacent empty node on $\ell_i$ taking the shortest path. By moving, $r_1$ either joins $u_3$ or becomes the only robot that is the closest to $u_3$. Hence, eventually $u_3$ becomes occupied and we can deduce that the lemma holds. 
    \end{enumerate}
    
    From the cases above, we can deduce that starting from a configuration in which Align($\ell_i$, $\ell_k$) is called, an Aligned($\ell_i$, $\ell_k$) configuration is eventually reached. Hence the lemma holds. 
\end{proof}



Let us now focus on the correctness of the Preparation phase. Let $|\ell_{\mathit{max}}|(C)$ be the number of maximal $\ell$-rings in configuration $C$. 

\begin{lemma}\label{lem:L--}
Let $C$ be a rigid configuration in which $|\ell_{\mathit{max}}|_C>1$ and $nb_{\ell_{\mathit{max}}}(C)=\ell$. On each maximal $\ell$-ring on $C$ there exists at least one robot $r$ such that if $r$ moves to one of its adjacent occupied node on its current $\ell$-ring, the configuration reached $C'$, is rigid and $|\ell_{\mathit{max}}|_C>|\ell_{\mathit{max}}|_{C'}$.  

\end{lemma}

\begin{proof}
We proceed by contradiction. Assume that for every robot $r$ located on a maximal $\ell$-ring in $C$, if $r$ moves to its adjacent occupied node on its current $\ell$-ring in $C$,  then the configuration reached $C'$ is either symmetric or $|\ell_{\mathit{max}}|_C \leq |\ell_{\mathit{max}}|_{C'}$. Let $r$ be one robot on a given $\ell_{\mathit{max}}$. When $r$ moves to its adjacent occupied node on its current $\ell$-ring that we denote $\ell_m$, since initially there is one robot on each node, $nb_{\ell_m}(C')=\ell-1$. Since by assumption no other robot is allowed to move, $|\ell_{\mathit{max}}|_C>|\ell_{\mathit{max}}|_{C'}$. Hence, we deduce that $C'$ is either symmetric or periodic. We first show the following claim:\\

\noindent \textbf{Claim 1}. On each maximal $\ell$-ring $\ell_m$, there exists at most one robot $r$ such that if $r$ moves to one of its adjacent nodes on $\ell$, a periodic configuration is reached. \\

\noindent \textbf{Proof of Claim 1}. Assume by contradiction that the claim does not hold. That is, there exists at least $2$ robots on each maximal $\ell$-rings $\ell_m$ such that if they move to their adjacent node on $\ell_m$, a periodic configuration is reached. 

Let $r$ be such a robot and let $C'$ be the periodic configuration reached. First, observe that when $r$ moves, as there is a single empty node on $\ell_m$, there exists no sequence of $L$-rings $L_0, L_1, L_2, ... L_k$ which is repeated at least twice. Since $C'$ is periodic, we can deduce that, instead, there exists a sequence of $\ell$-rings $\ell_0, \ell_1, \dots, \ell_k$ with $k>1$ which is repeated $t$ times with $t>1$. Assume without loss of generality that $\ell_i$ ($ 0 \leq i \leq k$) is the $\ell$-ring in $C'$ that hosts $r$ ($\ell_i = \ell_m$). Configuration $C'$ can be expressed as $(p_0, \ell_i, p_1)^t$ where $p_0 = \ell_0, \ell_1, \dots, \ell_{i-1}$ and $p_1 = \ell_{i+1}, \ell_{i+2}, \dots, \ell_{k}$.

By assumption, there is another robot $r'$ on the same $\ell$-ring as $r$ such that if $r'$ moves to one of its adjacent node on its current $\ell$-ring, a periodic configuration $C''$ is reached. Let us consider the case in which $r'$ moves. As for the previous case, we can deduce that there exists a sequence of $\ell$-rings $\ell'_0, \ell'_1, \dots, \ell'_{k'}$ with $k'>1$ which is repeated $t'$ times with $t'>1$. Assume without loss of generality that $\ell'_{i}$ is the $\ell$-ring that hosts $r'$ ($\ell'_{i} = \ell_m)$. Recall that $\ell_i=\ell'_i$ . Since $r'$ is the only robot that has moved in $C$, configuration $C''$ can be expressed as $(p_0, \ell_i, p_1)^{t-1}(p_0, \ell'_i, p_1)$. That is, if $C'$ is periodic then $r=r'$ which is a contradiction. \\

We can now assume that each robot $r' \ne r$ on $\ell_m$ when it moves, it creates a symmetric configuration. Let $A$ be one of the axes of symmetry in $C'$. We prove some properties of $A$\\

\noindent \textbf{Claim 2}. $A$ is horizontal to $\ell_m$. \\

\noindent \textbf{Proof of Claim 2}. Suppose by contradiction that $A$ is perpendicular to $\ell_m$. Since $\ell$ contains a single empty node in $C'$ (the node that became empty once $r$ moved), $A$ must intersect with $\ell_m$ on this empty node. Since  $C'$ is symmetric, $C$ is also symmetric. Contradiction.  \\

\noindent \textbf{Claim 3}. $A$ does not lay on $\ell_m$. \\

\noindent \textbf{Proof of Claim 3}. Since $r'$ moved on $\ell$ on its adjacent node which is on the axes of symmetry in $C'$ then $C$ was also symmetric. Contradiction. \\

\noindent \textbf{Claim 4}. Let $r'' \ne r$ be another robot on $\ell_m$ in $C$. If $r''$ moves to its adjacent occupied node on $\ell_m$ in $C$ then if $C''$ is the configuration reached once $r''$ moves and if $C''$ is symmetric then the axes of symmetry in $C''$ are different from the axes of symmetry in $C'$.  \\

\noindent \textbf{Proof of Claim 4}. Let $X$ and $X'$ be the axes of symmetry in respectively $C'$ and $C''$ such that $X=X'$. By Claim 3, there exists another $\ell$-ring $\ell'_m$ that is symmetric to $\ell_m$ with respect to $X$ and $X'$. Since $r'$ and $r''$ are located on two different nodes in $C$, there are two empty nodes on $\ell'_m$ which is a contradiction.\\

Since we consider an asymmetric torus with $L<\ell$, by Claims 3 and 4 we can deduce a contradiction. Thus the lemma holds. 

\end{proof}

From Lemma \ref{lem:L--}, we can deduce:


\begin{lemma}\label{lem:uniquel}\label{lem:UniqueCl}
Starting from a rigid configuration $C$ in which $Unique(C)$ is false and $nb_{\ell_{\mathit{max}}}(C)=\ell$, a rigid configuration $C'$ with no outdated robots and in which $Unique(C)$ is true is eventually reached. Moreover, for any $\ell_i$ with $i \in \{1, 2, ..., \ell\}$ in $C'$, if $\ell_i$ hosts a multiplicity node then the following three conditions are verified: (1) $nb_{\ell_i}(C')= \ell-1$, (2) the multiplicity node hosts exactly two robots and (3) the multiplicity node is at a border of a $1$.block of size $\ell-1$.   
\end{lemma}

\begin{lemma}\label{lem:uniquelessl}\label{lem:UniqueCless}
Starting from a rigid configuration $C$ in which $|nb_{\ell_{mac}}|_C>1$ and $nb_{\ell_{\mathit{max}}}(C)<\ell$,  by executing our algorithm, a rigid configuration $C'$ with no outdated robots and in which $Unique(C')$ is true is eventually reached. Moreover, for any $\ell_i$ with $i \in \{1, 2, ..., \ell\}$ in $C'$, $\ell_i$ does not host a multiplicity node. 
\end{lemma}

\begin{proof}

    By our algorithm, the robot with the largest view in the set $R(C)$ is supposed to move. As $C$ is rigid and initially no node hosts a multiplicity,  a unique robot is allowed to move. Let us first state and prove some important claims: \\

\noindent \textbf{Claim 1}. Assume that $r$ moves on its current $\ell$-ring and its new position is not on the same $L$-ring as $u$. Then, the configuration $C'$ reached once $r$ moves is rigid. Moreover $r \in R(C')$ and $|R(C')|=1$. \\

\noindent \textbf{Proof of Claim 1}.
Assume by contradiction that the claim does not hold. That is either $C'$ is symmetric or $r$ is not the only closest robot to an empty node on $\ell$-max in $C'$. In both cases, this means that there is another robot $r'$ in $C$ which was even closer to an empty node on $\ell$-max. Contradiction.  \\

\noindent \textbf{Claim 2}. Assume that $r$ moves outside its $\ell$-ring but does not join a maximal $\ell$-ring in $C$. Then, the configuration $C'$ reached once $r$ moves is rigid. Moreover $r \in R(C')$ and $|R(C')|=1$. \\

\noindent \textbf{Proof of Claim 2}.
Assume by contradiction that the claim does not hold. That is either $C'$ is symmetric or $r$ is not the closest robot to an empty node on $\ell$-max in $C'$. In both cases, this means that there is another robot $r'$ in $C$ which was even closer to an empty node on $\ell$-max. Contradiction.  \\

By Claims 1 and 2, we know that as long as $r$ does not join a maximal $\ell$-ring or a node that is on the same $L$-ring as $u$ (for the first time), the configuration reached remains rigid. Note that by moving $r$ remains the only robot allowed to move as it becomes the only closest robot to an empty node on a maximal $\ell$-ring. At some time, $r$ has either to join an empty node in the same $L$-ring as $u$ or join $u$. If the configuration remains rigid at each time then a rigid configuration $C'$ in which there is a unique maximal $\ell$-ring is reached and the lemma holds. 

Let us consider the case in which a symmetric configuration $C'$ can be  reached when $r$ moves to join either a node that is on the same $L$-ring as $u$ or $u$ . Recall that by our algorithm, $r$ does not move in this case but $C'$ is computed by all robots to elect one robot to move in $C$. \\




Two cases are possible:

        \begin{enumerate}
            \item Robot $r$ joins an empty node on the same $L$-ring as $u$ for the first time. Note that in this case, the axes of symmetry in $C'$ lies on the $L$-ring containing $r$. Let $d1$ and $d2$ be the two $1$.blocks that has $u$ as a neighbor on $\ell$-target. Since $C'$ is symmetric, the size of $d1$ is equal to the size of $d2$. By our algorithm, $r'$, the robot which is on target-$\ell$ which is on the same $L$-ring as $r$ in $C$ is the one allowed to move. Its destination is $u$ (Note that we are sure of the existence of such a robot since otherwise $r$ will be supposed to move on its current $L$-ring instead of its current $\ell$-ring). In the case where $\ell_{\mathit{max}}(C)=\ell-1$,  by moving, $r'$ does not create a symmetric configuration as there is at least another maximal $\ell$-rings in $C$ whose unique empty node is not on the same $L$-ring as $r$. In the other cases, assume by contradiction that when $r'$ moves, the configuration reached $C''$ is also symmetric. As previously stated, the axes of symmetry lies on the $L$-ring that contains $r$ (otherwise $C$ contains a robot which was closer to an empty node on a maximal $\ell$-ring than $r$). Assume without loss of generality that $r'$ moved toward $d1$. Note that the size of $d1$ in $C''$ has increased by one while the size of $d2$ has decreased by one. That is, $d1$ and $d2$ cannot be symmetric which is a contradiction. We can thus conclude that the reached configuration once $r'$ moves is rigid. Robot $r$ is now on the same $L$-ring as $u$. By Claim 2, we know that as long as $r$ moves toward $u$ without joining $u$, the configuration remains rigid. Once $r$ becomes adjacent to $u$, if by joining $u$ the configuration remains rigid then the lemma holds. Otherwise, we retrieve Case \ref{case:join-u}.

             \item \label{case:join-u} Robot $r$ joins $u$. Let $\ell_m$ and $\ell_i$ be respectively the $\ell$-ring that hosts $u$ and $r$. Let $\ell_k$ be the other $\ell$-ring that is adjacent to $\ell_m$. If the reached configuration $C'$ is rigid then the lemma holds. Let us focus on the case in which $C'$ is symmetric. First assume that the axes of symmetry in $C'$ lies on the unique maximal $\ell$-ring (\ie it lies on $\ell_m$). Note that since $C'$ is symmetric, on $\ell_k$ there is no robot on the same $L$-ring as $r$ in $C$. Since initially, $Unique(C)$ is false, there are two $\ell$-rings which were maximal in $C$ and that are symmetric in $C'$ with respect to the unique maximal $\ell$-ring $\ell_m$. Let us refer to such $\ell$-rings by respectively $\ell_s$ and $\ell_{s'}$. In the following, in a given configuration,  we say that $\ell_i == \ell_{i'}$ (for any $i$ and $i'$) if given an orientation of the torus, the position of the occupied nodes in $\ell_i$ is the exactly the same as in $\ell_{i'}$. Note that $\ell_s == \ell_{s'}$ in $C'$. We write $\ell_i != \ell_{i'}$ if $\ell_i == \ell_{i'}$ does not hold. If $\ell_i == \ell_{i'}$ then $r == r'$ holds if $r'$ has the same position on $\ell_{i'}$ as $r$ on $\ell_i$. 
             We now state some important observations. Since $C'$ is symmetric, we can deduce that the number of $\ell$-rings $\ell_r$ satisfying the following properties is odd in $C$: (i) $\ell_r == \ell_m$, (ii) if $\ell_{i'}$ and $\ell_{k'}$ are adjacent to $\ell_r$ then without loss of generality $\ell_i == \ell_{i'}$ and $\ell_k == \ell_{k'}$. Let $R$ be the set of such $\ell$-ring. For any $\ell$-ring $\ell_r \not \in R$, their number is even in $C$.  
             
             By our algorithm, using the rigidity of $C$, one robot on either $\ell_s$ or $\ell_{s'}$ moves to its adjacent empty node on its current $\ell$-ring. Assume without loss of generality that this robot was selected from $\ell_s$.  Let us refer to the reached configuration after such a move by $C''$. We first show by contradiction that $C''$ is rigid. Note that, because of $\ell_s$ and $\ell_{s'}$, we are sure that there is no axes of symmetry that is perpendicular to $\ell_m$ in $C''$ and in addition, $\ell_m$ cannot be an axes of symmetry in $C''$. That is, the axes of symmetry in $C''$ lies on another $\ell$-ring that we denote by $\ell_{axes}$. First assume that, for any $\ell$-ring $\ell_r$ in $C$, $\ell_{axes} == \ell_r$ does not hold. That is, the number of $\ell$-ring $\ell_r \in R$ remains odd in $C''$. This means that $C''$ is rigid. A contradiction. Next, assume that for any $\ell$-ring $\ell_r \in R$, $\ell_{axes} != \ell_{r}$ but there exists an $\ell$-ring $\ell_r'$ in $C$ such that $\ell_r == \ell_{axes}$. As the number of such $\ell$-rings was even in $C$, their number becomes odd in $C''$. Similarly to the previous case, the number of $\ell$-ring $R$ remains also odd. Hence, $C''$ is rigid. A contradiction. Lastly, assume that for an $\ell$-ring $\ell_r \in R$ $\ell_r == \ell_{axes}$. Let $\ell_{i'}$ and $\ell_{k'}$ be the two $\ell$-rings adjacent to $\ell_{axes}$. Assume without loss of generality that $\ell_i == \ell_{i'}$. Let $r'$ be the robot on $\ell_{i'}$ such that $r' == r$. As $C''$ is symmetric and $\ell_{axes}$ is on the axes of symmetry, the node that is on $\ell_{k'}$ which is on the same $L$-ring as $r'$ is occupied. A contradiction. We can thus deduce that $C''$ is rigid.  
             
             Observe that, in the new configuration $C''$, another robot might become the closest to an empty node on a maximal $\ell$-ring with a the largest view, if by moving, the configuration is rigid then the robot can move, otherwise, we are sure that there is at least one robot in $R(C'')$ which is $r$ that can move and join $u$ without reaching a symmetric configuration (recall that the priority is to choose a robot that keeps the configuration rigid). Now, assume by contrast that the axes of symmetry is perpendicular to the unique maximal $\ell$-ring in $C'$. The axes of symmetry in this case cannot be on the same $L$-ring as $u$ otherwise $C$ is also symmetric. Since $C'$ is symmetric then each $\ell$-ring in the configuration, except maybe for target-$\ell$ and the $\ell$-ring hosting $r$, has at least one axes of symmetry in $C$ (considering the $\ell$-rings individually). Moreover, they all share at least one axes of symmetry.  By our algorithm, if there is an $\ell$-ring in $T(C)$ (the set of all none empty $\ell$-rings except for target-$\ell$) that does not contain exactly two $1$.blocks at distance $2$ from each other from each side then using the rigidity of $C$, one robot on such an $\ell$-ring is elected to move. The elected robot has to be the closest to the largest $1$.block on its current $\ell$-ring. By $r'$ and $\ell$, let us refer to respectively the elected robot to move and its current $\ell$-ring in $C$. If by moving, $r'$ does not join the biggest $1$.block then $r'$ is the only closest robot to a largest $1$.block and hence $\ell$ does not contain any axes of symmetry. That is, the reached configuration is rigid. Moreover, $r$ can join $u$ without reaching a symmetric configuration.  By contrast, if $r'$ joins the largest $1$.block, we show in the following that $C'$ cannot be symmetric. Assume by contradiction, that $C'$, is symmetric. Two cases are possible:
            \begin{enumerate}
                \item the $1$.block that $r'$ joined is on the axes of symmetry in $C'$. In this case, let $d1$ and $d2$ be two $1$.blocks of $C$ that are symmetric with respect to the axes of symmetry in $C'$ and which are the closest to it. In this case, $r'$ was part either of $d1$ or $d2$ in $C$. Assume without loss of generality that $r'$ was part of $d1$. Recall that by our algorithm, the size of $d1$ is less or equal to the size of $d2$.  When $r'$ moves, the size of $d1$ has decreased by one while the size of $d2$ did not change. That is, there is no axes of symmetry in $\ell$ (recall that the unique largest $1$.block has to be on the axes of symmetry). Contradiction. 
                
                \item Otherwise. Let $B$ be the size of the largest $1$.block in $C$. By our algorithm, the robot $r'$ to move is the one that is in the smallest $1$.block in $C$ which is the closest to a $1$.block of size $B$. When $r'$ moves, if $r'$ do not join a $1$.block of size $B$ then $r'$ is the only robot that is the closest to a $1$.block of size $B$ and hence the configuration reached is rigid. By contrast, if $r'$ joins a $1$.block of size $B$  then in the configuration reached $C''$, there is a unique largest $1$.block on $\ell$. Its size is $B+1$. By  \textit{Biggest} (respectively \textit{Block}) the largest $1$.block in $C''$ (respectively to $1$.block to which $r'$ belonged in $C$). If \textit{Block} contained more than one robot in $C$ then \textit{Block} still exists in $C''$ and its size decreased by one. As by assumption $C''$ is also symmetry, \textit{Biggest} is on the axes of symmetry. That is, \textit{Block} is symmetric to another $1$.block which shares a hole with \textit{Biggest}. This means that the size of this block is smaller then the size of \textit{Block} in $C$. Hence $r'$ was not enabled to move. Contradiction. Now, let us consider the case in which the size of \textit{Block} is equal to $1$ in $C$ ($r'$ is an isolated robot). When $r'$ moves, as $C'$ is supposed to be symmetric, the number of $1$.blocks of size $B$ with an isolated robot at distance $2$ is odd. Hence the configuration is rigid.  Contradiction.  
            \end{enumerate}
            
           Finally, if all $\ell$-rings in $T(C)$ contains only two $1$.blocks separated by exactly one empty node on each side then they either share exactly one axes of symmetry or all the axes (Observe that the axes of symmetry either crosses an $\ell$-ring on an empty node or in the middle of a $1$.block.  By our algorithm, using the rigidity of $C$, one robot at the border of the smallest $1$.block moves outside the $1$.block it belongs to. By doing so, it joins the largest $1$.block creating a new axes of symmetry which not shared by the other $\ell$-ring in $T(C)$. Thus the configuration that is reached is rigid. moreover, when $r$ moves to join $u$, the configuration remains rigid.  
           
        \end{enumerate}
        
        From the cases above, we can deduce that the lemma holds.

\end{proof}

From Lemmas \ref{lem:uniquel} and \ref{lem:uniquelessl} we can deduce the following Corollaries :

\begin{corollary}\label{cor:uniqueC}
Starting from a rigid configuration $C$ in which $Unique(C)$ is false, a rigid configuration $C'$ in which $Unique(C')$ is true, is  eventually reached. 
\end{corollary}

\begin{corollary}\label{cor:Unique_Multi}
Let $C$ be the first configuration in which Unique($C$) is true. Then, for any $\ell_i$ with $i\in \{1, 2, ..., \ell\}$, there is at most one multiplicity node. Moreover, if $\ell_i$ hosts a multiplicity node then the three following conditions hold: (1) $nb_{\ell_i}(C')= \ell-1$, (2) the multiplicity node hosts exactly two robots and (3) the multiplicity node is at a border of a $1$.block of size $\ell-1$.      
\end{corollary}

We now focus on configurations $C \in \mathcal{C}_{p_1}$ in which $Unique(C)$ is true and show that eventually a configuration $C' \in \mathcal{C}_{p_2}$ is reached.


\begin{lemma} \label{lem:undefined-1}
Let  $C \in \mathcal{C}_{Undefined}$ be a rigid configuration in which the following conditions hold:
\begin{itemize}
    \item $C$ does not contain outdated robots
    \item $nb_{\ell_i}(C) < nb_{\ell_k}(C)$
    \item  $\Gamma(C)$ is either rigid or contains a single axes of symmetry
\end{itemize} 
From $C$, a configuration $C' \in \mathcal{C}_{oriented }$ with no outdated robots is eventually reached. 
\end{lemma}

\begin{proof}
From corollary \ref{cor:uniqueC}, we know that $C$ do not contain an outdated robot. We consider the following cases:
\begin{enumerate}
    \item $\Gamma(C)$ is rigid. As by our algorithm, no robot on an $\ell$-ring $\ell_r \neq \ell_i$ moves, $\Gamma(C)$ remains rigid. That is, a unique node on $\ell_i$ can be uniquely identified. Let us refer to such a node by $u$. By our algorithm, robots on $\ell_i$ which are the closest to $u$ are the ones to move. their destination is their adjacent empty node on their current $\ell$-ring toward $u$. As the robots join $u$ one by one, the number of robots on $\ell_i$ decreases to eventually be equal to $1$. That is, eventually a configuration $C'$ in which $nb_{\ell_i}(C')=1$ is reached. Observe that once the robots join $u$ they are not allowed to move anymore. That is in $C'$ there are no outdated robots. Hence the lemma holds in this case.
    \item $\Gamma(C)$ is symmetric. First note that by definition, whatever the position of the robots on $\ell_i$, $\Gamma(C)$ remains symmetric with the same axes of symmetry. Depending on how the axes of symmetry crosses $\ell_i$, we consider the following cases: 
    \begin{enumerate}
        \item The axes of symmetry of $\Gamma(C)$ crosses $\ell_i$ on a single node ($\Gamma(C)$ is node-edge symmetric). Let us refer to this node by $u$. By our algorithm, $u$ is identified as a target node for all robots on $\ell_i$. The closest ones first move to join $u$ taking the empty path. That is, eventually all robots join $u$ and a configuration $C'$ in which  which $nb_{\ell_i}(C')=1$ is reached.
        \item The axes of symmetry crosses $\ell_i$ on two nodes ($\Gamma(C)$ is node-node symmetric). Let us refer to these nodes by $u$ and $u'$ respectively. By our algorithm, using the rigidity of $C$, one robot on either $u$ or $u'$ is elected to move. It destination is its adjacent node on its current $\ell$-ring. Since $nb_{\ell_i}(C) < nb_{\ell_k}(C)$, we are sure that neither $u$ nor $u'$ contains a tower (recall that a tower is only created to reduce the number of maximal $\ell$-rings when each of their nodes is occupied). As for case in which the axes of symmetry crosses a single node on $\ell_i$, a configuration $C'$ in which $nb_{\ell_i}(C')=1$ is reached. 
        \item The axes of symmetry crosses $\ell_i$ on edges ($\Gamma(C)$ is edge-edge symmetric). Assume without loss of generality that the axes of symmetry crosses $\ell_i$ on the edges $e_1=(u_1,u_2)$ and $e_2=(u_3,u_4)$ and that $u_1$ and $u_3$ are on the same side of the axes of symmetry. Let $U=\left\{u_j, j \in \{1,\ldots, 4\}\right\}$. The following cases are distinguished by our algorithm: 
        \begin{enumerate}
            \item $\forall u \in U$, $u$ is occupied. By our algorithm, using the rigidity of $C$, one of these nodes is elected. Assume without loss of generality that $u_1$ is elected. Robots on $u_1$ are the ones to move. their destination is $u_2$. Note that if $u_1$ hosts more than one robot and that the scheduler activates only a subset of robots then robots on $u_1$ remain the only ones allowed to move. That is, eventually, $u_1$ becomes empty and no robot has an outdated view. Also, thanks to both $\ell_i$ and $\ell_{\mathit{max}}$, the configuration reached is rigid. We retrieve the case \ref{case:3occupied}.
            
            \item \label{case:3occupied} Three nodes of $U$ are occupied. Assume without loss of generality that $u_1$ is empty. By our algorithm, robots on $\ell_i$ which are on on the side of $u_3$ and $u_1$ and are the closest to $u_3$ are the ones allowed to move. Their destination is their adjacent node toward $u_3$ taking the shortest path. As long as there are robots between $u_3$ and $u_1$, these robots are the only one allowed to move. That is, eventually, the side between $u_3$ and $u_1$ becomes empty. Once such a configuration is reached, by our algorithm, robots on $u_2$ are now the ones allowed to move. Their destination is their adjacent node in the opposite direction of $u_1$. Observe that as long as $u_2$ is occupied, robots on $u_2$ remain the only ones allowed to move. That is, eventually, $u_2$ becomes empty and we retrieve Case \ref{case:2occupied} with a rigid configuration (there is at least one occupied node on the side of $u_4$ and $u_2$ and no robots on the other side) in which $u_3$ and $u_4$ are occupied. We retrieve Case \ref{case:2occupied}.
            
            \item \label{case:2occupied} Two nodes of $U$ are occupied. Several scenarios are considered depending on the nodes that are occupied: first \textbf{($i$)} assume without loss of generality that $u_1$ and $u_2$ are occupied (the case in which the two nodes are neighbors). 
            By our algorithm, if there is no other occupied node on $\ell_i$ then using the rigidity of $C$, a node of $U$ is elected. Assume without loss of generality that this node is $u_1$. Robots on $u_1$ are the ones allowed to move. Their destination is $u_2$. Note that as long as there are robots on $u_1$, robots on $u_1$ remain the only ones allowed to move. That is, eventually, $u_1$ becomes empty. Thanks to $\ell_i$, we can deduce that the configuration reached is rigid. That is, we retrieve Case \ref{case:1occupied}. By contrast, if all robots which are not on a node of $U$ are on the same side of the axes of symmetry of $\Gamma(C)$ (assume without loss of generality that they are on the same side as $u_1$) then, robots on $u_2$ are the ones allowed to move. Their destination is $u_1$. Observe that as long as there are robots on $u_2$, these robots are the only ones allowed to move. That is, eventually $u_2$ becomes empty and no robot has an outdated view. Observe that since robots are on only one side of the axes of symmetry of $\Gamma(C)$ the configuration reached is rigid. Thus, we retrieve Case \ref{case:1occupied}. By contrast, if robots are on both sides of the axes of symmetry then, by our algorithm, one robot remains idle to play the role of a landmark. This robot is the one that is on $\ell_i$ which the farthest from the an occupied node of $U$ being on the same side of the axes of symmetry of $\Gamma(C)$. That is, if there are two robots that are the farthest to an occupied node of $U$ (there is one robot on each side) then, since $C$ is rigid, one robot is elected to move. It destination is its adjacent node towards the occupied node of $U$ being on the same side. By doing so, a configuration in which there is only one occupied node, $u$, which is the farthest from the occupied node of $U$ is reached. This robot (occupied node) is the landmark.  Assume without loss of generality that $u$ is on the same side as $u_1$. By our algorithm, robots that are on the opposite side of $u_1$ which are the closest to $u_2$ moves on their current $\ell$-ring to join $u_2$. As robots join $u$ one by one, eventually, a configuration in which robots occupy only one side of the axes of symmetry is reached. This case was discussed earlier.  Hence, we can deduce that we retrieve Case \ref{case:1occupied}. 
            
            Next, \textbf{($ii$)} let us now consider the case in which $u_1$ and $u_3$ are occupied (the case in which the two nodes of $u$ are on the same side of the axes of symmetry). By our algorithm, robots on the node of $U$ with the largest view are the ones allowed to move. Since $C$ is rigid, robots on exactly one node of $U$ are elected to move. Assume without loss of generality that robots on $u_1$ are the ones allowed to move. Their destination is their adjacent node in the opposite direction of $u_2$. Observe that as long as $u_1$ is occupied, robots on $u_1$ are the only ones allowed to move. That is, eventually, $u_1$ becomes empty and no robot has an outdated view. We thus, retrieve Case \ref{case:1occupied}. 
            
            Finally, ($iii$) assume without loss of generality that $u_1$ and $u_4$ are occupied (the case when the two nodes of $U$ are not neighbors and at different sides of the axes of symmetry of $\Gamma(C)$. By our algorithm, robots with the largest view on either $u_1$ or $u_4$ are allowed to move. Their destination is their adjacent node on $\ell_i$ which is not in the set $U$. Assume without loss of generality that robots on $u_1$ are the ones that are elected. Note that as long as $u_1$ is occupied, robots on $u_1$ remain the only ones to move. That is, eventually $u_1$ becomes empty. Thanks to $\ell_i$ the reached configuration is rigid. Thus, we retrieve Case \ref{case:1occupied}.
            
            \item \label{case:1occupied} Exactly one node of $U$ are occupied. Assume without loss of generality that $u_1$ is the node of $U$ which is occupied. By our algorithm, if all robots on $\ell_i$ are on the same side as $u_1$, then the closest robots on $\ell_i$ to $u_1$ moves to their adjacent node on their current $\ell$-ring towards $u_1$. That is eventually, all robots on $\ell_i$ join $u_1$ and no robot has an outdated view. Hence, a configuration $C' \in \mathcal{C}_{oriented}$ is eventually reached. By contrast, if all robots on $\ell_i$ are on the opposite side of $u_1$ then by our algorithm robots on $u_1$ are the ones to move. Their destination is their adjacent node of $U$. That is either we remain in Case  \ref{case:1occupied} with $u_2$ the only node of $U$ which is occupied and all robots on $\ell_i$ being on the same side as $u_2$ or a configuration in which there are exactly two nodes of $U$ ($u_1$ and $u_2$) which are occupied and neighbors. Note that in the later case, robots on $u_1$ remain the only ones allowed to move as there is no robot on the same side as $u_1$. Their destination remains also the same (they move toward $u_2$). Hence we can deduce that a configuration $C' \in \mathcal{C}_{oriented }$ is eventually reached. Finally, if there are robots on both sides of the axes of symmetry of $\Gamma(C)$ then by our algorithm, robots that are on the same side as $u_1$ which are the closest to $u_1$ are the ones allowed to move. Their destination is their adjacent node on their current $\ell$-ring towards $u_1$. That is eventually, the side of $u_1$ becomes empty and as discussed previously a configuration in which all robots on $\ell_i$ join the same node of $U$ is eventually reached. 
            
            \item $\forall u \in U$, $u$ is empty. By our algorithm, the closest robot on $\ell_i$ to a node of $U$ being on the same side of the axes of symmetry is the one allowed to move. Its destination is its adjacent node toward a node of $U$. In the case in which there are more than one such robot, as $C$ is rigid, robots on exactly one node are elected to move. that is, we retrieve Case \ref{case:1occupied}. 
        \end{enumerate}
    \end{enumerate}
       From the case above, we can deduce that eventually, all robots on $\ell_i$ gather on a unique node of $U$. Moreover, no robot has an outdated view as they have all moved and once they are on the same node, they are no more allowed to move. Hence, a configuration $C' \in \mathcal{C}_{oriented}$ with no outdated robots is eventually reached and the lemma holds. 
\end{enumerate}
\end{proof}

\begin{lemma}\label{lem:undefined-2}
Let  $C \in \mathcal{C}_{Undefined}$ be a rigid configuration in which the following conditions hold:
\begin{itemize}
    \item $C$ has no outdated robots.
    \item $nb_{\ell_i}(C) = nb_{\ell_k}(C)$.
    \item $\Gamma(C)$ is either rigid or contains a single axes of symmetry.
\end{itemize}

From $C$, a configuration $C' \in \mathcal{C}_{oriented }$ with no outdated robots is eventually reached. 
\end{lemma}

\begin{proof}

By Corollary \ref{cor:Unique_Multi}, we know that an $\ell$-ring $\ell'$ can host a multiplicity in $C$ only if $nb_{\ell'}(C)=\ell-1$. Moreover, the multiplicity is adjacent to an empty node in $\ell'$. That is, whenever $nb_{\ell'}(C)<\ell-1$, we are sure that each node of $\ell'$ hosts one robot. Hence, by moving, no new maximal $\ell$-ring is created. Depending on $\Gamma(C)$, robots behavior is different. Our algorithm distinguish the following cases:
\begin{enumerate}
    \item $\Gamma(C)$ is rigid. Using the rigidity of $\Gamma(C)$, a unique node, $u$ is selected from either $\ell_i$ or $\ell_k$. Assume without loss of generality that the selected node $u$ is on $\ell_i$. Note that since $C$ does not contain any outdated robots and since no robot of $\Gamma(C)$ is allowed to move, $u$ keeps being identified. By our algorithm, two cases are possible: $(i)$ $u$ is empty and $nb_{\ell_i}(C)=\ell-1$ and $(ii)$ all the other cases. In case ($i$), by our algorithm, the robot $r$ on $\ell_i$ that is adjacent to $u$ which is not part of a multiplicity moves to $u$. By Corollary \ref{cor:Unique_Multi}, we are sure that such a robot exists. By moving $u$ becomes occupied. Note that since $r$ was note part of a multiplicity, by moving no new maximal $\ell$-ring is created. That is, we retrieve Case $(ii)$. In case $(ii)$, by our algorithms, robots on $\ell_i$ move to join the selected node starting from the closest ones.  That is, eventually, one robot will join $u$ and a configuration $C'$ in which $nb_{\ell_i}(C') < nb_{\ell_k}(C')$ is reached. Robots on $u$ are not allowed to move anymore. If $C'$ is rigid and does not contain outdated robots then by Lemma \ref{lem:undefined-1}, we can deduce that the lemma holds. Otherwise, as $u$ is identified in a unique manner, the target node of the robots on $\ell_i$ remains node $u$. That is eventually, all robots on $\ell_i$ join $u$ and the lemma holds in this case.  
    
    \item $\Gamma(C)$ is node-edge symmetric.  Let $u$ (respectively $u'$) be the node on $\ell_i$ (respectively $\ell_k$) which is on the axes of symmetry of $\Gamma(C)$. Assume without loss of generality that $u$ is occupied while $u'$ is empty. By our algorithm, the robot that is on $\ell_i$ which is the closest to $u$ is the one allowed to move. If there are more than one robot, using the rigidity of $C$ only one robot is elected to move. By moving, either it becomes the only closest robot to $u$ or it joins $u$. In the first case, by our algorithm the robot that has moved remains the only one allowed to move. Its destination is its adjacent empty node toward $u$. That is eventually it joins $u$. As there is only one robot that was allowed to move. When the robot joins $u$, the configuration does not contain any outdated robot. In the reached configuration $C'$, $nb_{\ell_i}(C') < nb_{\ell_k}(C')$ holds. By our algorithm, robots on $u$ are not allowed to move anymore. That is, if $C'$ is rigid then by Lemma \ref{lem:undefined-1} we can deduce that the lemma holds. By contrast if $C'$ is symmetric, then as no robot moved from $\Gamma(C)$, the only possible axes of symmetry in $C'$ is the one that lies on the axes of $\Gamma(C)$. By our algorithm, in this case, robots on $\ell_i$ move to join $u$ taking the shortest path and starting from the closest robots. That is, eventually, all robots on $\ell_i$ are located on $u$. A configuration $C'' \in \mathcal{C}_{oriented }$ is reached. As robots are not allowed to move anymore once on $u$, $C''$ does not contain outdated robots. Hence, the lemma holds. Now assume that both $u$ and $u'$ are occupied. Let $R_i$ (respectively $R_k$) be the closest robot to $u$ (respectively $u'$) on $\ell_i$ (respectively $\ell_k$). By our algorithm, since $C$ is rigid, exactly one robot is selected from either $R_i$ or $R_k$. Assume without loss of generality that the selected robot is in $R_i$. This robot is the one allowed to move. Its destination is its adjacent empty on its current $\ell$-ring toward $u$. By moving, the robot either joins $u$ and we retrieve the case in which $nb_{\ell_i}(C') < nb_{\ell_k}(C')$ discussed earlier or the robot remains the only closest robot to $u$ and hence it is the only one allowed to move. That is, eventually, the robot joins $u$ and we can deduce that the lemma holds. Finally, assume that both $u$ and $u'$ are empty. Observe that in this case $nb_{\ell_i}(C) \ne \ell-1$ (otherwise $C$ is symmetric and not rigid). That is, by Corollary \ref{cor:Unique_Multi} neither $\ell_i$ nor $\ell_k$ contain a multiplicity. Let $R_i$ (respectively $R_k$) be the closest robot to $u$ (respectively $u'$) on $\ell_i$ (respectively $\ell_k$). By our algorithm, since $C$ is rigid, exactly one robot is selected from either $R_i$ or $R_k$. Assume without loss of generality that the selected robot is in $R_i$. This robot is the one allowed to move. Its destination is its adjacent empty node toward $u$. By moving, the robot either joins $u$ or it becomes the only robot allowed to move (as it is the only closest robot to $u$). That is eventually, the robot joins $u$. If the reached configuration is rigid then we retrieve the previous discussed case. Otherwise, thanks to the unique $\ell_{max}$, both $\ell_i$ and $\ell_k$ can still be identify, the same for $u$ as no robot from $\Gamma(C)$ has moved. By our algorithm, the closest robots to $u$ on $\ell_i$ are the ones allowed to move, their destination is their adjacent node towards $u$. That is, all robots eventually move to join $u$. Since they are no more allowed to move, we can deduce that the lemma holds in this case. 
    
    \item $\Gamma(C)$ is node-node symmetric. Let $u_i$ and $u'_i$ (respectively $u_k$ and $u'_k$) be the two nodes on $\ell_i$ (respectively $\ell_k$) on which the axes of symmetry of $\Gamma(C)$ passes through. We consider the following two cases:
    \begin{enumerate}
        \item $\forall u \in \{u_i, u'_i, u_k, u'_k\}$, $u$ is occupied. Let $U \subseteq \{u_i, u'_i, u_k, u'_k\}$ be the set of nodes that have an occupied adjacent node on their $\ell$-ring. By our algorithm, if $|U|\geq 1$ then by our algorithm, using the rigidity of $C$ a single node in $U$ is elected. Assume without loss of generality that this node is $u_i$. Robots on $u_i$ are the ones allowed to move. Their destination is their adjacent occupied node on their $\ell$-ring. Observe that if $u_i$ hosts a multiplicity then by Corollary \ref{cor:Unique_Multi}, there are only $2$ robots part of the multiplicity. If the scheduler activates only one robot, as the destination node is occupied, the robot on $u_i$ that were supposed to move remains the only one allowed to move. Its destination remains the same. That is,  by moving, we reach a configuration in which $\exists u \in  \{u_i, u'_i, u_k, u'_k\}$ such that $u$ is empty. We retrieve Case \ref{Case:UaxesEmpty}. By contrast, if  $|U| = 0$, by our algorithm, using the rigidity of $C$, one node of $ u \in  \{u_i, u'_i, u_k, u'_k\}$ is elected. The robot on $u$ node moves to one of its adjacent occupied node (the scheduler chooses the direction to take). Assume without loss of generality that $u=u'_i$. From Corollary \ref{cor:Unique_Multi}, we are sure that $u$ does not contain a multiplicity (A multiplicity exists only in the case where $nb_{\ell_i}(C)=\ell-1$). By moving, $u'_i$ becomes empty and we retrieve Case \ref{Case:UaxesEmpty}.

        \item \label{Case:UaxesEmpty} $\exists u \in  \{u_i, u'_i, u_k, u'_k\}$ such that $u$ is empty. If there is a unique node $u \in  \{u_i, u'_i, u_k, u'_k\}$ which is occupied then the closest robot to $u$ on the same $\ell$-ring as $u$ moves to their adjacent empty node toward $u$. Assume without loss of generality that $u=u_i$. Once one robot joins $u_i$, a configuration $C' \in \mathcal{C}_{Undefined}$ with $nb_{\ell_i}(C') < b_{\ell_k}(C')$ and no outdated robots is reached.  By contrast, let us consider now the case in which there are three node in  $\{u_i, u'_i, u_k, u'_k\}$ which are occupied. Assume without loss of generality that $u'_i$ is the empty node. By our algorithm, robots on $\ell_i$ which are the closest to $u_i$ are the only one allowed to move. By moving they keep being the closest ones. That is eventually, one robot on $\ell_i$ joins $u_i$, a configuration  $C' \in \mathcal{C}_{Undefined}$ with $nb_{\ell_i}(C') < b_{\ell_k}(C')$ is then reached. Finally, if there are two nodes of $ \{u_i, u'_i, u_k, u'_k\}$ which are occupied then first note that $nb_{\ell_i}(C)<\ell-1$ (otherwise three nodes of the set $\{u_i, u'_i, u_k, u'_k\}$ should have been occupied). That is, by Corollary \ref{cor:Unique_Multi}, neither $\ell_i$ nor $\ell_k$ contain a multiplicity. If the two occupied nodes of $\{u_i, u'_i, u_k, u'_k\}$ are part of the same $\ell$-ring (assume without loss of generality that these two nodes are $u_i$ and $u'_i$) then by our algorithm, using the rigidity of $C$ one node among $u_i$ and $u'_i$ is elected. The robot on the elected node is the one allowed to move. Its destination is its adjacent node on its current $\ell$-ring. By moving a configuration $C'$ in which $nb_{\ell_i}(C') < b_{\ell_k}(C')$ is reached. By contrast, if the two occupied nodes are on two different $\ell$-rings then let $R(C)$ be the set of robots on $\ell_i$ and $\ell_k$ which are the closest to the occupied node of $\{u_i, u'_i, u_k, u'_k\}$ on their $\ell$-ring. One robot from $R(C)$ is elected to move. Its destination is its adjacent node on its current $\ell$-ring toward the closest occupied node of $\{u_i, u'_i, u_k, u'_k\}$ which is on its $\ell$-ring. By moving it either joins the node or becomes the unique closest robot. In the latter case, the robot remains the only one allowed to move. That is eventually,  a configuration  $C' \in \mathcal{C}_{Undefined}$ with $nb_{\ell_i}(C') < b_{\ell_k}(C')$ is eventually reached.  
      
       Once a configuration  $C' \in \mathcal{C}_{Undefined}$ in which $nb_{\ell_i}(C') < b_{\ell_k}(C')$ is reached then if $C'$ is rigid then by Lemma \ref{lem:undefined-1} we can deduce that the lemma holds. By contrast, if $C'$ is symmetric then the axes of symmetry of $C'$ lies on the axis of symmetry of $\Gamma(C)$. By our algorithm, the robots on $\ell_i$ which are the closest to $u_i$ are the ones allowed to move. Their destination is their adjacent empty node on their current $\ell$-ring. That is, eventually all robots on $\ell_i$ joins $u_i$. A configuration $C'' \in \mathcal{C}_{oriented}$ is then reached. As robot on $u_i$ are not allowed to move we can deduce that the lemma holds.
        
    \end{enumerate}
    
    \item $\Gamma(C)$ is edge-edge symmetric. Assume without loss of generality that: 
        \begin{itemize}
            \item the axes of symmetry of $\Gamma(C)$ passes through respectively $e_1=(u_1,u_2)$ and $e_2=(u_3,u_4)$ on $\ell_i$ and $e'_1=(u'_1,u'_2)$ and $e'_2=(u'_3,u'_4)$ on $\ell_k$.
            \item nodes $u_1$ and $u_3$ (respectively $u'_1$ and $u'_3$) are the same side of the axes of symmetry on $\ell_i$ (respectively $\ell_k$). 
        \end{itemize}
    
     Let $\mathcal{L}=\{\ell_i, \ell_k\}$, $U_i=\{u_1, u_2, u_3, u_4\}$, $U_k=\{u'_1, u'_2, u'_3, u'_4\}$ and  $U=U_i \cup U_k$. By our algorithm, a single $\ell$-ring is selected and robots that are on the selected $\ell$-ring behave exactly in the same manner as in the case in which $nb_{\ell_i}(C) \ne nb_{\ell_k}(C)$. To determine which $\ell$-ring to select, robots check if a given number of properties are verified in a given order: 
    
    \begin{enumerate}
        \item \label{case:ui=2Empty} There exists an $\ell$-ring in $\mathcal{L}$, let this $\ell$-ring be, without loss of generality $\ell_i$, in which $|U_i|=2$ with the two nodes of $U_i$ being neighbors to reach other and $Free(u_1,u_3) \wedge Free(u_1,u_3)$ holds. As $C$ is rigid and $\Gamma(C)$ is symmetric, we are sure that if $|U|_k=2$, the the two nodes of $U_k$ cannot be neighbors (Otherwise $C$ is symmetric). By our algorithm, $\ell_i$ is the one that is selected. The robot to move is the one that is on an occupied node of $U_i$ with the largest view. Its destination is adjacent occupied node. That is, the robot moves to an occupied node of $U_i$. By moving, a configuration $C'$ in which $nb_{\ell_i}(C')=1$ is reached. As no other robot is allowed to move and since all robots that were allowed to move have moved, we can deduce deduce that there are no outdated robots. Hence, we can deduce that the lemma holds.  
        
        \item \label{case:ui=1} There exists an $\ell$-ring in $\mathcal{L}$, let this $\ell$-ring be without loss of generality $\ell_i$ in which $|U_i|=1$ (Assume that $u_1$ is the occupied node of $U_i$), $Free(u_2,u_4) \wedge \neg Free(u_1,u_3)$ holds. If $\ell_k$ does not satisfy the same properties then $\ell_i$ is the one that is selected. The robot to move is the one that is on $\ell_i$ which the closest to $u_1$. Note that since $nb_{\ell_i}(C) \ne \ell -1$, by Corollary \ref{cor:Unique_Multi}, we are sure that the robot to move is not part of a multiplicity. That is, by moving, it either joins $u_1$ or becomes the only robot that is the closest to $u_1$. In the latter case, this robot remains the only one a allowed to move. Its destination is its adjacent node toward $u_1$. Hence, eventually it joins $u_1$. By joining $u_1$, a configuration $C'$ in which $nb_{\ell_i}(C') < nb_{\ell_k}(C')$. Since there is no other robot that was allowed to move, in $C'$, we are sure that there is no robot with an outdated view. If $nb_{\ell_i}(C') = 1$ then the lemma holds. By contrast, if $nb_{\ell_i}(C')>1$ then, since $\Gamma(C)=\Gamma(C')$ and there are robots only on one side of $\Gamma(C')$'s axes of symmetry, $C'$ is rigid. Moreover, robots on $u_1$ are not allowed to move anymore, the robots that are on $\ell_i$ which are the closest to $u_1$ are the ones to move. Their destination is their adjacent node on their $\ell$-ring toward $u_1$. That is, eventually all robots on $\ell_i$ joins $u_1$. Thus, lemma holds. 
        
        Now assume that $\ell_k$ verify the same properties as $\ell_i$. That is, $|U_k|=1$ (Assume without loss of generality that $u'_1$ is the occupied node of $U_k$), $Free(u'_2,u'_4) \wedge \neg Free(u'_1,u'_3)$. In this case, by our algorithm, the selection is done with respect to distance of a robot to the occupied node of $U$ located on its $\ell$-ring. More precisely, let $d_i$ (respectively $d_k$) be the smallest distance between a robot on $\ell_i$ (respectively $\ell_k$) from $u_1$ (respectively $u'_1$). Assume without loss of generality that $d_i < d_k$. In this case, the robot on $\ell_i$ which is the closest to $u_1$ is the one allowed to move (let us refer to this robot by $r$). The destination of $r$ is its adjacent node on its $\ell$-ring toward $u_1$. By moving, $r$ either joins $u_1$ or becomes the only closest robot to $u_1$. That is, in the configuration reached $d_i$ remains smaller than $d_k$ and hence $\ell_i$ keeps being the only $\ell$-ring to be selected. Note this holds until $r$ joins $u_1$. When $r$ joins $u_1$ as discussed previously, we can deduce that the lemma holds. Finally, in the case in which $d_i = d_k$ then the $\ell$-ring that hosts a robot which is at distance $d_i$ from the occupied node of $U$ being on its current $\ell$-ring, with the largest view is the one which is selected. Note that this is possible as $C$ is rigid. By moving, the robot either joins the occupied node of $U$ or becomes the only closest robot to an occupied node of $U$ which has been already discussed. Thus, We can deduce that the Lemma holds. 
        
        \item \label{case:ui=2-neighbor} Next, there exists an $\ell$-ring in $\mathcal{L}$, let this $\ell$-ring be without loss of generality $\ell_i$ in which $|U_i|=2$ and the two occupied nodes of $U_i$ are neighbors (let these two nodes be respectively $u_1$ and $u_2$), $Free(u_1,u_3) \vee Free(u_2,u_4)$ holds. Note that both $Free(u_1,u_3) \wedge Free(u_2, u_4)$ does not hold otherwise we are in Case \ref{case:ui=2Empty}. By our algorithm, if $\ell_k$ does not satisfy the same properties then, $\ell_i$ is selected. Otherwise, let the two occupied nodes of $U_k$ be respectively $u'_1$ and $u'_2$ and let $r$ and $r'$ be the robots located on respectively $u_2$ and $u'_2$. By our algorithm, $view_{r}(t)(1) < view_{r'}(t)(1)$ then $\ell_i$ is selected. Otherwise $\ell_k$ is selected (Recall that since $C$ is rigid, we are sure that $view_{r}(t)(1) \ne view_{r'}(t)(1)$). Assume without loss of generality that $\ell_i$ is selected. By our algorithm, if $Free(u_1,u_3)$ holds then the robot on $u_1$ is the one to move. Its destination is $u_2$. By contrast, if $Free(u_2,u_4)$, then the robot on $u_2$ is the the one allowed to move. Its destination is $u_1$. By moving, $u_2$ becomes empty. As $u_1$ was occupied, in the configuration reached $C'$, $nb_{\ell_i}(C') < nb_{\ell_k}(C)$. Moreover, as there is one side of $\ell_i$ which is occupied, we are sure that the configuration reached $C'$ is rigid. Robots on $u_1$ are no more allowed to move as robots being on the same side as $u_1$ are the ones that move, starting from the ones that are the closest to $u_1$, toward $u_1$. That is, eventually, all robots on $\ell_i$ join $u_1$ and we can deduce that the lemma holds. 
        
        \item There exists an $\ell$-ring in $\mathcal{L}$, let this $\ell$-ring be, without loss of generality $\ell_i$, in which $|U_i|=2$ and the two occupied nodes of $U_i$ are on the same side of $\Gamma(C)$'s axes of symmetry (let these robots be respectively $u_1$ and $u_3$), $Free(u_2, u_4)$ and $\neg Free(u_1,u_3)$. By our algorithm, if $\ell_k$ does not satisfy the same properties then, $\ell_i$ is selected. Otherwise, the $\ell$-ring that hosts a robot of $U$ with the largest view is the one to be elected (note that this is possible as $C$ is rigid). Assume without loss of generality that $\ell_i$ is elected. Let $r$ and $r'$ be the robots on respectively $u_1$ and $u_3$. If $view_{r}(t)(1) < view_{r'}(t)(1)$ then robot $r'$ is the one allowed to move. Its destination is its adjacent node on its $\ell$-ring in the opposite direction of $u_4$. By moving $u_3$ becomes empty. If $r'$ moved to an empty node, as $Free(u_2, u_4)$ holds, we retrieve Case \ref{case:ui=1}. Otherwise ($r'$ moves to an occupied node), a configuration $C'$ in which $nb_{\ell_i}(C')<nb_{\ell_k}(C')$ is reached. By our algorithm, robots that are the closest to $u_1$, are the ones allowed to move. Their destination is their adjacent node toward $u_1$. Observe that if $r'$ is the closest to $u_1$, as it is part of a multiplicity of size $2$, by moving, we can reach alternatively a configuration $C'$ in which $nb_{\ell_i}(C')<nb_{\ell_k}(C')$ and a configuration $C''$ in which $nb_{\ell_i}(C'')=nb_{\ell_k}(C'')$. However, since in both configurations, robots keep the same destination (recall that $|U_i|=1$ and $Free(u_2, u_4)$ holds) then eventually, they will join $u_1$ and we can deduce that the lemma holds. 
        
        \item \label{case:ui=3} there exists an $\ell$-ring in $\mathcal{L}$, let this $\ell$-ring be without loss of generality $\ell_i$ in which $|U_i|=3$ (let $u_1$ be the unique empty node in $U_i$) and $Free(u_1,u_3)$. By our algorithm, if $\ell_k$ does not satisfy the same properties then, $\ell_i$ is elected. Otherwise, assume loss of generality that the empty node on $\ell_k$ is $u'_1$. Let $r$ and $r'$ be the two robots located on respectively $u_2$ and $u'_2$. By our algorithm, if $view_{r}(t)(1) > view_{r'}(t)(1)$ then, $\ell_i$ is elected (recall that $C$ is rigid and hence we are sure that $view_{r}(t)(1) \ne view_{r'}(t)(1)$). Otherwise $\ell_k$ is selected. Assume without loss of generality that $\ell_i$ is the elected $\ell$-ring. The robot, $r$ on $u_2$ is the one allowed to move. Its destination is its adjacent node in the opposite direction of $u_1$. If $r$ moves to an empty node then, since $Free(u_1,u_3)$ holds, we retrieve Case \ref{case:ui=2-neighbor}. By contrast, if it moves to an occupied node then, a configuration $C'$ in which $nb_{\ell_i}(C')<nb_{\ell_k}(C')$ is reached. Again, as $Free(u_1,u_3)$ holds, the robot on $u_3$ is the one that is allowed to move. Its destination is $u_2$. By moving, a configuration $C''$ in which the only node of $U_i$ to be occupied is $u_4$. Moreover, $nb_{\ell_i}(C'')<nb_{\ell_k}(C'')-1$. That is, when $r$ moves toward $u_4$, as it is part of a multiplicity of size $2$, the scheduler can break the multiplicity. However, as $nb_{\ell_i}(C'')<nb_{\ell_k}(C'')-1$, $\ell_i$ remains the one to be elected and $u_4$ remains the target of the robots on $\ell_i$. That is, eventually, all robots join $u_4$ and the lemma holds. 
        
        \item Non of the cases above is verified. Assume without loss of generality that $|U_i| \leq |U_k|$. The following cases are possible: 
            \begin{enumerate}
                \item $|U_i|=4$. Observe that in this case $|U|_i = |U|_k$. Let $R$ be the set of robots on a node of $U$. The $\ell$-ring that hosts the robot of $R$ with the largest view is the elected $\ell$-ring (Recall that this is possible as $C$ is rigid). Assume without loss of generality that $\ell_i$ is selected. By our algorithm, using the rigidity of $C$, a unique robot, say $r$, on a node of $U_i$ is selected to move. Assume without loss of generality that $r$ is located on $u_1$. The destination of $r$ is $u_2$. By Corollary \ref{cor:Unique_Multi}, $u_1$ hosts no more than $2$ robots in $C$. If $u_1$ hosts a multiplicity and the scheduler activates only one robot on $u_1$, as the robot moves to an occupied node, the configuration remains the same and hence, the robot that was supposed to move is the only one allowed to move. That is eventually, $u_1$ becomes empty and $u_2$ hosts at least $2$ and at most $3$ robots ($3$ robots in the case in which $nb_{ell_i}(C)=\ell-1$. In the configuration reached $C'$, $nb_{\ell_i}(C') < nb_{\ell_k}(C')$. 
                
                Let us first discuss the case in which $u_2$ hosts a multiplicity of size $3$. Since $Free(u_1,u_3)$ does not hold (recall that $nb_{ell_i}(C)=\ell-1$). By our algorithm, the robot that is on the same side as $u_1$ which is the closest to $u_3$ is the one allowed to move. Its destination is its adjacent node toward $u_3$. That is, eventually, $Free(u_1,u_3)$ becomes true. Robots on $u_2$ are then allowed to move as we retrieve Case \ref{case:ui=3}. Note that, in every configuration reached $C'$, $nb_{\ell_i}(C')<nb_{\ell_k}(C')$ (even if the scheduler activates only a subset of robots, as by moving they add at most one occupied node) and hence $\ell_i$ is uniquely identified. When $u_2$ becomes empty, robots that are on $u_3$ are the one toward $u_4$. As long as $u_4$ is occupied, the configuration remains the same and hence only robots on $u_3$ are allowed to move. Eventually, $u_3$ becomes empty and in the configuration reached $|U_i|=1$ and $Free(u_1,u_3)$ holds. As discussed previously, all robots on $\ell_i$ move to join $u_4$ starting from the closest ones. Eventually, a configuration $C''$ in which $nb_{\ell_i}(C'')=1$ is reached. Hence the lemma holds.  
                
                Lastly, let us consider the case in which $u_2$ hosts two robots. If $Free(u_1,u_3)$ does not hold then robots located on $\ell_i$ which are on the same side as $u_1$ move to join $u_3$ starting from the closest one. That is, as for the case in which $u_2$ hosts a multiplicity of size $3$, $\ell_i$ is uniquely identified as the number of its occupied nodes is strictly smaller than the number of occupied nodes on $\ell_k$. By contrast, if $Free(u_1,u_3)$ holds, then since in the configuration $|U|_k=4$ and $|U|_i=3$ (recall that $u_1$ became empty), $\ell_i$ is also uniquely identified. As discussed in Case \ref{case:ui=3}, we can deduce that the lemma holds.  
                
                \item $|U_i|=3$. By our algorithm if $|U_i| < |U_k|$ then $\ell_i$ is the $\ell$-ring that is selected. By contrast, if $|U_i| = |U_k|$ then, assume without loss of generality that $u_1$ and $u'_1$ are the empty nodes on $\ell_i$ and $\ell_k$ respectively. Observe that $\neg Free(u_1,u_3)$ and $\neg Free(u'_1,u'_3)$ holds (otherwise we are in Case \ref{case:ui=3}).  Let $R_i$ (respectively $R_k$) be the set of robots on $\ell_i$ (respectively $\ell_k$) being on the same side of the axes of symmetry as $u_1$ (respectively $u'_1$). Let $R= R_i \cup R_k$. If without loss of generality $|R_i| < |R_k|$ then, $\ell_i$ is selected. By contrast, if $|R_i| = |R_k|$ then, let $r_1$ (respectively $r'_1$) be the robot on $\ell_i$ (respectively $\ell_k$) which is the closest to $u_3$ (respectively $u'_3$). If without loss of generality  $dist(r_1, u_3) < dist(r'_1, u'_3)$ then $\ell_i$ is selected. Otherwise, if $dist(r_1, u_3) = dist(r'_1, u'_3)$ then if $view_{r_1}(t)(1)> view_{r'_1}(t)(1)$ then $\ell_i$ is selected. Otherwise, $\ell_k$ is selected. Assume without loss of generality that $\ell_i$ is the one that is selected. By our algorithm, the robot on $\ell_i$ which is the closest to $u_3$, is the one allowed to move. Its destination is its adjacent node toward $u_3$. By moving, the robot joins $u_3$ or becomes the only robot that is closest to $u_3$. Let $C'$ be the reached configuration. If the robot joined $u_3$ then $nb_{\ell_i}(C') < nb_{\ell_k}(C')$. Since $U_i=3$ in $C'$, $C'$ is rigid. By Lemma \ref{lem:undefined-1}, we can deduce that the lemma holds. By contrast, if the robot does not join $u_3$, as the robot get even closer to $u_3$, $\ell_i$ remains the selected $\ell$-ring and the robot that has moved remains the only one allowed to move. That is eventually, it joins $u_3$ and as discussed previously, we can deduce that the lemma holds.

                \item $|U_i|=2$. By our algorithm if $|U_i| < |U_k|$ then $\ell_i$ is the $\ell$-ring that is selected. By contrast, if $|U_i| = |U_k|$ then, three cases are possible on each $\ell$-ring of $\mathcal{L}$ depending on the nodes of $U$ that are occupied: (I) the two nodes are neighbors, (II) the two nodes are on the same side of $\Gamma(C)$'s axes of symmetry. (III) the two nodes are not neighbors and are on different side of the axes of symmetry. We set: $Case(I) > Case(II) > Case(III)$ where $Case(a) > Case(b)$ means that $Case(a)$ has a higher priority than $Case(b)$. That is, if $\ell_i$ and $\ell_k$ have two different priorities then, the $\ell$-ring with the largest priority is one that is selected. By contrast, if the two $\ell$-rings have the same priority (they belong to the same case), the selection is done in the following manner: if both $\ell_i$ and $\ell_k$ belongs to $Case(I)$ (assume without loss of generality that $u_1$, $u_2$, $u'_1$ and $u'_2$ are the nodes of $U$ that are occupied) then, let $F_1$, $F_2$ (respectively $F'_1$, $F'_2$) be the number of robots on $\ell_i$ (respectively $\ell_k$) being on each side of $\Gamma(C)$'s axes of symmetry (Observe that $\forall i \in \{1,2\}$, $F_j \ne 0$ and $F'_j \ne 0$, otherwise we are in Case \ref{case:ui=2-neighbor}). Let $F = \mathit{min}(F_1,F_2,F'_1,F'_2)$. The $\ell$-ring that has a side with $F$ robots is elected. If both $\ell$-rings has a side with $F$ robots, we use the distance to break the symmetry. That is, let $d$ be the largest distance between an occupied node on an $\ell$-ring of $\mathcal{L}$ and the occupied node of $U$ on the same $\ell$-ring. Let $R_1$ and $R_2$ be the set of these nodes located on respectively $\ell_i$ and $\ell_k$. Let $R= R_1 \cup R_2$.  If without loss of generality $R_1 < R_2$ then $\ell_i$ is selected. Otherwise, let $d_i$ (respectively $d_k$) be the smallest distance between a robot on $\ell_i$ (respectively $\ell_k$) that are in a different side from the node of $R_i$ (respectively $R_k$). If without loss of generality $d_i < d_k$ then $\ell_i$ is selected. If $d_i = d_k$, the $\ell$ that hosts a robot that is at distance $d_i$ from an occupied node of $U$ with the largest view is one that is selected.  Now, if both $\ell$-ring of $\mathcal{L}$ belong to $Case(II)$ or if they both belong to $Case(III)$ then, let $r$ (respectively $r'$) be the robot allowed to move on $\ell_i$ (respectively $\ell_k$) with respect to our algorithm. If $view_r(t)(1) > view_{r'}(t)(1)$ then $\ell_i$ is elected. Otherwise, $\ell_k$ is elected. From here, we know that a unique $\ell$-ring is selected. Assume without loss of generality that this $\ell$-ring is $\ell_i$. We discuss each case separately:
                    \begin{itemize}
                        \item The two occupied nodes of $U_i$ are neighbors. Let $R$ be the set of occupied nodes on $\ell_i$ that are the farthest from an occupied node of $U_i$ being on the same side of $\Gamma(C)$'s axes of symmetry. Note that $ 1 \leq |R| \leq 2$ (one at each side of the axes of symmetry). If $|R|=2$ then, as $C$ is symmetry the robot that is on a node of $R$ with the largest view is the one allowed to move. Its destination is its adjacent node on its $\ell$-ring toward the occupied node of $U_i$, being on the same side, taking the shortest path. By moving, a configuration $C'$ which is still in $Case(I)$ and in which $|R|=1$ is reached. When $|R|=1$, the robot that is on $\ell_i$ in the opposite side of the axes from the node of $u$ and which is the closest to the occupied node of $U_i$ is the one allowed to move. By moving, it either joins the occupied node of $U_i$ or it becomes the only closest robot to an occupied node of $U$. Hence, in the latter case, $\ell_i$ remains the one that is elected and the robot that has moved remains the only one allowed to move. That is, eventually it joins the occupied node of $U_i$. In the configuration reached $C'$, $nb_{\ell_i}(C') < nb_{\ell_k}(C')$ and thanks to the farthest occupied node in $R$, the configuration is rigid. Hence, by Lemma \ref{lem:undefined-1}, we can deduce that the lemma holds. 
                        
                        \item The two occupied nodes of $U_i$ are located on the same side or different sides of $\Gamma(C)$'s axes of symmetry. By our algorithm, the robot on an occupied node of $U_i$ having the largest view is the one allowed to move. Its destination is its adjacent node on its current $\ell$-ring in the opposite direction of the adjacent node in $U_i$. Assume without loss of generality that this node is $u_1$. Once the robot moves, $u_1$ becomes empty. If the robot of $u_1$ joined an occupied node, then in the configuration reached $C'$, $C'$, $nb_{\ell_i}(C') < nb_{\ell_k}(C')$. As $C'$ is rigid (thanks to the unique occupied node of $U_i$), by Lemma \ref{lem:undefined-1}, we can deduce that the lemma holds. By contrast, if the robot of $u_1$ joined an empty node then $nb_{\ell_i}(C') = nb_{\ell_k}(C')$ but $U_i=1$. We thus retrieve Case \ref{case:ui=1nsp}. 
                        
                    \end{itemize}
                
                 \item \label{case:ui=1nsp} $|U_i|=1$. By our algorithm if $|U_i| < |U_k|$ then $\ell_i$ is the $\ell$-ring that is selected. By contrast, if $|U_i| = |U_k|$ then, let $d$ be the smallest distance between a robot on an $\ell$-ring of $\mathcal{L}$ and a node of $U$ located on the same $\ell$-ring. The $\ell$-ring that hosts a robot at distance $d$ from a node of $U$ with the largest view is the one that is elected (this is possible as $C$ is symmetric). Assume without loss of generality that $\ell_i$ is the one that is selected and that $u_1$ is the only occupied node of $U_i$. By our algorithm, the robot that is the closest to $u_1$ is the one allowed to move. Its destination is its adjacent node on its current $\ell$-ring toward $u_1$. By moving, it either joins $u_1$ or get closer to $u_1$. If the robot of $u_1$ joined an occupied node, then in the configuration reached $C'$, $C'$, $nb_{\ell_i}(C') < nb_{\ell_k}(C')$. As $C'$ is rigid (thanks to the unique occupied node of $U_i$), by Lemma \ref{lem:undefined-1}, we can deduce that the lemma holds. By contrast, if the robot did not join $u_1$, then as this robot is the only one that is the closest to $u_1$, $\ell_i$ remains the selected $\ell$-ring and the same robot remains the only one allowed to move. That is, eventually, it joins $u_1$. Hence, we can deduce that the lemma holds in this case too. 
                 
                 \item $|U_i|=0$.if $|U_i| < |U_k|$ then $\ell_i$ is the $\ell$-ring that is selected. By contrast, if $|U_i| = |U_k|$ then, let $d$ be the smallest distance between a robot on an $\ell$-ring of $\mathcal{L}$ and a node of $U$ on its $\ell$-ring. The $\ell$-ring that hosts a robot at distance $d$ from a node of $U$ with the largest view is the one that is elected (Again, this is possible as $C$ is rigid). Assume without loss of generality that $\ell_i$ is selected. By our algorithm, the robot that is the closest to a node of $U_i$ with the largest view is the one that is allowed to move. Its destination is its adjacent node toward the closest node of $u_i$ taking the shortest path. By moving it either joins a node of $U_i$ and in this case, we retrieve Case \ref{case:ui=1nsp} or it becomes the only robot that is the closest to a node of $U$. Hence, $\ell_i$ remains the one that is selected and the robot that has moved remains the only one allowed to move. That is, the robot eventually joins a node of $U_i$ and we retrieve Case \ref{case:ui=1nsp}.

            \end{enumerate}

    \end{enumerate}

 \end{enumerate}
From the cases above we can deduce that the lemma holds. 
\end{proof}

\begin{lemma}\label{lem:periodicGamma}
Let  $C \in \mathcal{C}_{Undefined}$ be a rigid configuration with no outdated robots such that $\Gamma(C)$ is periodic with at least four occupied $\ell$-rings . There exists a robot $r$ on an $\ell$-ring $\ell_t \not \in \{\ell_i, \ell_k\}$ such that if $r$ moves then a configuration $C' \in \mathcal{C}_{Undefined}$ with no outdated robots is reached with $\Gamma(C')$ either rigid or contains a single axes of symmetry. 
\end{lemma}

\begin{proof}
By contradiction assume that such a robot does not exist. Let $d$ and $\mathcal{D}$ be respectively the smallest distance between two occupied nodes on $\ell_t$ in configuration $C$ and the largest $d$.block on $\ell_t$ IN $C$. Note that since $\Gamma(C)$ is periodic, $nb_{\ell_t}(C)<\ell-1$ (Otherwise $\Gamma(C)$ is not periodic as $d=1$ and $\ell_t$ contains a single block $\mathcal{D}$ of size $\ell$-1). Let $R(C)$ be the set of robots on $\ell_t$ in $C$ that are the closest to a block $\mathcal{D}$. We distinguish two cases: 
\begin{itemize}
    \item Set $R(C)$ is not empty. As $C$ is rigid, one robot in $R$ can be elected to move. Let this robot be the one with the smallest view. We refer to this robot by $r$. Let $C'$ be the configuration reached when $r$ moves. In $C'$ either there is a single biggest $d$.block $\mathcal{D}$ as $r$ joins one of the biggest $d$.block in $C$ or $r$ becomes the only robot that is the closest to a $d$.block $\mathcal{D}$. In the case in which $\Gamma(C')$ contains an axes of symmetry then this axes is unique (it crosses the unique biggest block $\mathcal{D}$ in $C'$) which is a contradiction. In the later case, $C'$ is rigid otherwise in $C$ there was another robot $r'$ which was closer to a block $\mathcal{D}$ which is also a contradiction. 
    
    \item Set $R(C)$ is empty (there is a single $d$.block on $\ell_t$ in $C$). In this case, let $r$ be the robot at the border of the unique $d$.block with the smallest view. Note that as $C$ is rigid, such a robot is unique. If $r$ moves to its adjacent node on $\ell_t$ inside the $d$.block it belongs to then in $C'$, the configuration reached once $r$ moves, there is a unique $d-1$.block. Hence, if $\Gamma(C')$ contains an axes of symmetry then this axes is unique which a contradiction.    
\end{itemize}

From the cases above we can deduce that there exists a robot $r$ on  an $\ell$-ring $\ell_t \not \in \{\ell_i, \ell_k\}$ such that if $r$ moves then a configuration $C' \in \mathcal{C}_{Undefined}$ is reached with $\Gamma(C')$ being either rigid or contains a single axes of symmetry. Note that since there is a unique robot that is allowed to move. In $C'$ there are no outdated robots. Hence the Lemma holds. 
\end{proof}

From Lemmas \ref{lem:undefined-1}, \ref{lem:undefined-2} and Lemma \ref{lem:periodicGamma} we can deduce the following corollary: 

\begin{corollary}\label{cor:undefined}
Starting from a configuration $C \in \mathcal{C}_{Undefined}$ with no outdated robots, a configuration $C' \in \mathcal{C}_{Oriented}$ with no outdated robots is eventually reached. 
\end{corollary}


\begin{lemma}\label{lem:empty}
Starting from Configuration $C \in \mathcal{C}_{Empty}$ in which there are no outdated robots, a configuration $C' \in \mathcal{C}_{p_2}$ with no outdated robots is eventually reached. 
\end{lemma}

\begin{proof}
Two cases are possible as follows: 
\begin{enumerate}
\item $\forall j \in \{1,2,\dots,L\}$ such that $\ell_j \ne \ell_{\mathit{max}}$, $nb_{\ell_j}(C)=0$. Remember that $C$ contains a single $\ell$-ring that is occupied which is $\ell_{\mathit{max}}$ by default. Since the configuration is rigid, by our algorithm, one robot is selected from $\ell_{\mathit{max}}$ to move to its adjacent empty node outside the $\ell$-ring it belongs to (the direction is chosen by the adversary). Once it moves, a configuration $C' \in \mathcal{C}_{p_2}$ is then reached. 
 
\item Otherwise \ie $\exists~j \in \{1,2 \dots, L\}$ such that $\ell_{j} \ne \ell_{\mathit{max}}$ and $nb_{\ell_j}(C) \ne 0$. According to our algorithm robots in $R_m$ are the ones allowed to move where $R_m$ is the set of robots that are the closest to a node on $\ell_{\mathit{max}}$. Using the rigidity of $C$, a unique robot from $R_m$ is selected. Let us refer to this robot by $r$. The destination of $r$ is its adjacent empty node outside its $\ell$-ring toward $\ell_{\mathit{max}}$ taking the shortest path. Observe that when the robot moves, if it is neither on $\ell_i$ nor $\ell_{k}$, $r$ remains the only one allowed to move (since it is the closest one to $\ell_{\mathit{max}}$, $|R_m|=1$). Robot $r$ keeps the same destination and hence keep moving toward $\ell_{\mathit{max}}$ taking the shortest path. That is, eventually, $r$ becomes on either $\ell_{i}$ or $\ell_{k}$. That is, a configuration $C \in \mathcal{C}_{p_2}$. 
\end{enumerate}

Observe from the cases above that at each instant, there is only one robot which is allowed to move. That is, as long as such a robot remains idle, it keeps being the only one allowed to move. Hence, we can deduce that the lemma holds.
\end{proof}
 
 \begin{lemma}
Starting from a configuration $C \in \mathcal{C}_{Semi-Empty}$ in which there are no outdated robots, a configuration $C' \in \mathcal{C}_{Oriented}$ with no outdated robots is eventually reached.
\end{lemma}

\begin{proof}
Let $\ell_i$ and $\ell_k$ be the two $\ell$-rings that are adjacent to $\ell_{\mathit{max}}$. Assume without loss of generality that $nb_{\ell_k}(C)>1$ and $nb_{\ell_i}(C)=0$.  Let $\ell_n$ be the $\ell$-ring that is neighbor to $\ell_i$. Note that $\ell_n=\ell_k$ is a possible case. Let $\rightarrow$ be the direction from $\ell_{\mathit{max}}$ to $\ell_{k}$ taking the shortest path. By our algorithm, since $C$ is rigid, exactly one robot from $\ell_{n}$ is selected to move. Its destination is its adjacent empty node outside its $\ell$-ring with respect to $\rightarrow$. Let us refer to this robot by $r$. Observe that when $r$ moves, if $r$ is not on $\ell_i$, then no other robot is allowed to move except $r$ (since $r$ is the only closest robot now to $\ell_i$). Hence, eventually, $r$ joins $\ell_{i}$. Let us refer to the configuration reached by $C'$. Then, $nb_{\ell_i}(C')=1$ and $nb_{\ell_k}(C')>1$. Therefore $C' \in \mathcal{C}_{Oriented}$ is eventually reached. Moreover, as there is a unique robot that is activated at each time, when $C'$ is reached there are no robots with outdated views. Hence the lemma holds. 

\end{proof}


\begin{lemma}\label{lem:oriented2}
Starting from a configuration $C \in \mathcal{C}_{Oriented-2}$ with no outdated robots, a configuration $C' \in \mathcal{C}_{Oriented-1}$ with no outdated robots is eventually reached. 
\end{lemma}

\begin{proof}
Let $u$ be the node on $\ell_k$ that is on the same $L$-ring as $v_{targer}$. According to our algorithm, as long as $nb_{\ell_k}(C)>3$, robots that are the closest to $u$ moves to their adjacent node towards $u$ taking the shortest path. As robots move to join $u$, the number of robots decreases on $\ell_k$ and eventually a configuration $C'$ in which $nb_{\ell_k}(C')=3$ or $nb_{\ell_k}(C')=2$ is eventually reached. Once such such a configuration is reached, \textbf{Align}($\ell_k,\ell_i$) is executed. By Lemma \ref{lem:align}, a configuration $C''\in \mathcal{C}_{Oriented-1}$ with no outdated robot is reached and the lemma holds.

\end{proof}


\begin{lemma}\label{lem:oriented1}
Starting from a configuration $C \in \mathcal{C}_{Oriented-1}$ with no outdated robots, a configuration $C' \in \mathcal{C}_{p_2}$ with no outdated robot is eventually reached. Moreover, $nb_{\ell_{\mathit{max}}}(C) \ne 4$ and $\ell_{\mathit{max}}$ contains at most one multiplicity node. This node is adjacent to $v_{\mathit{target}}$.  
\end{lemma}

\begin{proof}

Let $u$ be the node on $\ell_{\mathit{max}}$ such that $u$ is on the same $L$-ring as the unique occupied node on $\ell_i$.  Some extra steps are taken in the case where $nb_{{\ell}_{\mathit{max}}}(C)<5$. By our algorithm, if $nb_{{\ell}_{\mathit{max}}}(C)=3$ then robots on $\ell_{\mathit{max}}$ executes \textbf{Align}$(\ell_{\mathit{max}}, \ell_i)$. By Lemma \ref{lem:align}, robots on $\ell_{\mathit{max}}$ eventually form a $1$.block of size $3$ whose middle node is adjacent to the unique robot on $\ell_i$ and no robot on $\ell_{\mathit{max}}$ has an outdated view. From the reached configuration, the unique robot on $\ell_i$ moves to its adjacent node on $\ell_{\mathit{max}}$ (note that this node is occupied). That is a configuration $C' \in \mathcal{C}_{p_2}$ with no outdated robot is eventually reached. Moreover, in $C'$, $\ell_{\mathit{max}}$ contains one multiplicity node. This node is adjacent to $v_{\mathit{target}}$.  By contrast, if $nb_{{\ell}_{\mathit{max}}}(C)=4$ then if $u$ is empty then the unique robot on $\ell_i$ moves to join $u$ and the lemma holds (Observe that in this case $\ell_{\mathit{max}}$ contains no multiplicity node). If $u$ is occupied then by our algorithm if $u$ has an empty adjacent node on $\ell_{\mathit{max}}$ then the robot on $u$ moves to its adjacent empty node on $\ell_{\mathit{max}}$ (the scheduler chooses the direction to take if the robot has a choice to make). If $u$ has no adjacent empty node on $\ell_{\mathit{max}}$ then by our algorithm the robot on $\ell_{\mathit{max}}$ which is adjacent to $u$ and does not have a neighboring robot at distance $\lfloor \ell / 2 \rfloor $ is the one allowed to move. Its destination is its adjacent empty node on $\ell_{\mathit{max}}$. That is, eventually $u$ becomes empty. The unique robot on $\ell_i$ is then the only one allowed to move. Its destination is $u$. Hence, a configuration $C' \in \mathcal{C}_{p_2}$ in which there are no multiplicity nodes and no outdated robots on $\ell_{\mathit{max}}$ is reached. 

Finally, if $nb_{{\ell}_{\mathit{max}}}(C)\geq 5$ then by our algorithm, the unique robot on $\ell_i$ is the one allowed to move. Its destination is its adjacent node on $\ell_{\mathit{max}}$. We can thus deduce that the lemma holds.

\end{proof}


From Lemmas \ref{lem:oriented1} and \ref{lem:oriented2}, we can deduce the following corollary: 

\begin{corollary}\label{cor:oriented}
Starting from a configuration $C \in \mathcal{C}_{Oriented}$, a configuration $C' \in \mathcal{C}_{p2}$ is eventually reached. 
\end{corollary}

\begin{lemma}\label{lem:semiO}
Starting from a configuration $C \in \mathcal{C}_{Semi-Oriented}$ with no outdated robots, a configuration $C' \in \mathcal{C}_{Oriented}$ with no outdated robots is eventually reached. 
\end{lemma}

\begin{proof}
Recall that when $C \in \mathcal{C}_{Semi-Oriented}$, $nb_{\ell_{i}}(C)= nb_{\ell_{k}}(C)=1$ where $\ell_i$ and $\ell_k$ are the two $\ell$-rings that are adjacent to $\ell_{\mathit{max}}$. Let $\ell_{n_i}$ and $\ell_{n_k}$ be the two $\ell$-rings that are neighbors of respectively $\ell_i$ and $\ell_k$. In the case where $\ell_i=\ell_{n_k}$ and hence $\ell_k=\ell_{n_i}$ (the configuration contains only three occupied $\ell$-rings) then by our algorithm, using the rigidity of $C$ one robot on either $\ell_k$ or $\ell_i$ is elected to move. Its destination is its adjacent empty node in the opposite direction of $\ell_{\mathit{max}}$. Once the robot moves, a configuration $C' \in \mathcal{C}_{p_2}$ is eventually reached. By contrast, if $\ell_i\neq \ell_{n_k}$ and hence $\ell_k \neq \ell_{n_i}$ then using the rigidity of $C$, one of the closest robot to either $\ell_i$ or $\ell_k$ is elected to move. Its destination is its adjacent empty node towards the closest node to either $\ell_i$ or $\ell_k$. Note that once the robot moves, it is the only one allowed to move as it is the closest one. That is a configuration $'C' \in \mathcal{C}_{Oriented}$ is reached. Note also that since at each time only one robot is allowed to move, $C'$ does not contain outdated robots. 

\end{proof}

By Corollary \ref{cor:oriented} and Lemma \ref{lem:semiO}, we can deduce the following:

\begin{corollary}\label{cor:semi-oriented}
Starting from a configuration $C \in \mathcal{C}_{Semi-Oriented}$, a configuration $C' \in \mathcal{C}_{p2}$ with no outdated robots is eventually reached. 
\end{corollary}

\begin{theorem}\label{theo:Phase1}
Starting from any rigid configuration $C\in \mathcal{C}_{p_1}$ with no outdated robots, a configuration $C'\in \mathcal{C}_{p_2}$ with no outdated robots is eventually reached. Moreover, $\ell_{\mathit{max}}$ contains at most one multiplicity node. This node is the one hat is on the same $L$-ring as $v_{\mathit{target}}$. 
\end{theorem}

\begin{proof}
Derived from Corollaries \ref{cor:uniqueC}, \ref{cor:undefined},  \ref{cor:oriented} and \ref{cor:semi-oriented} and Lemma \ref{lem:empty}. 
\end{proof}

Figure \ref{fig:Preparation_transition} summarizes the transitions within Preparation phase configurations.  

    \begin{figure}[h]
        \begin{center}
            \includegraphics[scale=0.17]{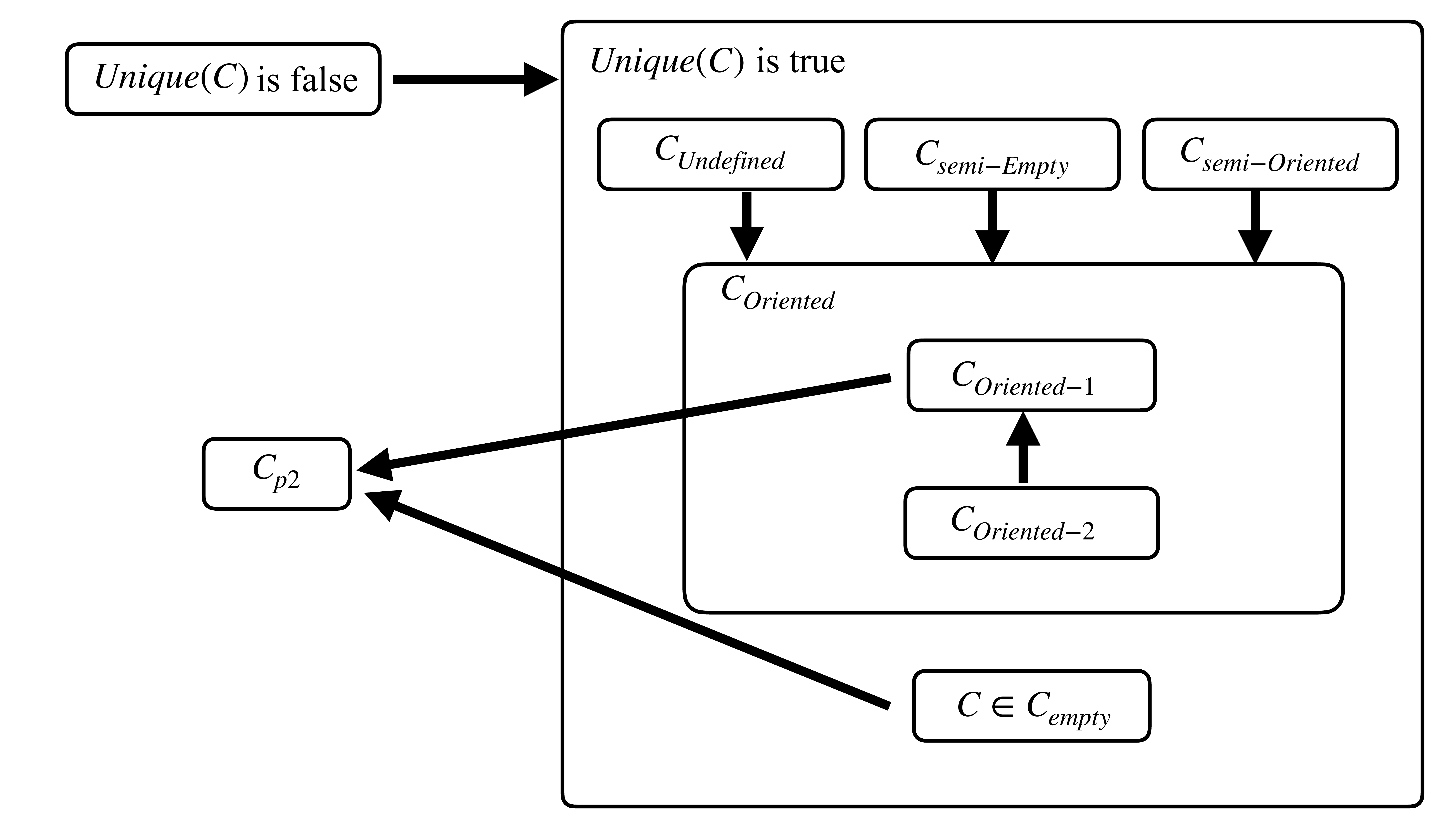}
            \caption{Transitions within Preparation phase}\label{fig:Preparation_transition}
        \end{center}
    \end{figure}


We now show that starting from a configuration $C \in \mathcal{C}_{p_2}$ with no outdated robots and a potential one multiplicity node on $\ell_{\mathit{max}}$ which is on the same $L$-ring as $v_{\mathit{target}}$, the gathering is eventually achieved.

We now show that starting from a configuration in either $\mathcal{C}_{pr}$ or $\mathcal{C}_{ls}$, a configuration in $\mathcal{C}_{sp}$ is eventually reached.  

\begin{lemma}\label{lem:prtols}
Starting from $C \in \mathcal{C}_{pr}$ with no outdated robots, a configuration $C'\in \mathcal{C}_{ls}$ with no outdated robots is eventually reached. 
\end{lemma}

\begin{proof}
Recall that for any $C \in \mathcal{C}_{pr}$, $ C \in \mathcal{C}_{\mathit{target}}$ holds. Let $\ell_i$ be the $\ell$-ring adjacent to $\ell_{\mathit{target}}$ such that $\ell_i\ne \ell_{\mathit{max}}$. We distinguish two cases:

\begin{enumerate}
\item $nb_{\ell_i}(C)>0$ \label{Case:li_occupied}. Recall that in some cases during the preparation phase, to reduce the number of maximal $\ell$-rings, a single tower of size $2$ is created in some $\ell$-rings that were maximal in the initial configuration (the case in which the number of occupied nodes on each $\ell_{\mathit{max}}$ is equal to $\ell$). Robots need to move carefully so that the number of $\ell$-rings that are maximal never increases again. Hence, in the proposed solution, depending on the number of robots on $\ell_i$, robots behave differently. Let $R_{{\ell}_i}$ be the set of robots on $\ell_i$ and let $R_m \subseteq R_{{\ell}_i}$ be the set of robots that are the closest to $v_{\mathit{target}}$. We refer by $u$ to the node on $\ell_i$ which is adjacent to $v_{\mathit{target}}$. The following sub-cases are possible:
    \begin{enumerate}
       \item \label{case:neighbor} $u$ is occupied. With respect to the algorithm, robots on $u$ move to join $v_{\mathit{target}}$. Observe that as long as there are robots on $u$, these robots remain the only ones that are allowed to move. Since the scheduler is weakly fair, eventually, $u$ becomes empty, $R_{{\ell}_i}$ decreases and we retrieve Case~\ref{case:notminus1}.
       
       \item \label{case:notminus1} $u$ is empty. With respect to the proposed solution, if $nb_{\ell_i}(C)<\ell-1$, robots in $R_m$ are allowed to move, their destination is their adjacent empty node on $\ell_i$ toward $u$ taking the shorted path. That is, at least one robot on $\ell_i$ eventually joins $u$ and we retrieve Case \ref{case:neighbor} (Observe that in this case, there are no towers on $\ell_i$). By contrast, if $nb_{\ell_i}(C)=\ell-1$ then by Lemma \ref{cor:Unique_Multi}, $\ell_i$ can hosts at most one multiplicity node. That is, there is at least one robot in $R_m$ which is not part of a multiplicity node. According to the algorithm, this robot is the one allowed to move, its destination is its adjacent empty node on $\ell_i$ toward $u$ taking the shortest path. By moving, we retrieve eventually Case \ref{case:neighbor}.

    \end{enumerate}
    
According to the algorithm, as long as $nb_{\ell_i}(C)>0$, no robot other than those which are on $\ell_i$ are allowed to move. That is, $R_{{\ell}_i}$ never increases. On another hand, at each time Case~\ref{case:neighbor} occurs, $R_{{\ell}_i}$ decreases. From Cases~\ref{case:neighbor} and \ref{case:notminus1}, we can deduce that a configuration $C'$ in which $nb_{\ell_i}(C')=0$ is eventually reached. We then retrieve Case~\ref{case:neighbor0}.

\item \label{case:neighbor0}$nb_{\ell_i}(C)=0$. Let $\ell_k$ be the neighboring $\ell$-ring of $\ell_{\mathit{target}}$ with respect to direction $\rightarrow$ (from $\ell_{\mathit{target}}$ to $\ell_{\mathit{max}}$ taking the shortest path). Let $R_{m}$ be the set of robots on $\ell_{k}$ that are the closest to $v_{\mathit{target}}$ with respect to $\rightarrow$. From the proposed algorithm, robots on $R_{m}$ are the ones allowed to move, their destination is their adjacent node outside $\ell_{k}$ and toward $\ell_{\mathit{target}}$ with respect to $\rightarrow$. Let us refer to the $\ell$-ring to which the robots have moved by $\ell_{k-1}$. Note that $\ell_{k-1}$ is empty in $C$ (By definition of neighboring $\ell$-ring). Once the scheduler activates at least one robot from $R_m$ in $C$, a configuration $C'$ in which $1<=nb_{\ell_{k-1}}(C')<=2$ is eventually reached. Since no other robots on $\ell_{k}$ is allowed to move, by induction, we can show that at least one robot from the robots that moved from $\ell_k$ eventually reaches $\ell_i$. Thus, we retrieve Case \ref{Case:li_occupied}. 

\end{enumerate}

From the cases above, we can deduce that robots on an $\ell$-ring different from $\ell_{\mathit{max}}$ and $\ell_{\mathit{target}}$ keep getting closer to $v_{\mathit{target}}$ to eventually join it. Once they join $v_{\mathit{target}}$ they are not allowed to move anymore. Hence, we can deduce that the lemma eventually holds. 

\end{proof}


\begin{lemma} \label{lem:Cls5-to-Csp}
Starting from a configuration $C \in \mathcal{C}_{ls}$ in which $nb_{\ell_{\mathit{target}}}(C)=1$, a configuration $C' \in \mathcal{C}_{sp-3}$ or $C' \in \mathcal{C}_{sp-2}$ is eventually reached. Moreover, the only multiplicity node on $\ell_{\mathit{max}}$ is the one that is adjacent to $v_{\mathit{target}}$. 
\end{lemma}

\begin{proof}
From Theorem \ref{theo:Phase1}, we know that when the second phase starts, no robot has an outdated view. On another hand, by Lemma \ref{lem:oriented2}, if $nb_{\mathit{max}}(C)=4$ then no robot on $\ell_{\mathit{max}}$ contains a multiplicity node. That is, by Lemma \ref{lem:align}, we can deduce that a configuration $C' \in \mathcal{C}_{sp-3}$ or $C' \in \mathcal{C}_{sp-2}$ is eventually reached and no multiplicity node exists on $\ell_{\mathit{max}}$. By contrast, if $nb_{\mathit{max}}(C) \neq 4$ then by Lemmas \ref{theo:Phase1} and Lemma \ref{lem:prtols} $\ell_{\mathit{max}}$ contains no outdated robots and a most one multiplicity node. This multiplicity node (if it exists) is the one that is adjacent to $v_{\mathit{target}}$. 

Let $R$ be the set of robots on $\ell_{\mathit{max}}$ which are the closest to $u_3$. By our algorithm. If $|R|=2$ then both robots of $R$ are allowed to move. Their destination is their adjacent node toward $u_3$. If only one of the two robots move then by our algorithm the robot that was supposed to move is the only one allowed to move (the robot that is at distance $d$+1 from $u_3$ or the robot that is adjacent to $u_3$). That is a configuration $C'$ in which both robots of $R$ have moved is eventually reached.  Note that in $C'$, either $nb_{\ell_{\mathit{max}}}(C)= nb_{\ell_{\mathit{max}}}(C')$ (robots do not join $u_3$) or $nb_{\ell_{\mathit{max}}}(C)= nb_{\ell_{\mathit{max}}}(C')+1$ (both robots join $u_3$). In the first case, robots in $R$ remain the only one allowed to move. That is eventually they join node $u_3$. In the later case, if $nb_{\ell_{\mathit{max}}}(C')>5$ then by our algorithm, robots on $u_3$ are the ones allowed to move. That is, eventually $u_3$ becomes empty in the configuration reached $C''$, $nb_{\ell_{\mathit{max}}}(C)= nb_{\ell_{\mathit{max}}}(C'')+2$. By contrast, if $nb_{\ell_{\mathit{max}}}(C')=5$ (respectively $nb_{\ell_{\mathit{max}}}(C')=3$), robots on $\ell_{\mathit{max}}$ execute \textbf{Align}$(\ell_{\mathit{max}}, \ell_i)$. Note that during the execution of \textbf{Align}, if $u_3$ is occupied and possibly hosts a multiplicity, robots on $u_3$ do not move. That is, through out the execution of \textbf{Align}, the number of robots on $\ell_{\mathit{max}}$ remains equal to $5$ (respectively $3$). By Lemma \ref{lem:align}, we can deduce that a configuration $C' \in \mathcal{C}_{sp-3}$ or $C' \in \mathcal{C}_{sp-2}$ is eventually reached. Moreover, the only multiplicity node on $\ell_{\mathit{max}}$ is the one that is adjacent to $v_{\mathit{target}}$. Hence the lemma holds. 

\end{proof}

Similarly to Lemma \ref{lem:Cls5-to-Csp}, we can show the following lemma:

\begin{lemma}\label{lem:Cls-to-Csp2}
Starting from a configuration $C \in \mathcal{C}_{ls}$ with $nb_{\ell_{\mathit{target}}}(C)>1$ a configuration $C' \in \mathcal{C}_{sp-1}$ is eventually reached. Moreover, if $\ell_{\mathit{max}}$ contains a multiplicity node in $C'$ then this multiplicity is adjacent to $v_{\mathit{target}}$.
\end{lemma}





\begin{lemma}\label{lem:Csp-1toCsp-2}
Starting from a a configuration $C \in \mathcal{C}_{sp-1}$, a configuration $C' \in \mathcal{C}_{sp-2}$ is eventually reached. 
\end{lemma}

\begin{proof}

Let $C \in \mathcal{C}_{sp-1}$.  We refer to the non empty $\ell$-ring that is adjacent to $\ell_{\mathit{max}}$ by $\ell_{mark}$. Depending on the number of robots on $\ell_{mark}$, the following cases are considered:

\begin{enumerate}
    \item \label{case:2onli} $nb_{\ell_{mark}}(C)=2$ and $\ell_{mark}$ contains a $1$.block of size 2. By our algorithm, robots on $\ell_{mark}$ that have two adjacent occupied nodes are the ones allowed to move. If the scheduler activates all the robots allowed to move, $C' \in C_{sp-2}$ is reached and hence the lemma holds. By contrast, if the scheduler activates only a subset of robots, the configuration remains in $\mathcal{C}_{sp-1}$. Moreover, the same robots keep being the only ones allowed to move. Since the scheduler is weakly fair, these robots are eventually activated and hence a configuration $C' \in C_{sp-2}$ is eventually reached. 
    
    \item \label{case:3-To-2} $nb_{\ell_{mark}}(C)=2$ and $\ell_{mark}$ contains a $2$.block of size 2. Observe that in this case $C \in C_{\mathit{target}}$ and $\ell_{\mathit{target}}=\ell_{mark}$. By our algorithm, robots adjacent to $v_{\mathit{target}}$ are the ones allowed to move. Their destination is $v_{\mathit{target}}$. If the scheduler activates all the robots allowed to move, a configuration $C' \in C_{sp-2}$ is reached and the lemma holds. If the scheduler activates only robots at one border of the $2$.block, we retrieve Case~\ref{case:2onli}. Finally if the scheduler activates a subset of robots then a $1$.block of size $3$ is created and we retrieve Case \ref{case:3onli}. 
    
    \item \label{case:3onli} $nb_{\ell_{mark}}(C)=3$. Observe that, in this case  $C \in C_{\mathit{target}}$ and $\ell_{\mathit{target}}=\ell_{mark}$. According to our algorithm, robots on $\ell_{mark}$ that are at the border of the $1$.block are the only ones allowed to move. Their destination is their adjacent on $\ell_{mark}$ inside the block they belong to. If the scheduler activates all the robots allowed to move, $C' \in C_{sp-2}$ is reached and the lemma holds. If by contrast, the scheduler activates only robots on one border of the block, $\ell_{\mathit{target}}$ contains a single $1$.block of size 2 and hence we retrieve Case~\ref{case:2onli}. Finally, if the scheduler activates only a subset of robots that move then the configuration remains in $\mathcal{C}_{sp-1}$ and the same robots remain the only one allowed to move. As the scheduler is weakly fair, the number of robots on the borders of the $1$.block decreases. Thus, eventually, a configuration $C' \in C_{sp-2}$ is reached.
    
\end{enumerate}

From the cases above, we can deduce that eventually all robots on $\ell_{\mathit{target}}$ join $v_{\mathit{target}}$ and hence we can deduce that the lemma holds. 

\end{proof}


\begin{lemma}\label{lem:Csp4-Csp-3}
Starting from a configuration $C \in \mathcal{C}_{sp-2}$ in which $\ell_{\mathit{max}}$ contains at most one multiplicity node and such a multiplicity (if it exists) is adjacent to $v_{\mathrm{\mathit{target}}}$, a configuration $C' \in \mathcal{C}_{sp-3}$ is eventually reached. 
\end{lemma}

\begin{proof}

Let $C \in \mathcal{C}_{sp-2}$. The following cases are possible:
\begin{enumerate}
    \item $nb_{\ell_{\mathit{max}}}(C)=5$ \label{Case:5}. By our algorithm, robots on $\ell_{\max}$ which are adjacent to the borders of the $1$.block of size $5$ move to join the middle node of the $1$.block. If all robots allowed to move, execute the move phase then a configuration in which $\ell_{\mathit{max}}$ contains a single $2$.block of size $3$ is reached. By our algorithm, the robots on $\ell_{\mathit{max}}$ execute \textbf{Align}$(\ell_{\mathit{max}}, \ell_{\mathit{target}})$ and hence the robots at the extremities of the $2$.block eventually move to create a single $1$.block of size $3$ whose middle robot is adjacent to $v_{\mathit{target}}$. Thus a configuration $C' \in \mathcal{C}_{sp-3}$ is reached in this case. By contrast, if only a subset of robots which are allowed to move perform the move phase then we retrieve Case \ref{Case:3-1}. 
    
     \item $nb_{\ell_{\mathit{max}}}(C)=4$  and $\ell_{\mathit{max}}$ contains two $1$.blocks of size $2$. According to our algorithm, robots in each $1$.block who has an adjacent occupied node at distance $2$ are the one allowed to move. Depending on the robots that have moved, we retrieve either Case \ref{Case:5} or Case \ref{Case:3-1} or a configuration in which  a single $1$.block of size $3$ whose middle robot is adjacent to $v_{\mathit{target}}$ is created. In the later case, by our algorithm \textbf{Align}$(\ell_{\mathit{max}}, \ell_{\mathit{target}})$ is executed. That is, the extremities of the $2$.block eventually move to create a single $1$.block of size $3$ whose middle robot is adjacent to $v_{\mathit{target}}$. A configuration $C' \in \mathcal{C}_{sp-3}$ is then reached.
     \item  $nb_{\ell_{\mathit{max}}}(C)=4$ and $\ell_{\mathit{max}}$ contains a $1$.block of size $3$ \label{Case:3-1}. By our algorithm, robots in the middle of the $1$.block are the ones allowed to move. As the scheduler is weakly fair, eventually a $2$.block of size $3$ is created. By our algorithm, as \textbf{Align}$(\ell_{\mathit{max}}, \ell_{\mathit{target}})$ is executed the extremities of the $2$.block eventually move to create a single $1$.block of size $3$ whose middle robot is adjacent to $v_{\mathit{target}}$. Thus the Lemma holds.
    
\end{enumerate}

From the cases above we can deduce that the lemma holds. 

\end{proof}

\begin{lemma} \label{lem:Csp-3--Csp-4}
Starting from a configuration $C \in \mathcal{C}_{sp-3}$ in which the only node that hosts a multiplicity  on $\ell_{\mathit{max}}$ is the node that is adjacent to $v_{\mathit{target}}$, a configuration $C' \in \mathcal{C}_{sp-4}$ is eventually achieved. Moreover, in $\mathcal{C}_{sp-4}$, there is at most one occupied node that hosts a multiplicity. This node is the onde that is in the middle of the $1$.clock of size $3$. 
\end{lemma}

\begin{proof}

Let $C \in \mathcal{C}_{sp-3}$. Depending on the number of robots on $\ell_{\mathit{max}}$, two cases are distinguished:
   \begin{enumerate}
      \item $nb_{\ell_{\mathit{max}}}(C)=2$. By assumption, since the node on $\ell_{\mathit{max}}$ which is adjacent to $v_{\mathit{target}}$ is empty, $\ell_{\mathit{max}}$ does not host a multiplicity. From our algorithm, robots on $v_{\mathit{target}}$ move to their adjacent node on $\ell_{\mathit{max}}$. If $v_{\mathit{target}}$ hosts only one robot, then when the robot moves, a configuration $C' \in \mathcal{C}_{sp-4}$ is reached and the lemma holds. By contrast, if $v_{\mathit{target}}$ hosts more than one robot and the scheduler activates all robots on $v_{\mathit{target}}$ to move then a configuration $C \in C_{sp-4}$ is also reached and the lemma holds. Finally, if the scheduler activates only a subset of robots on $v_{\mathit{target}}$ then we retrieve Case~\ref{Case:3onmax}. 
      \item \label{Case:3onmax} $nb_{\ell_{\mathit{max}}}(C)=3$. Robots on $v_{\mathit{target}}$ are the only one allowed to move. Their destination is their adjacent node on $\ell_{\mathit{max}}$. Note that as long as $v_{\mathit{target}}$ is occupied, the only robots allowed to move are the ones that are on $v_{\mathit{target}}$. Hence, eventually all robots on $v_{\mathit{target}}$ move to their adjacent node on $\ell_{\mathit{max}}$ and a configuration $C' \in \mathcal{C}_{sp-4}$ is then reached. 
   \end{enumerate}

From the cases above, we can deduce that the lemma holds. 
\end{proof}

\begin{lemma}\label{lem:noTowerCsp2}
In a configuration $C\in \mathcal{C}_{sp-4}$, if $nb_{\ell_{\mathit{max}}}(C)=3$ (respectively $nb_{\ell_{\mathit{max}}}(C)=2$) then the border robots (respectively at least one of the two border robots) of the $3$.block (respectively $2$.block) do not host a multiplicity.  
\end{lemma}

\begin{proof}
From Theorem \ref{theo:Phase1}, we know that when the second phase starts, no robot has an outdated view and the only multiplicity node that can exist on $\ell_{\mathit{max}}$ is the one that is adjacent to $v_{\mathit{target}}$. On another hand, in any configuration $C \in \mathcal{C}_{ls}$ or $C \in \mathcal{C}_{pr}$, no robot on $\ell_{\mathit{max}}$ is enabled to move and no robot moves to a node of $\ell_{\mathit{max}}$. That is when a configuration $C' \in \mathcal{C}_{sp}$ is reached, the only possible multiplicity node on $\ell_{\mathit{max}}$ is the one that is adjacent to $v_{\mathit{target}}$. 
By Lemmas \ref{lem:Cls5-to-Csp}, \ref{lem:Cls-to-Csp2}, \ref{lem:Csp-1toCsp-2}, the only node on $\ell_{\mathit{max}}$ that can host a multiplicity is the one that is adjacent to $v_{\mathit{target}}$. From Lemma \ref{lem:Csp-3--Csp-4}, the only node that hosts a multiplicity is the one that is in the middle of the $1$.block of size $3$. By our algorithm, in a configuration $C \in \mathcal{C}_{sp-4}$ in which there are $3$ occupied nodes, the robot that is at the border of the $1$.block is the one allowed to move. If the scheduler activates one robot then the configuration remains in $\mathcal{C}_{sp-4}$ but with only two occupied nodes. That is, one of the two occupied nodes does not belong to a multiplicity. Thus, the lemma holds. 

\end{proof}

\begin{lemma}\label{lem:gethering}
Starting from a configuration $C \in \mathcal{C}_{sp-4}$, the gathering is eventually achieved. 
\end{lemma}

\begin{proof}
With respect to our algorithm, the following cases are considered :
  \begin{enumerate}

      \item $C$ contains a $1$.block of size $3$ \label{Case:1block3}. From Lemma~\ref{lem:noTowerCsp2}, the nodes at the border of the $1$.block do not host a multiplicity. According to our algorithm, the border robots are the ones allowed to move. Their destination is their adjacent occupied node. If both robots move at the same time then the gathering is achieved and the lemma holds. Otherwise, we retrieve Case~\ref{case:2robots}.
      
      \item \label{case:2robots} $C$ contains a $1$.block of size $2$. From Lemma~\ref{lem:noTowerCsp2}, exactly one of the two occupied nodes hosts a multiplicity. The unique robot that is allowed to move is the one that is not part of a multiplicity. By moving, the gathering is achieved and the lemma holds. 
  \end{enumerate}
From the cases above we can deduce that the lemma holds.   
  
\end{proof}

\begin{lemma}\label{lem:gathering}
Starting from a configuration $C \in \mathcal{C}_{sp}$, the gathering is eventually achieved. 
\end{lemma}

\begin{proof}
Holds from Lemmas \ref{lem:Csp-3--Csp-4}, \ref{lem:Csp-1toCsp-2}, \ref{lem:Csp4-Csp-3} and \ref{lem:gethering} (refer to Figure \ref{fig:spTransition}). 
\end{proof}

 \begin{figure}[htb]
    \begin{center}
        \includegraphics[scale=0.46]{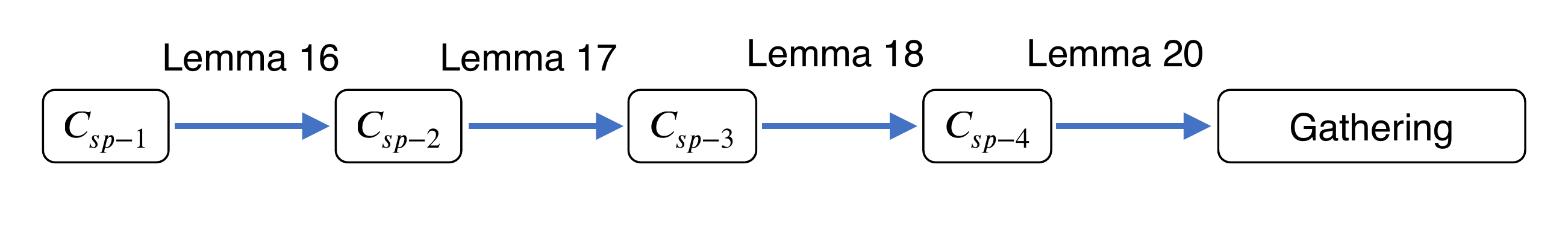}
        \caption{Transitions within the set $\mathcal{C}_{sp}$}\label{fig:spTransition}
    \end{center}
\end{figure}


\begin{theorem} \label{theo:phase2}
Starting from a configuration $C \in \mathcal{C}_{p_2}$, the gathering is eventually achieved. 
\end{theorem}

\begin{proof}
Can be deduced from Lemmas \ref{lem:prtols}, \ref{lem:Cls5-to-Csp}, \ref{lem:Cls-to-Csp2}, \ref{lem:Csp-1toCsp-2}, \ref{lem:Csp4-Csp-3}, \ref{lem:Csp-3--Csp-4} and \ref{lem:gathering}.
\end{proof}


\begin{theorem}
Assuming an $(\ell, L)$-torus in which $L <\ell$ and $L>4$ and starting from an arbitrary rigid configuration, Protocol 1 solves the gathering problem for any $k \geq 3$. 
\end{theorem}

\section{Concluding remarks}

We presented the first algorithm for gathering oblivious mobile robots in a fully asynchronous execution model in a torus-shaped space graph. Our work raises several interesting open questions:
\begin{enumerate}
\item What is the exact set of initial configurations that are gatherable? Our work considers initial rigid configurations only, and we know that periodic and edge-symmetric configurations make the problem impossible to solve. As in the case of the ring, there may exist special classes of symmetric configuration that are still gatherable.
\item The case of a square torus is intriguing: the robots would loose the ability to distinguish between the big side and the small side of the torus, so additional constraints are likely to hold if gathering remains feasible.
\item Following recent work by Kamei et al.~\cite{KLOTW19c} on the ring, it would be interesting to consider myopic (\emph{i.e.} robot whose visibility radius is limited) yet luminous (\emph{i.e.} robots that maintain a constant size state that can be communicated to other robots in the visibility range) robots in a torus. 
\end{enumerate}

\bibliography{biblio}
\bibliographystyle{plain}

\end{document}